\documentclass[12pt]{article}
\usepackage[latin9]{inputenc}
\usepackage{rotfloat}
\usepackage{units}
\usepackage{amsthm}
\usepackage{palatino}
\usepackage{graphicx}
\usepackage{geometry}
\usepackage{soul,color}
\usepackage{enumerate}

\usepackage{mathtools}
\usepackage{pstricks}
\usepackage{paralist}

\usepackage[T1]{fontenc}
\usepackage{amsmath, mathtools}
\usepackage{amssymb,paralist}
\usepackage{graphicx}

\usepackage{setspace}
\usepackage{esint}
\usepackage[authoryear]{natbib}

\usepackage{hyperref}
\usepackage{cleveref}

\crefname{subsection}{section}{sections}
\crefname{subsubsection}{section}{sections}
\crefname{prop}{proposition}{propositions}

\makeatletter

  \theoremstyle{remark}
  
  \theoremstyle{plain}
  \newtheorem{lem}{\protect\lemmaname}
\theoremstyle{plain}
\newtheorem{thm}{\protect\theoremname}
  \theoremstyle{plain}
  \newtheorem{prop}{\protect\propositionname}
 \theoremstyle{definition}
  \newtheorem{example}{\protect\examplename}

\theoremstyle{plain}
\newtheorem{theorem}{Theorem}
\newtheorem{lemma}{Lemma}
\newtheorem{corollary}{Corollary}
\theoremstyle{definition}
\newtheorem{definition}{Definition}
\newtheorem{remark}{Remark}
\newtheorem{claim}{Claim}

\DeclareMathOperator*{\E}{\mathbb{E}}

\usepackage[small,hang]{caption}\usepackage{latexsym}\usepackage{amsfonts}\usepackage{amsthm}\usepackage{bm}\usepackage{mathrsfs}\usepackage{amscd}\usepackage{alltt}\usepackage{eucal}\usepackage{url}\usepackage{colortbl}\usepackage{fancybox}\usepackage{array}\usepackage{rotating}\usepackage{multirow}\usepackage{longtable}\usepackage{multicol}\usepackage{enumitem}\usepackage{xcolor}\usepackage{eurosym}\usepackage{booktabs}\usepackage{multimedia}\usepackage{epigraph}\usepackage{bigints}\usepackage{soul}\usepackage{tabularx}\usepackage{mathrsfs}\usepackage{rotating}\usepackage{lscape}\usepackage{array}

\newcolumntype{C}{>{\centering\arraybackslash}m{.8in}}

\makeatother

  \providecommand{\examplename}{Example}
  \providecommand{\lemmaname}{Lemma}
  \providecommand{\propositionname}{Proposition}
  \providecommand{\remarkname}{Remark}
\providecommand{\theoremname}{Theorem}

\begin{document}
\global\long\def\sign{\textrm{sign}}

\newtheorem{proposition}{Proposition}

\title{Descending Price \\Optimally Coordinates Search
\thanks{We are grateful to conversations with Matt Gentzkow and Larry Samuelson for inspiring us to work on this project and to Nick Arnosti, Brendan Lucier, Bruno Strulovici, Vasilis Syrgkanis, Sam Taggart, Juuso V\"alim\"aki and workshop participants at Boston College, the Harvard Center for Research on Computation and Society and Microsoft Research New England  for useful comments.  Special thanks to Vivek Bhattacharya Jimmy Roberts and Andrew Sweeting for sharing their code for and assisting directly with our Timber calibration.%
}}

\author{Robert Kleinberg%
\thanks{Department of Computer Science, Cornell University: 402 Gates Hall, Ithaca, NY 14853: robert.kleinberg@cornell.edu, http://www.cs.cornell.edu/\textasciitilde rdk/. This work was done while the author was a researcher at Microsoft
Research New England.
} \and Bo Waggoner\thanks{Warren Center for Network \& Data Sciences, 220 South 33rd Street, Philadelphia, PA 19104: bwag@seas.upenn.edu, http://www.bowaggoner.com.} \and E.~Glen Weyl%
\thanks{Microsoft Research New York City, 641 Avenue of the Americas, New York, NY 10011 and Department of Economics, Yale University: glenweyl@microsoft.com,
http://www.glenweyl.com.%
}}

\date{December 2016}

\maketitle
\thispagestyle{empty}

%\include{Preamble}
%\begin{document}

\begin{abstract}
Investigating potential purchases is often a substantial investment under uncertainty.  Standard market designs, such as simultaneous or English auctions, compound this with uncertainty about the price a bidder will have to pay in order to win.  As a result they tend to confuse the process of search both by leading to wasteful information acquisition on goods that have already found a good purchaser and by discouraging needed investigations of objects, potentially eliminating all gains from trade.   In contrast, we show that the Dutch auction preserves all of its properties from a standard setting without information costs because it guarantees, at the time of information acquisition, a price at which the good can be purchased.  Calibrations to start-up acquisition and timber auctions suggest  that in practice the social losses through poor search coordination in standard formats are an order of magnitude or two larger than the (negligible) inefficiencies arising from ex-ante bidder asymmetries.
\end{abstract}

%\emph{Keywords}: search, auctions, market design, corporate control, economics job market, information acquisition

%\emph{JEL codes}: C61, D44, D47, D82, D83

%\end{document}

\onehalfspacing

\section{Introduction}\label{intro}
\setcounter{page}{0}

 \citet{vickrey} famously argued that ``the common or progressive type of auction...provide(s) better chances for optimal allocation than the regressive or `Dutch' auction.''  Yet an important weakness of common ascending or simultaneous auctions is that a bidder faces significant uncertainty about both the price she faces and her chances of winning at any point in the auction process.  By contrast in a Dutch auction at any time the good is still available a participant knows that there is a maximum price for which she can claim the good.  In this paper we show that this property implies an important efficiency benefit of the Dutch format when individuals face costs of information acquisition, as it allows for rational planning about whether to make such an investment in a way that is impossible under standard formats.  As a result, standard formats may eliminate all gains from trade when costly information acquisition is necessary, while a Dutch format will always perform as well with information costs as without them.  In calibrations to start-up acquisition and US timber auctions we find that standard formats lose a few percent of potential welfare from poor search coordination. However even these small losses dwarf those  highlighted by \citeauthor{vickrey} from ex-ante bidder asymmetries in the Dutch auction, which we find to be negligible.

Information acquisition costs play a large role in a variety of auction, assignment and matching markets.  Press reports indicate that due diligence costs for the acquisition by a large technology company of a start-up are typically 20-40\% of the size of a deal and \citet{salz} finds through empirical estimation that search costs represent a similar fraction of the total costs of waste haulage contracts.  Any participant on either side of the academic job market knows how costly information acquisition is in that context.  Conserving on such information costs is thus both an important component of efficiency in these markets and is necessary to ensure an efficient allocation, as it will be inefficient for individuals to expend these costs unless they offer sufficient individual or social benefit.

To ensure such investments are not wasted it is important for a given bidder to be confident there is not already some other bidder who greatly values the object and therefore will end up being assigned it regardless of the outcome of that bidder's investigation.  To analyze this risk, it is useful to adopt the framework of \citet{investment}, who show that any irreversible investment under uncertainty may be viewed as a ``real option''.  In particular, for any bidder there will be a ``strike price'' such that, if she knows she is able to buy the good at exactly this strike price, she will find it just in her interest to investigate the object.  Any time the bidder is guaranteed to be able to buy the object weakly below this price, the bidder will be able to eliminate all risk by selling a call option at this strike price to an external party to fund her information acquisition and committing to give to this external party any profits above this price.  Any strategy that does not ``exercise this option in the money'' will lead to lower utility than this ``covered call position'' as it will expose the bidder to the risk of having wasted the cost of information acquisition.

In \Cref{prelim}, we formally develop this novel interpretation of \citet{weitzman}'s characterization of optimal information acquisition, combining ideas from \citet{weber} and \citet{olszewski-weber}.  This analysis immediately suggests that standard market designs may perform poorly.  In particular, simultaneous auctions of any form, ascending price auctions and a sequential bargaining procedure proposed by \citet{BK} and \citet{sweeting} all  leave bidders exposed, at the time of information acquisition, to the risk that the effective price they will face for purchase may be well above their strike price.  In fact we show in \Cref{flaws} that this exposure may destroy all of the gains from trade in some examples for each of these mechanisms.

Of course efficient procedures in this context do exist based on dynamic versions of the \citeauthor{vickrey}-\citeauthor{clarke}-\citeauthor{groves} procedure \citep{dynamicpivot, atheysegal}.  However, as  \citet{parkesVCG} and \citet{cremer} detail, this mechanism either depends on the designer having detailed information about bidders' value distributions or is {\em extremely} elaborate as it requires every individual to communicate her full private information about the distribution of her values to a central agent in order to determine the externalities each bidder creates on others.  Furthermore, even if the procedure could be simplified, it would suffer critiques analogous to those of \citet{levins} for other dynamic \citeauthor{vickrey} auctions, that bidders have very weak incentives to provide truthful reports about much information and thus may instead predate or collude with their rivals. It thus seems unlikely that such a procedure would be used in practice, especially when  combined with other practical issues facing Vickrey auctions \citep{lovelybutlonely}.

Luckily, though, these fully efficient procedures are not the only ones that protect bidders from exposure and thus permit a rational search procedure.  In fact many commonly-used, ad-hoc procedures in a variety of contexts provide participants in markets with guarantees that protect their investments in information acquisition.  For example, in the academic job market participants are made firm offers before they are asked to invest in determining which school they prefer and in college admission early decision offers an opportunity for students to commit to schools.

Our principal result, which we state and prove in \Cref{main}, is that the Dutch auction that \citeauthor{vickrey} criticized has the benefit of preventing exposure.  Intuitively, a bidder never wishes to investigate her value until the price has dropped to her strike price and by this time she knows she can secure the object at or below her strike price.  She thus is able to fully unload the risk by selling a call option.  Any equilibrium of the Dutch auction is thus welfare-equivalent to an equilibrium of an auction without information acquisition where each individual's value is replaced by the value of her covered call position.  This implies that all the desirable properties of the Dutch auction and its cousin the first-price auction (efficiency when bidders are symmetric, approximate efficiency otherwise, etc.) when there are no information costs  carry over to our setting.\footnote{Note, however, that the first-price auction, unlike the Dutch auction, does not achieve efficiency in our context because it is simultaneous.  This divergence and the tendency of the Dutch auction to encourage more information acquisition and revenue were first observed by \citet{during}.}

Of course, none of this undermines \citeauthor{vickrey}'s concerns that, when individuals are asymmetric, the Dutch auction fails to achieve full efficiency because stronger bidders will tend to shade their bids down by more than weaker bidders.  Therefore in \Cref{calibrations} we consider the magnitudes of the benefits and costs of the Dutch auction compared to alternatives in two quantitative models.  The first is calibrated to match press discussions of the start-up acquisition market.  The second matches the empirical moments recovered by \citet{sweeting} in the US Timber auctions.

In these cases we find that bidder asymmetries have only a negligible impact on the efficiency of the Dutch auction.  In all but the most extreme scenarios, designed deliberately to break this result, the efficiency loss due to asymmetries is less than a percentage point and it is typically an order of magnitude or two smaller than this.  On the other hand the inefficient search coordination created by standard formats typically leads to losses of several percentage points of welfare in typical scenarios and tens of percentage points in some extreme cases.  This shows that in  some practical cases \citeauthor{vickrey}'s reservations about the Dutch procedure may be negligible while the search benefits of the Dutch procedure are significant, though much smaller than suggested by our worst case examples.  Results on revenue, which are less central to our analysis, favor the Dutch procedure much more strongly.

We conduct most of our analysis under a series of stylized assumptions that simplify our exposition and sharpen our message.  In particular we assume that agents play an equilibrium, that there is only a single object for sale, that there is only a single stage of information acquisition, that investigating an object is necessary to claim it and that values are private.  In \Cref{extensions} we show that our core message about the value of descending price auctions over standard formats is robust against substantial (non-nested) relaxations of these assumptions.  We also discuss why we believe descending prices are preferable to a fixed posted price, which can also avoid the exposure problem.

Our paper focuses analytically on the relatively narrow problem of auctions; that is environments where goods have to be assigned to individuals who can freely transfer utility in the form of currency between them.  However in Section \ref{matching} we use the connections drawn in recent years by the ``market design'' literature between auctions and mechanisms for matching to suggest how our insights about the value of descending prices might extend to those settings.  However, such broader implications for market design require significant additional research, especially related to issues of timing (viz. in practice inspections are not instantaneous) that our analysis largely abstracts from.  We highlight the additional analysis we think would be most useful for reaching concrete design insights in our conclusion, \Cref{conclusion}.  Important details and that do not fit in the text appear in appendices following it, as do all proofs not appearing inline and details of our numerical methods.

\newcommand{\allocsub}[1]{\mathbb{A}_{#1}}
\newcommand{\allocsubi}{\allocsub{i}}
\newcommand{\allocsubij}{\allocsub{ij}}
\newcommand{\inspect}[1]{\mathbb{I}_{#1}}
\newcommand{\inspecti}{\inspect{i}}
\newcommand{\inspectij}{\inspect{ij}}

\newcommand{\Inspect}{\inspectij}

\newcommand{\littleo}{\mathit{o}}

\section{Preliminaries}\label{prelim}
In this section we describe the baseline model we employ until \Cref{extensions} and use this to develop the fundamental ideas behind our analysis.

\subsection{Model}\label{model}
A single item is for sale to one of $n$ bidders.
Bidders initially do not know their values, but only some information about them in the form of a type, which also determines a cost that must be paid in order to inspect and discover the value.
Formally, each bidder $i$ initially draws a private type $\theta_i \in \Theta_i$.
% and a private cost $c_i \in \mathbb{R}$.
These are distributed according to a common prior $\mathcal{P}$ on
% $\Theta_1 \times \mathbb{R} \times \cdots \Theta_n \times \mathbb{R}$.
$\Theta_1 \times \cdots \times \Theta_n$.

Each bidder $i$ has a family of value distributions $\{ F_{\theta_i} : \theta_i \in \Theta_i\}$
and a cost function $c_i : \Theta_i \to \mathbb{R}_+$.
When $i$'s type is $\theta_i$, $i$'s value $v_i \in \mathbb{R}$ for the item is drawn
according to $F_{\theta_i}$.
These values $v_1,\ldots,v_n$ are independent conditional on the type profile
$(\theta_1,\ldots,\theta_n)$. However, they are not initially observable.
At any time, any bidder $i$ may pay a cost $c_i(\theta_i)$ to instantaneously \emph{inspect} the item, learning her value $v_i$.  Inspections are private and are not observed by other bidders.
We will frequently abuse notation and write $c_i$ in place of $c_i(\theta_i)$, when the
dependence on $\theta_i$ is understood from context.

We use $\inspecti$ as an indicator for individual $i$ inspecting the object, $\allocsubi$ as an indicator for the object being allocated to $i$ and let $(Z)^+ := \max\{Z, 0\}$.  Following \citet{weitzman}, we assume that if $i$ obtains the item, $i$ must pay the inspection cost, \emph{viz.} $\allocsubi = 1 \implies \inspecti = 1$.   We denote the expectation under the measure $F_{\theta_i}$ by $\E_{\theta_i}$.

 All bidders have quasilinear utility, \emph{viz.} $i$'s utility is
 \[ u_i = \allocsubi v_i - \inspecti c_i - t_i \]
where $t_i$ is the net payment made by $i$. We assume bidders are risk-neutral expected value maximizers.

%%%%%%%%%%%%%%%%%%%%%%%%%%%%%%%%%%%%%%%%%%%%%%%%%%%%
\subsection{Real option}\label{option}

Before considering the competitive interaction among bidders, we begin by characterizing the structure of each bidder's individual problem in terms of real options, an analogy that simplifies the exposition of our results.  While this problem is a special case of \citet{weitzman}'s ``Pandora's box'' problem, our characterization applies only in the special case we focus on and to our knowledge our characterization, and particular the analogy to real options, is novel.\footnote{In particular, \citeauthor{weitzman} allows for investigation to take place over time in an explicitly dynamic model with discounting, while we assume investigation is instantaneous and thus have no explicit dynamics beyond the sequencing of actions.}

The decision by a bidder to invest in learning her value is irreversible and has an uncertain payoff.  As such it is an ``investment under uncertainty'' in the sense of \citet{investment}.  It may thus be viewed as a real option.  It will be worth  exercising this option as long as the (opportunity) price  that would have to be paid for the object is sufficiently low. A natural way to quantify this is to define the highest price such that it would still be advantageous to investigate the object if the price were known with certainty.  In particular, the bidder could finance her full inspection cost by selling a ``call option'' on the value by guaranteeing that if, upon investigation, she finds that the value of the object is above this level she will give any profits above this ``strike price'' to the owner of the option.

\begin{definition}\label{strike} The \emph{strike price} of bidder $i$ given $\theta_i$ is denoted $\sigma_i$ and is the unique solution to
 $$ {\E}_{\theta_i}\left[ (v_i - \sigma_i)^+ \right] = c_i . $$
\end{definition}

This strike price is a special case of the (non-nested) indices proposed by \citeauthor{weitzman} and \citet{gittins} for more general optimal information acquisition problems.  Another critical quantity often used in the broader analysis of information acquisition problems, related to the ``prevailing charge'' introduced by  \citet{weber} in his elegant proof of the Gittins Index Theorem, can be interpreted in terms of option pricing in our setting.  In particular, if a bidder does decide to sell a call option on the object at its strike price and use the proceeds to fund her inspection, she is left holding a position with no cost of information acquisition, but a limited upside potential, with the rest of the upside potential being held by the owner of the call option.  This position is known in options trading as a ``covered call position''.  The value of this position is the ``covered call value''.

\begin{definition} The \emph{covered call value} of bidder $i$ given realizations of $\theta_i$ and $v_i$ is $\mbox{$\kappa_i := \min\{\sigma_i, v_i\}$.}$
\end{definition}

In order to sell a call option on the object at the bidder's strike price, the bidder must {\em commit} that any time the value ends up strictly above this strike price (if the option is ``in the money'') she will allow the option holder to exercise her option at the strike price.  Only then will the purchaser be willing to pay the full inspection cost and thus the bidder will be able to avoid all risk from inspection, instead simply earning the covered call value.  The following lemma proves that pursuing such a commitment can only benefit the bidder; any other policy will cause her to earn weakly less than the covered call value in expectation and strictly less if she has positive probability of failing to exercise in the money.

\begin{lem} \label{lemma:upper-bound-policy}
Given any $\{\theta_i\}$, for any procedure and any $i$, $\E\left[\allocsubi v_i - \inspecti c_i \right] \leq \E \allocsubi \kappa_i$. Furthermore, this holds with equality if and only if bidder $i$ always ``exercises in the money'', \emph{viz.} takes the item if she inspects and $v_i > \sigma_i$.
\end{lem}

\begin{proof}
Fix all $\{\theta_i\}$.
Using the definition of the option strike price, substitute for $c_i$ and use the independence of $v_i$ and $\inspecti$:
\begin{equation} \label{eq:lubp}
 \E \left[ \allocsubi v_i - \inspecti c_i \right] = \E\left[ \allocsubi v_i - \inspecti (v_i - \sigma_i)^+ \right]  \leq \E\left[ \allocsubi \left( v_i - (v_i - \sigma_i)^+ \right) \right]  = \E \left[ \allocsubi \kappa_i \right] .
\end{equation}
The inequality follows because $\allocsubi \leq \inspecti$ (one must inspect in order to be allocated).
Furthermore, by subtracting the left and right sides of the inequality occurring in the middle of
line~\eqref{eq:lubp}, we see that the two sides are equal if and only if
$\E \left[ (\inspecti - \allocsubi) (v_i - \sigma_i)^+ \right] = 0$, which
happens if and only if there is zero probability that $\inspecti = 1, \allocsubi = 0,$
and $v_i > \sigma_i$. This is precisely what it means to say that a procedure always
exercises in the money.
\end{proof}

Thus any deviation from selling the option and committing to exercise in the money whenever strictly profitable will cause strict losses in utility relative to the covered call value.  On the other hand the best welfare that any central planner  can possibly hope to achieve, even if that planner has access to all ex-ante (non-costly) information and can force obedience by bidders, is the highest covered call value among bidders.

\begin{corollary} \label{cor:upper-bound-opt}
The welfare of the optimal centralized procedure is at most $\E\left[ \max_i \kappa_i\right]$.
\end{corollary}

\begin{proof}
The welfare generated by bidder $i$ is exactly $\allocsubi v_i - \inspecti c_i$, so total welfare is
\begin{align*}
 \E \left[ \sum_i \allocsubi v_i - \inspecti c_i \right] \leq \E \left[ \sum_i \allocsubi \kappa_i \right] \leq \E \left[ \max_i \kappa_i\right] .
\end{align*}
\end{proof}

In fact, this welfare is achievable by a simple procedure, due to \citeauthor{weitzman}.

\begin{thm} \label{thm:opt-is-descending}
The first-best procedure for a planner who knows all $\{\theta_i\}$s causes bidders to inspect in order of decreasing $\sigma_i$, stopping when the largest observed value $v_i$ exceeds all remaining $\sigma_{-i}$ and assigning the item to $i$.
\end{thm}

\begin{proof}[Proof of Theorem \ref{thm:opt-is-descending}]
The descending-inspection procedure always exercises in the money, implying by Lemma \ref{lemma:upper-bound-policy} that its welfare is equal to $\E \left[\sum_i \allocsubi \kappa_i\right]$.
But it also always allocates to the bidder with the highest covered call value, meaning its welfare is $\E \left[ \max_i \kappa_i \right]$.
This is optimal by Corollary \ref{cor:upper-bound-opt}.
\end{proof}

In the next section we will use the fact that other procedures will tend to make consistently exercising in the money impossible to produce examples where these procedures destroy all gains from trade.

%%%%%%%%%%%%%%%%%%%%%%%%%%%%%%%%%%%%%%%%%%%%%%%%%%%%

%%%%%%%%%%%%%%%%%%%%%%%%%%%%%%%%%%%%%%%%%%%%%%%%%%%%%%%%%%%%%%%%%
\section{Standard Auctions May Hopelessly Confuse Search}\label{negative}

In this section we use show that standard simple auctions may be unboundedly inefficient in this context and  discuss why more elaborate efficient procedures are unlikely to be practically relevant.

\subsection{Failure of standard procedures}\label{flaws}

The most widely studied auction formats in models with information acquisition \citep{VCGinfo}  are simultaneous in the sense that all individuals must decide simultaneously and prior to communication with one another whether to inspect.  However simultaneous procedures make exercising in the money impossible because  individuals decide whether to inspect the object prior to any assurances about price.  The following example shows a case when this eliminates all gains from trade in the sense of \citet{mailathpostelwaite}: the best possible simultaneous mechanism  achieves an unboundedly small fraction of the first-best welfare.

\begin{example}[Sequencing and the inefficiency of simultaneous procedures]\label{simultaneous}
Every individual value  is drawn iid as either $M>1$ with probability $\frac{1}{M}$ or $0$ with probability $1-\frac{1}{M}$ and for all $i, c_i=c$ for a constant $0 < c < 1$ to be determined.  When $M$ is large, a success is a ``black swan'' \citep{taleb}.  Consider the limit as $M\rightarrow \infty$ and $\frac{n}{M}\rightarrow \infty$ so the value of a black swan is large but rare, and yet there are enough opportunities to find one that the probability of it existing in the entire population approaches unity.

The optimal policy then involves ordering individuals in any manner and having them sequentially query their values until a black swan is discovered and then the assignment made.
In this limit, the gross utility achieved approaches $M$ (a black swan is always found) while the expected number of inspections is $M$ (the mean of the geometric distribution with probability $1/M$), so the net welfare approaches $M(1-c)$.

 \citeauthor{VCGinfo} prove that, among all simultaneous mechanisms, a simple second-price auction leads to greatest efficiency in this setting.
Each bidder has only two strategies that weakly dominate over all others: either bidding expected value less inspection cost without inspecting; or paying the inspection cost and entering a bid equal to value. There is a fatal tradeoff: If too many inspect, then the large total inspection cost swamps all gains from efficient allocation.
But if too few do, then welfare suffers because there is a small chance that any bidder will find a black swan.
We need to offer bidders the opportunity to inspect only if no black swan has yet been found, but simultaneity makes this impossible.
\end{example}

Intuitively bidders face a high risk they will not be able to exercise in the money and thus their greatest possible payoff is greatly reduced.  We formalize this in Proposition \ref{prop:simultaneous-flaw}. In the interests of space, all proofs in this subsection have been moved to Appendix \ref{lower}.

\begin{prop} \label{prop:simultaneous-flaw}
In the example above  for $c$ appropriately close to $1$, the welfare of any set of bidder strategies is an unboundedly small fraction of the optimal welfare.
\end{prop}

While we believe this example is the first of unbounded inefficiency of simultaneous search, its weaknesses in a search context are well-appreciated. In particular, \citet{jehiel} show that in a model that is a special case of ours that (trembling-hand perfect equilibria of) English auctions always perform weakly better than simultaneous second-price auctions and sometimes strictly better.   Indeed, our black-swan example has an equilibrium with optimal welfare where at the beginning of the auction, bidders sequentially inspect in an arbitrary (but fixed, common-knowledge) order; if a bidder finds the high value, they stay in and all others drop.  Unfortunately, such auctions still make it impossible to ensure exercise in the money in equilibrium, because at the point when a bidder may have to inspect she still has no way to place an upper bound on the price for which she may be able to acquire the object.

\begin{example}[English auctions do not allow exercise in the money]\label{ascdending}
Consider a slight modification of the previous example. With a medium chance (say, $50\%$), the bidder's value is neither $0$ nor extremely high, but is instead some small yet significant value $L$.
For good welfare, it is still necessary that a black swan is almost always found.  But now, bidders cannot credibly signal that they have found the black swan until the price is already too large (higher than $L$). If a bidder inspects at an early time and decides not to drop out, this most likely only indicates that their value is $L$.  Therefore, it is not until the price ascends above $L$ that a sequential, high-welfare strategy set has a hope of succeeding.  However, any bidder who is left to inspect last (or subset of bidders who simultaneously inspect last) has a negative expected utility, as the price now exceeds their expected gain.
\end{example}

\begin{prop} \label{prop:ascending-flaw}
 There are instances where any equilibrium of the English auction has arbitrarily low welfare as a fraction of the first-best.
\end{prop}

Again, to our knowledge, this is the first example in the literature of unbounded inefficiency of an English auction based on information acquisition.  However, again the flaws of such procedures in related settings were pointed out by \citet{BK}.  They, and \citet{sweeting}, argue with numerical examples that an explicitly sequential  bargaining procedure may produce  higher welfare.  In the Bulow-Klemperer-Roberts-Sweeting (BKRS henceforth) procedure, potential buyers are approached sequentially.  There is (potentially) a reserve price and a buyer may become an incumbent by placing any bid above this reserve.  If the buyer chooses not to become the incumbent, the next buyer is offered that chance.  Once a buyer becomes the incumbent, the next buyer is approached and may become the incumbent if she beats the incumbent in a ``knockout'' round in which either active buyer can submit any bid higher than the current high bid.  This proceeds until one buyer drops out and a new incumbent is thus established; the dropping out buyer can never re-enter.  An incumbent can, after becoming such, raise her bid if she wishes.

This procedure always leads to efficiency in the examples we considered above, as they allow for appropriate sequencing.  However this is only the case because all individuals are symmetric in these examples prior to inspection.  If individuals have private information, because this bargaining procedure does not allow for the private information to be revealed, again sequencing may fail.  In particular, the threat of an individual who has already entered then bidding the price up undermines the ability of a future bidder to exercise in the money in the same way this may occur in a standard English auction.

\begin{example}[Random sequencing may be devastating]
Suppose that potential buyers are, unobservably, either risky ($R$) or safe ($S$).  Any individual is risky with probability $\pi \ll 1$ and safe with probability $1-\pi$.  Safe individuals have value $v_S$ with probability $1-p_S$ and value $-\infty$ with probability $p_S\ll 1$.\footnote{If our assumption that the individual must query her value in order to take the object is enforced, then it is equivalent for safe to have value $v_S$ with probability one. We introduce the extra punishment for a bad match just to make the example more plausible.}  The cost of determining her value is $c_S=v_s-\epsilon$, where $\epsilon>r$, the reserve price set by the seller.  Risky individual have probability $p_R\ll1$ of having value $v_R$ and probability $1-p_R$ of having value $0$.  The cost of learning her value is $c_R=\frac{v_R p_R}{2}>0$.  We assume  $\epsilon \ll \frac{v_R}{2}<v_S$.

In this example, unbounded inefficiency occurs because it is necessary for many risky individuals to investigate their values to achieve the optimum.  However it is more likely that at least one $S$ individual first is offered a chance and enters the auction.  This is enough to deter the $R$ types from investigating, as they fear a bidding war with an already-entered $S$ type will block their ability to make their investigation profitable ex-ante.
\end{example}

\begin{prop} \label{prop:sequential-flaw}
 There are instances where any equilibrium of the BKRS mechanism has unboundedly low welfare as a fraction of the first-best.
\end{prop}

Interestingly in this case the simultaneous or English auction  performs well.\footnote{In both of those mechanisms there is a mixed-strategy symmetric equilibrium where (for large $n$) the total probability of a black swan being found is roughly $\frac{1}{2}$.  In such an equilibrium only $R$ types investigate their values.  Because the expected surplus of an $R$ type from investigating and winning with probability one is $\frac{v_R}{2}-r$, $R$ types are still willing to investigate for probabilities of facing competition close to $\frac{1}{2}$.}  Thus while BKRS's argument that the sequential procedure is more efficient than the simultaneous one is true in some examples, the reverse holds in others.

%\footnote{In Subsection \ref{posted}, we show that there does exist a sequential auction procedure that attains approximately optimal welfare by using details of the bidders' distributions to choose a high reserve price in the case when individual types are drawn independently.  However, in addition to depending on details of the value distribution that our approach does not rely on, this procedure has a number of other weaknesses.  For example, even when values are drawn symmetrically it will typically achieve only a relatively coarse approximation to optimal welfare, while our procedure achieves the first-best in equilibrium in this case.  We thus put significantly less emphasis on sequential price posting than on the descending price procedure.}

\subsection{Impracticality of optimal mechanism}\label{VCGfail}

We must therefore wonder if more sophisticated auctions may outperform these simple formats. If bidder types (costs and value distributions) are known to the auctioneer and information acquisition is contractible, \citet{cremer} describe a procedure where agents are, in sequence, called upon to inspect and then bid in a second-price auction against a reservation price equal to the highest remaining strike price.  Obviously this relies heavily on the informational and observability assumptions, which are violated in most cases we are interested in.  The only mechanisms that are known to achieve efficiency more broadly are dynamic versions of the \citeauthor{vickrey}-\citeauthor{clarke}-\citeauthor{groves} (henceforth VCG) mechanism described in a special case of our environment by \citet{parkesVCG} and more generally by \citet{dynamicpivot} and \citet{atheysegal}.

This mechanism has all individuals report their costs and value distributions into a central mechanism designer who tells them to inspect in appropriate order and asks them to report  their values upon learning them.  The planner implements the \citeauthor{weitzman} procedure and every individual pays the expected amount, conditional on all information realized at the point of their report, by which their participation in the mechanism reduced the aggregate expected utility of all other bidders.  While efficient in the weakly-dominant truth-telling equilibrium, this mechanism seems unlikely to us to be used in practice for two primary reasons that motivate our interest in the simpler and (see below) approximately efficient Dutch auction.

 First, it requires participants to report the distribution of their value, an object that is typically complicated and psychologically difficult to communicate reliably \citep{elicitation} making such a protocol impractical.\footnote{With multiple objects there is no known computationally efficient algorithm even for calculating an appropriate externality payment, and in the case of multiple inspection stages the cognitive and communication complexity become extremely daunting.} Second, as \citet{levins} show, dynamic VCG mechanisms have many equilibria, nearly all of which are inefficient, because reports made by bidders may affect only the price paid by rivals and not their own allocation.  This problem is particularly severe in our setting: Only the strike price and value reported impacts a bidders' own allocation, but the costs and distributions she reports have a rich impact on others' payments.  Individuals have no incentive to correctly report this ancillary information and may manipulate it to raise or lower payments made by their rivals depending on whether they are predating or colluding with their rivals.  It seems unlikely that such other-regarding motives will be entirely absent in any context and thus this seems a setting where the use of a dynamic VCG mechanism is particularly unpalatable.\footnote{To highlight the centrality of this last point, we briefly now describe a mechanism with substantially lower communication requirements that can implement the optimum, but which even more clearly highlights the issues raised by \citeauthor{levins}: Every bidder reports her strike price and a single random sample from her value distribution.  Bidders are  sequentially approached in descending order of strike price and asked to report their value until a value is reached that exceeds all remaining strike prices.  At this point the highest  bidder with a value exceeding all remaining strike prices is given the object.  Individuals are then required to pay the utility that their participation would have denied to all other individuals on net if they had not participated and every bidder had, counter-factually, turned out to have a value given by the random sample reported in stage 1.
 
This mechanism is efficient, all necessary information is reported and there is a weakly dominant truthful equilibrium because the payments are computed based on a sample path through the optimal mechanism.  Thus expected payments are equal to the expected externalities that the dynamic VCG mechanism computes.  However, note that the random samples reported play  no role in determining any individual's allocation; they determine only payments for others.  There is thus no incentive for individuals to report these truthfully; even a tiny incentive to collude or predate will lead to extreme reports.  While this illustrates the point sharply, the standard dynamic VCG implementation  generically asks individuals to report dimensions of the value distribution that leave the strike price unchanged but  may have a dramatic impact on other bidders' payments.}

\section{Dutch Auctions' Performance is Invariant to Search}\label{main}
\newcommand{\ccpt}[1]{{#1}^{\circ}}

In this section, we show that the performance of the Dutch auction in our environment is ``equivalent'' to its performance in the ``corresponding'' environment without information acquisition, namely one in which each bidder has a value equal to her covered call value that she knows without paying any inspection costs.  Given that, by \Cref{lemma:upper-bound-policy}, no mechanism can achieve higher welfare than the expectation of the highest covered call value, this equivalence connection preserves welfare and thus implies that any desirable welfare property of the Dutch auction in the context without information acquisition carries over to our setting.

\begin{definition}\label{equivdist}
For every type $\theta_i \in \Theta_i$ we define the {\em covered counterpart}
$\ccpt{\theta}_i$ to be a type with zero inspection cost, whose value distribution
is identical to the distribution of $\kappa_i$ when bidder $i$'s type is $\theta_i$.
More formally, we enlarge each bidder $i$'s type space to $\Theta_i \cup \ccpt{\Theta}_i$,
where $\ccpt{\Theta}_i$ is a disjoint copy of $\Theta_i$ containing an element
$\ccpt{\theta}_i$ for every $\theta_i \in \Theta_i$. The cost function
$c_i$ is extended by setting $c_i(\ccpt{\theta}_i) = 0$ for every
$\ccpt{\theta}_i \in \ccpt{\Theta}_i$. The family of value distributions
$\{F_{\theta}\}$ is extended by setting
\[
  F_{\ccpt{\theta}_i}(x) = \begin{cases}
    F_{\theta_i}(x) & \mbox{if $x \leq \sigma_i$} \\
    1 & \mbox{if $x > \sigma_i$}. \end{cases}
\]
For any prior $\mathcal P$ define the {\em induced prior over covered counterparts}
$\ccpt{\mathcal P}$ to be the joint distribution of
$(\ccpt{\theta}_1,\ldots,\ccpt{\theta}_n)$ when $(\theta_1,\ldots,\theta_n)$ is
distributed according to $\mathcal{P}$.
\end{definition}

Note that when bidders have covered counterpart types
$\ccpt{\theta}_1,\ldots,\ccpt{\theta}_n$, they know their own
values $v_1,\ldots,v_n$ without having to pay inspection costs.
Therefore, an auction whose bidders have types $\ccpt{\theta}_1,\ldots,\ccpt{\theta}_n$
is equivalent to a standard private-value auction without inspection costs,
where each bidder $i$ in addition to knowing her value $v_i$ receives a signal $\ccpt{\theta}_i$
that is not payoff-relevant except to the extent that it correlates with
other bidders' types and hence with their bids. In particular, when
$\theta_1,\ldots,\theta_n$ are mutually independent under distribution $\mathcal P$,
then $\ccpt{\theta}_1,\ldots,\ccpt{\theta}_n$ are mutually independent under
distribution $\ccpt{\mathcal P}$ and the signal $\ccpt{\theta}_i$ becomes simply
a random signal containing no payoff-relevant information whatsoever. In this
special case, the equilibria of any auction with types distributed according to
$\ccpt{\mathcal P}$ are exactly equivalent to the equilibria of the same auction with
independent private values $\kappa_1,\ldots,\kappa_n$.

%For a bidder $i$ with type $\theta_i$, we define the {\em covered counterpart}
%of $i$ to be a bidder with type $\theta_i^\circ$
%For any prior $\mathcal P$ let the {\em induced prior over covered call values} $\mathcal K(\mathcal P)$ be defined
%to be the joint distribution of the $n$-tuple $(\kappa_1,\ldots,\kappa_n)$ obtained by first sampling
%$\{(\theta_i,c_i)\}_{i=1}^n$ according to $\mathcal P$, then sampling $(v_1,\ldots,v_n)$ conditionally
%independently given $\{(\theta_i,c_i)\}$, with each $v_i$ drawn from the corresponding distribution
%$F_{\theta_i}$, and finally setting $\kappa_i$ equal to $\min\{v_i,\sigma_i\}$ for all $i$.

\begin{lem}\label{equivwelfare}
The highest expected welfare achievable by any procedure when types are distributed
according to $\mathcal P$ is equal to the highest expected welfare achievable when bidders are replaced with their
covered counterparts, who face no inspection costs and have types distributed
according to $\ccpt{\mathcal P}$.
\end{lem}

\begin{proof}
Weitzman's optimal procedure achieves the highest expected welfare of any procedure, and as
noted in Subsection \ref{VCGfail} it can be implemented by the dynamic VCG mechanism. We have seen
in the proof of \Cref{thm:opt-is-descending} that this expected welfare is exactly equal to
$\E \left[ \max_i \kappa_i \right]$, which is equal to the highest expected welfare achievable in a
model where agents' types are jointly distributed according to
$\ccpt{\mathcal P}$.
\end{proof}

\begin{definition}\label{dutch}
The Dutch auction is a sales procedure in which a clock begins at $\infty$ and continuously decreases. At any time, any bidder may claim the item by paying the current clock value, ending the auction.
\end{definition}

To state the precise sense in which Dutch auction equilibria are ``equivalent'' under
a type distribution $\mathcal P$ and its covered counterpart distribution
$\ccpt{\mathcal P}$ it will be useful to define the following equivalence relation.
\begin{definition} \label{funcequiv}
Two auction outcomes are said to be {\em functionally equivalent} if they award
the item to the same bidder at the same price, and the set
of bidders who pay a non-zero cost to inspect their value is the same.
Two strategies for a bidder are functionally equivalent (against a given profile
of opponents' strategies) if they always result in
functionally equivalent outcomes.
Two equilibria are functionally equivalent if they result in functionally
equivalent outcomes for every profile of types.
\end{definition}

\begin{thm}\label{maintheorem}
There is a mapping from equilibria of the Dutch auction in our model with types jointly distributed
according to $\mathcal P$ to equilibria of the Dutch auction with the covered counterpart
distribution $\ccpt{\mathcal P}$ where bidders know their values without inspection.
This mapping preserves bidder expected utility and auctioneer expected revenue, and it
induces a bijection on functional equivalence classes of equilibria.
\end{thm}
Informally the theorem states that equilibria of the Dutch auction in our model are in
bijective correspondence with Dutch auction equilibria when bidders are replaced with
their covered counterparts. To the extent that the formal theorem statement differs
from this informal summary, it is because the correspondence may fail to preserve
irrelevant details such as whether a losing bidder with zero information acquisition
cost chooses to inspect her value.

In a first price auction without information costs, individuals shade their bids downward relative to their value.  In our setting, individuals shade downward both their inspection time and, contingent on learning their values, their claim times according to exactly the same function, derived from the equivalent auction without information acquisition costs.  This implies that, at the time they inspect their values the clock is always weakly below their strike price and thus they can be sure they will always exercise in the money.  This ensures that individuals can fully offload the information cost onto the market and thus that the presence of information costs is ``irrelevant'' except in that it changes the value distribution to the covered call value distribution. To formalize this idea, our proof of \Cref{maintheorem} defines procedures by which a bidder can ``simulate'' (in an expected-utility-preserving manner) the best-response behavior of her covered counterpart and vice-versa.
%  Unfortunately defining this equivalence relationship in our environment with potential correlation across bidder types requires some formal delicacy.
% The following proof covers most of the central ideas, but we prove a few more detailed claims appearing in this proof in Appendix \ref{mainapp}.
In this section we elaborate on this simulation procedure and sketch
the central ideas of the proof. The proofs of detailed claims
appearing in the argument are deferred to \Cref{mainapp}.

Define a bidder's strategy to be {\em normalized} if she always
inspects her value at the earliest possible moment, and if the 
price $b(v,\theta)$ at which she claims the item given that her
value is $v$ and type is $\theta$ (if it is not yet claimed by 
another bidder) is a monotically non-decreasing function of $v$. 
For a bidder who faces no inspection cost, it is without
loss of generality to assume her strategy is normalized, in the
sense that any best-response in the Dutch auction 
is functionally equivalent to a normalized strategy.
(\Cref{mainapp}, \Cref{clm:1}.) Accordingly, we will henceforth
only consider normalized strategies for bidders 
whose type belongs to $\ccpt{\Theta_k}$. 

If $b$ is a normalized strategy of a bidder with
type space $\ccpt{\Theta_k}$, then a bidder with
type space $\Theta_k$ can simulate $b$ using a 
strategy $\lambda(b)$ defined as follows. From
her type $\theta_k$, the bidder can compute the
covered counterpart type $\ccpt{\theta_k}$ and
the strike price $\sigma_k$. She inspects her
value $v_k$ when the price reaches
$b(\sigma_k,\ccpt{\theta_k})$. If 
$v_k \geq \sigma_k$ she claims the item
immediately, otherwise she waits until 
the price decreases to $b(v_k,\ccpt{\theta_k})$ 
and claims the item at that price, if no
other bidder has claimed it. 

Note that,
strategy $\lambda(b)$ is non-exposed and,
contingent on winning the item, the price 
paid is the lesser of $b(\sigma_k,\ccpt{\theta_k})$
and $b(v_k,\ccpt{\theta_k})$. Since
 $b(v,\theta)$ is monotonic in $v$,
this means the price paid is $b(\kappa_k, \ccpt{\theta_k})$,
which is the same price paid by a bidder
with type $\ccpt{\theta_k}$ and value 
$\kappa_k$. Thus, the utility of a bidder
with type $\theta_k$ and value $v_k$ using strategy 
$\lambda(b)$ duplicates the utility of 
a bidder with type $\ccpt{\theta_k}$ 
and value $\kappa_k$ using strategy $b$,
and the auction outcomes achieved by these
two strategies, when facing the same profile
of opponents' bids, are functionally 
equivalent. (See \Cref{mainapp}, Claims \ref{clm:2}-\ref{clm:3}.)

We next describe the reverse simulation---that is, 
the process that a bidder of type $\ccpt{\theta_k}$ uses to
simulate the strategy of a bidder of type 
$\theta_k$.
% Let us again denote the strategy to be simulated by $b$.
The bidder begins by inspecting her own value, $\kappa_k$,
at the earliest possible moment. Knowing her own type
$\ccpt{\theta_k}$, she also knows the corresponding
type $\theta_k$ and strike price $\sigma_k$. She next
samples a simulated value $v_k$ by setting
$v_k = \kappa_k$ if $\kappa_k \leq \sigma_k$, and
otherwise sampling $v_k$ from the distribution
whose cumulative distribution function is
$F(x) = (F_{\theta_k}(x) - F_{\theta_k}(\sigma_k))/(1 - F_{\theta_k}(\sigma_k)).$
Note that this means the distribution of 
$v_k$, conditional on $\theta_k$, matches $F_{\theta_k}$,
so we can couple two bidders with types $\theta_k$ 
and $\ccpt{\theta_k}$ in such a way that the first
bidder's value is equal to the random $v_k$ sampled by
the second one.
Letting $b$ denote the strategy to be simulated,
the simulating strategy $\mu(b)$ simply claims the
item at the time when a bidder with type $\theta_k$,
value $v_k$, and strategy $b$ would claim the
item, if no other bidder has claimed it.
This implies that $b$ and $\mu(b)$
produce functionally equivalent outcomes at every
sample point. Moreover, by \Cref{lemma:upper-bound-policy},
the utility that a bidder with type $\ccpt{\theta_k}$ 
achieves using $\mu(b)$ weakly improves upon
the utility a bidder with type $\theta_k$ achieves
using $b$. (See \Cref{mainapp}, \Cref{clm:2}
and \Cref{clm:3}.)

The above arguments imply that $\lambda$ and $\mu$ define
a pair of mutually inverse bijections on functional 
equivalence classes of best responses for a bidder 
with type space $\Theta_k$ and one with $\ccpt{\Theta_k}$.
An easy induction---swapping one bidder at a time with
her covered counterpart, while holding the types and strategies
 of others bidders fixed---then establishes the claimed
 bijection between functional equivalence classes of
equilibria, completing the proof of the theorem.

\Cref{maintheorem} immediately implies that the desirable properties of the first-price auction carry over to the Dutch auction with information costs.

\begin{corollary}\label{maincorollary}
% \begin{enumerate}
If $\mathcal K\left(\mathcal P\right)$ is symmetric across bidders and bidders' types are independent,
then any symmetric equilibrium of the Dutch auction is efficient.
Even if $\mathcal K\left(\mathcal P\right)$ is asymmetric or bidders' types are correlated,
any equilibrium achieves
at least a fraction $1-\frac{1}{e}$ of the welfare of the optimal mechanism.
%\item Large market result??
%\end{enumerate}
\end{corollary}
\begin{proof}
By \Cref{maintheorem} it suffices to consider the case when bidders face zero cost
to inspect their value, and hence can be presumed without loss of generality to do so
at the earliest possible moment. A strategy of such a bidder is characterized by a
function mapping her type and value to her bid, i.e., the clock value at which she
will claim the item if no other bidder has yet done so. When types are independent,
this function can be assumed without loss of generality to depend on the bidder's
value only,  since the type contains no payoff-relevant information conditional on
the value. In equilibrium, the bid must
be a monotonically increasing function of the value (at least, on the
interval of values that have a positive probability of being greater than everyone
else's bid) and in a symmetric equilibrium all bidders apply the same such function,
resulting in the highest bidder winning the item.
In asymmetric environments or when types may be correlated,
if bidders face no information acquisition costs, the fact that any Dutch auction
equilibrium achieves at least a fraction
$1-\frac{1}{e}$ of the welfare of the optimal mechanism
is a consequence of the ``price-of-anarchy'' analysis of
first-price auctions, e.g.~\citep{syrgkanise}.\footnote{In the case of arbitrary correlation, \citeauthor{syrgkanisthesis} shows this bound is tight.  However, in the case of asymmetries without correlation (the one we study in our calibrations of the next section) it is not known whether this bound can be reached and it is widely believed it cannot be.  This may partly account for our findings below. }
\end{proof}

More broadly, any welfare property of the first-price auction carries over to the Dutch auction with information costs.\footnote{Another example is that folk wisdom, building off formal results in related contexts \citep{swinkels} suggests that in a ``large market'' first-price auctions are fully efficient; in addition to our assumption of independence mentioned in the previous footnote, this may partly account for our findings in the next section that in plausible calibrations the Dutch auction obtains nearly full efficiency.  We conjecture that under suitable restrictions on the tails of the value distribution to ensure that the distribution of the highest covered call value concentrates as the number of bidders grows large we would obtain full efficiency even with asymmetries or correlation.  However we do not pursue a formalization of this result here.}  For example, while we focus here on welfare rather than revenue, the Dutch auction in our context will achieve the same revenue as the corresponding first-price auction.  With appropriate reserve prices (bidder-specific end times for the auction, pre-anounced) \citep{POArevenue}, or with sufficiently many bidders \citep{bkbidders}, these revenues approximate (in the constant factor sense discussed in point 2 of Corollary \ref{maincorollary}) the greatest achievable revenue by any auction in this simpler context.  An interesting question we leave open to future research is whether there are ways to use the existence of information costs to extract a greater fraction of bidder surplus than is possible in the corresponding setting without information costs.\footnote{In any case, so long as the bidder value distribution is not too convex \citep{anderson2003efficiency} the revenue achievable in the corresponding auction approximates total welfare which, by individual rationality, upper bounds the greatest revenue achievable.  Thus in most reasonable settings (those where many moments of the value distribution exists) this question impacts only the approximation factor and not whether an approximation exists.  We state and prove these results formally in \Cref{app:revenue}.}

\section{Calibrations}\label{calibrations}

We have thus shown that the worst case performance of standard formats is far worse than that of the Dutch auction when information acquisition is necessary.  In this section we show that, in a pair of reasonable calibrated examples, this ranking is maintained but losses in both types of auctions settings are much smaller.  In essentially all cases the Dutch auction loses less than 1\% of value from the inefficiencies arising from asymmetries and typically they are an order of magnitude smaller than this.  Losses from standard formats are also small, typically a few percent though occasionally they are an order of magnitude larger.  The Dutch auction outperforms by a larger margin in revenue.

\subsection{Start-up acquisitions}\label{startup}

Our first calibration is designed to match broad moments, described in press reports, of the market for the acquisition of start-up companies by established firms, a process that is well-known to require a lengthy process of ``due diligence'' information acquisition.\footnote{
Our approach is not to precisely match these figures or conduct an empirical analysis, but to capture the key features of this setting: relatively large inspection costs, high variance in values, and a small to moderate number of bidders.} We also consider a wide range of robustness checks about our baseline parameters.\footnote{Our results are quite consistent across these values and it is this consistency, rather than the particular values chosen in our baseline calibration, that gives us confidence in our conclusions.}  In the second calibration we build more directly off empirical analysis of a particular market in previous literature.

\subsubsection{Set-up}\label{startupsetup}

We consider a framework where all primitive variables are joint log-normally distributed independently and identically across bidders.  Each bidder has a cost $c_i$ and a value $v_i$ distributed as:
\begin{eqnarray*}
\log\left(v_i\right)& = & V_i^0 + V_i^1+ V_i^2\\
\log\left(c_i \right) & = & C_i^0+C_i^1.
\end{eqnarray*}
All variables labeled $0$ are known commonly {\em ex-ante} to all bidders prior to bidding.  As such, they are a source of asymmetries across bidders.  For example, a bidder with a high value of $V_i^0$ will be commonly known by her rivals to be likely to be stronger than other bidders.  Thus while all bidder values are drawn i.i.d. bidders are not symmetric at the beginning of the auction.  Variables labeled $1$ are the costless private information (or ``type'') of the bidder {\em ex-interim}, but not known to her rivals; thus in the language of our set-up from Subsection \ref{model} above,  $\theta_i=\left(V_i^1, C_i^1\right)$.  $V_i^2$ is learned only after value inspection, at a cost $c_i$, which is mandatory due diligence for the acquisition of the start-up. 

We make the following distributional assumptions, with all variables whose correlation structure is not explicitly described being distributed independently:
\begin{eqnarray*}
\left( \begin{array}{c} V_i^0 \\ C_i^0  \end{array}  \right) & \sim & \mathcal N \left(\begin{array}{c} 0\\ \mu_c \end{array}  , \begin{array}{cc} \alpha_0  \sigma_v^2 & \rho \frac{\alpha_0}{\sqrt{\alpha_0+\alpha_1}} \sigma_v \sigma_c \\  \rho \frac{\alpha_0}{\sqrt{\alpha_0+\alpha_1}} \sigma_v \sigma_c & \frac{\alpha_0}{\alpha_0+\alpha_1} \sigma_c^2 \end{array}  \right) \\
\left( \begin{array}{c} V_i^1 \\ C_i^1  \end{array}  \right) & \sim & \mathcal N \left(\begin{array}{c} 0\\ 0 \end{array}  , \begin{array}{cc} \alpha_1  \sigma_v^2 & \rho \frac{\alpha_1}{\sqrt{\alpha_0+\alpha_1}} \sigma_v \sigma_c \\  \rho \frac{\alpha_1}{\sqrt{\alpha_0+\alpha_1}} \sigma_v \sigma_c & \frac{\alpha_1}{\alpha_0+\alpha_1} \sigma_c^2 \end{array}  \right) \\
V_i^2 & \sim & \mathcal N \left(0, \left(1-\alpha_0-\alpha_1\right) \sigma_v^2 \right)
\end{eqnarray*}
Thus we allow only for correlations across variables at the same time step and not across time steps.  We now briefly describe the interpretation of each of these parameters:
\begin{itemize}
\item $\mu_c$ determines the relative scale of cost compared to value; the scale of value is normalized to $0$ as in the log-normal distribution all properties of shape are invariant to this scale.

\item $\sigma_v^2, \sigma_c^2>0$  control the total variance in the logarithm of values and costs, which are in one-to-one correspondence with the associated Gini coefficient of the distributions of these variables.

\item $\rho \in [-1,1]$ is the correlation between log-costs and log-values (at the two time steps where such correlation is still possible) and we assume it is equal at the ex-ante and ex-interim stages.

\item $\alpha_0,\alpha_1 \in [0,1]$ with $\alpha_0+\alpha_1\in [0,1]$ determine the share of the total variance, in both cost and values, at stages $0$ and $1$.  For instance, the variance of $V_i^0$ is $\alpha_0 \sigma_v^2$ and that of $V_i^1$ is $\alpha_1 \sigma_v^2$. Because costs are fully determined by the ex-interim stage, the fraction of variance at the ex-ante stage is $\frac{\alpha_0}{\alpha_0+\alpha_1}$ and correspondingly for the ex-interim stage.

\end{itemize}

We compare the performance of the simultaneous second-price auction and the Dutch auction to the first-best allocation.  We do not consider the performance of an ascending auction for reasons of computational tractability.  With two bidders it can easily be shown that the (trembling hand perfect) equilibrium of the ascending auction is equivalent to that of the simultaneous second-price auction.\footnote{In a two bidder auction, each bidder will obtain the object if the other bidder drops out before she does.  As a result, each must choose a time at which she will take some action if that time is reached without the other bidder dropping out and whether this action will be to drop out or to inspect the object.  If the bidder expects to drop out, it is weakly dominant for her to do so at $\mathbb E\left[  v_i \left| V_i^1+v_i^0 \right. \right] -c$ as this is her value for winning the object if she has not yet inspected.  If the bidder expects to inspect, it is weakly dominant for her to do so immediately at the start of the auction, as she will have to inspect in any case if she wins and she can only gain by giving herself a more informed option to drop out earlier if her value turns out to be low.  Thus bidders either pursue a policy of immediate inspection and bidding of value or of no inspection and bidding $\mathbb E_{\theta_i}\left[ v_i\right] -c$.  But both of these strategies are available in the simultaneous second-price auction and thus (as long as we focus on weakly dominant strategy equilibrium) these two auctions are equivalent.  Matters are much more complicated with more than two bidders, as the drop out of one bidder may then trigger others to inspect dynamically; this is the source of \citet{jehiel}'s results, but also of the numerical intractability of the ascending auction in this context.} With more than two bidders, we were not able to determine, and to our knowledge no one has ever proposed, a computationally tractable method of finding approximate equilibrium. Given how well the Dutch auction performs with more than two bidders, it seems unlikely that investigating the ascending auction in greater detail would change our qualitative conclusions, though quantitatively \citet{jehiel}'s results suggest it might somewhat mitigate loses from search coordination relative to the simultaneous second-price auction.  We defer analysis of the BKRS sequential auction to the next calibration.  We briefly describe our computational methods in \Cref{numericsapp} and in greater detail in \Cref{sec:app-calibrations-extra}.

\subsubsection{Calibration}\label{startupcalibration}

We chose baseline parameter values to correspond roughly to press accounts of the market for start-up acquisition.  We chose $\sigma_c^2=.2$, a relatively low value, as costs of due diligence seem relatively homogeneous across potential purchasers, and $\sigma_v^2=1$, corresponding to a Gini coefficient of $.52$, as this degree of inequality (similar to the degree of pre-tax inequality in many developed countries) seems a reasonable spread in values in a financial context.  We chose $\mu_c=-.62$ as this leads average costs to equal a third of average values, which is consistent with press accounts of the industry and with the fraction of search costs in procurement contracts found in a different context by \citet{salz}. We chose $\rho=.7$ because press accounts indicate that due diligence costs covary heavily with firm size, as larger purchasers tend to hire more expensive investment bankers and firms with high idiosyncratic value tend to do more due diligence. We had no particularly principled way to choose the values of $\alpha$, so we settled on $\alpha_0=.1$ and $\alpha_1=.4$ as a baseline as these imply relatively modest levels of asymmetry, a fair bit of selective entry, and significant variation revealed by inspection.  We chose five bidders for our baseline scenario as this seems a reasonable number of potential bidders for a hot start-up.

We considered variations in each parameter above and below these baseline parameters: $\sigma_v^2=.5,1.5$; $\mu_c=-1.0, -.35$; $\sigma_c^2=0, 1$; $\rho=-.5,0,1$; $(\alpha_0,\alpha_1) = (.5,.25), (.5,0.05), (.1,.1), (.1,.7)$. 
% \alpha_0=.5 \& \alpha_1=.25, \alpha_0=.5 \& \alpha_1=0, \alpha_0=.1 \& \alpha_1=.1, \alpha_0=.1 \& \alpha_1=.7$. 
We also consider the case of two and ten bidders not just for our baseline but also for the calibrations where we found the greatest efficiency loss for the Dutch auction, as we expect a small number of bidders to exacerbate these harms.  

Finally we manipulated parameter values to find the worst case for each auction within our general structure.  The worst case we could find for the Dutch auction was our baseline but with two bidders.  For the simultaneous second-price auction the worst case we found had very high cost $\mu_c=-.2$ and the vast majority of uncertainty realizing only after costly inspection ($\alpha_0=\alpha_1=.01$).\footnote{Given the very small ex-ante asymmetry in this case, we only drew a single bidder type and ran a symmetric example.}  We ran versions of this with our baseline of $5$ and $20$ bidders, as we found increasing the number of bidders made the inefficiency worse.

\subsubsection{Results}\label{startupresults}

\begin{table}[t]
\begin{center}
\begin{tabular}{ccc}
Parameter values   & Second Price & Dutch \\ \hline \hline
Baseline           & $3$ & $.01$ \\
$\sigma_v^2=.5$    & $.4$ & $0$ \\
$\sigma_v^2=1.5$   & $2$ &  $.02$ \\
$\mu_c=-1.0$       & $4$  &  $0$ \\
$\mu_c=-.35$       & $3$ & $.03$ \\
$ \sigma_c^2=0$    & $4$  & $.01$\\
$ \sigma_c^2=1$    & $4$  & $0$ \\
$ \rho=-.5$        & $3$  & $.03$ \\
$ \rho=0$          & $3$   & $.02$\\
$ \rho=1$          & $3$   & $.01$ \\
$ (\alpha_0,\alpha_1) = (.5,.25)$ &  $2$ & $.01$ \\
$(\alpha_0,\alpha_1) =  (.5,0.05)$   & $4$ & $0$ \\
$(\alpha_0,\alpha_1) =  (.1,.1)$  & $7$ & $.01$ \\
$(\alpha_0,\alpha_1) =  (.1,.7)$  &  $1$ &  $.04$ \\
$N=2$              & $1$ &  $1$ \\
$N=10$ & $5$ & $0$ \\
Worst case $N=5$ & $20$ &$0$ \\
Worst case $N=20$ & $50$ & $0$
\end{tabular}
\end{center}
\caption{Welfare loss as a percentage of the first best in start-up acquisition calibration, rounded to the first significant digit.  Zero indicates that in all runs there was exact agreement with the first best allocation.}\label{startuptable}
\end{table}

\Cref{startuptable} represents our results for all the parameter settings we study. Results are the percent welfare loss relative to the first best rounded to a single significant digit.  Because we only averaged over three ex-ante draws, this is the degree of our numerical precision across these three draws.

In all but two cases the welfare losses from the simultaneous second-price auction are two orders of magnitude greater than those from the Dutch auction.  Only in the case with two bidders are the losses comparable; in this case the losses from the two mechanisms were differently ranked in the three draws and thus within our numerical error.  In one other case there was only a single order of magnitude difference in the loses; this was the case when $\alpha_1=.7$ and thus there are very strong ex-interim signals of eventual values so inspection is not critical.

  In fact, as we show in Appendix \ref{sec:app-calibrations-extra}, very few inspections occur on average; typically about half as many inspections occur in the second-price auction as in the Dutch auction.  Thus the allocation ends up being very different in the two cases; allocation is much less efficient under the second-price auction, but fewer inspections take place.  Given that in our baseline setting there is a quite strong signal of value available to bidders ex-ante, this still results in a quite high fraction of social welfare, implying that the second-price auction only loses a few percentage points of welfare relative to the first best. However, this failure to invest in inspection decimates revenue in the auction; in our baseline the second-price auction achieves 40\% less revenue than the first-best and consistently gets 30\% less revenue, as we discuss in \Cref{sec:app-calibrations-extra} in greater detail.  Furthermore, in a worst case scenario, when there are many bidders, very high inspection costs and a weak signal of values at the interim stage, the second-price auction loses 50\% of the total value.

This suggests to us that in almost any realistic case, the losses from poor search coordination of the second-price auction will be dramatically larger than those from the asymmetries in the Dutch auction.  However, these losses will not be very large in magnitude unless several key features conspire together: large inspection costs, high variance in value from inspection, and many bidders.  This was precisely the conspiracy underlying our worst case results.  In most realistic settings which lack one or more of these features, the Dutch auction seems likely to lead to much higher revenue and orders of magnitude smaller welfare loss relative to the optimum, but not to a dramatic quantitative rise in the absolute value of welfare.
%This suggests to us that in almost any realistic case, the losses from poor search coordination of the second-price auction will be dramatically larger than those from the asymmetries in the Dutch auction.  However, these losses will not be very large in magnitude unless there is a ``conspiracy of parameter values'' that make inspection both important and costly.  This was precisely the conspiracy underlying our worst case results.  In more realistic settings the Dutch auction seems likely to lead to much higher revenue and orders of magnitude smaller welfare loss relative to the optimum, but not to a dramatic quantitative rise in the absolute value of welfare.

\subsection{Timber auctions}\label{timber}

Our previous calibration was quite flexible and we considered a wide range of scenarios, but was only very loosely calibrated to the start-up acquisition setting.  We now consider a second calibration that is somewhat less flexible and in an environment where information acquisition costs are relatively smaller, but which is drawn from an actual empirical estimation.  In particular we import the framework of \citet{sweeting} which is estimated based on US timber auctions.  We do not discuss the institutional context of these auctions at all and present the set-up and calibration extremely briefly to avoid excessive duplication of material in their article.    We defer discussion of numerical methods entirely to appendices,

\subsubsection{Set-up}\label{timbersetup}

\citeauthor{sweeting}'s model is similar to that in our previous calibration but somewhat simpler.  First, costs are assumed homogeneous across all bidders: $\sigma_c^2=0$ and $\rho$ is irrelevant.  Second, $V_i^0$, rather than being drawn normally, takes on two values (one for ``loggers'' and one for ``millers'').  As a result, there is only a single $\alpha$ parameter governing the fraction of the variance realized in $V_i^1$ v. $V_i^2$.

 \citeauthor{sweeting} are interested in the comparison between the simultaneous second-price auction and the sequential BKRS mechanism we described in Subsection \ref{flaws} above.  \citeauthor{sweeting} actually consider an altered version of the second-price auction in which inspection is not only mandatory for receiving the object, but even for placing a bid.  This makes the second-price auction perform strictly worse than the version we studied in the previous calibration by the main theorem of \citet{VCGinfo}.  In particular in the previous section we found that the relatively good performance of the second-price auction arose almost entirely from the fact that it allowed for bidding prior to inspection.\footnote{\citeauthor{sweeting} also consider a sealed-bid first price auction; we do not include this here as by \citeauthor{VCGinfo}'s theorem this must perform weakly worse than the simultaneous second-price auction.}  As a result it may be that the gains they find of the sequential over the second-price auction would diminish or disappear if bidding prior to inspection were allowed, but for consistency with their work we do not consider this quantitatively here.

\subsubsection{Calibration}\label{timbercalibration}

We simply import the parameter values estimated by \citeauthor{sweeting} and the robustness variations about these that they consider.  To parameterize ex-ante asymmetries, they assume that the log-mean of logger values is $\mu_{\mathit{logger}}=3.582$ and of mills $\mu_{\mathit{diff}}=.378$ higher.  They set $\sigma_v^2=.332$, $\alpha=.689$ and the the homogeneous inspection cost at $K=2.05$.  They assume $N_{\mathit{loggers}}=4$  loggers and $N_{\mathit{mill}}=4$ mills bid in the auction.  About these baseline parameters they consider the following variations:  lowering and raising the $N_{\mathit{mill}}$ to $1$ and $7$; $N_{\mathit{logger}}$ to $0$ and $8$; $\mu_{\mathit{logger}}$  (while holding $\mu_{\mathit{diff}}$ constant) to $2.921$ and $4.243$; $\mu_{\mathit{diff}}$ to $.169$ and $.587$; $\sigma_v^2$ to $.122$ and $.646$; $\alpha$ to $.505$ and $.872$ and $K$ to $.39$ and $3.72$.

In addition, we consider the variation where the inspection cost $K=16$, about one-third of expected value, to check our hypothesis that rising inspection costs significantly impact relative performance.

\subsubsection{Results}\label{timberresults}

\begin{table}[t]
\begin{center}
\begin{tabular}{cccc}
Parameter values              & Second Price & Sequential & Dutch \\
\hline \hline
Baseline                      & $2.5$       & $.79$     & $.02$ \\
$N_{\mathit{mill}}=1$         & $2.4$       & $.58$     & $.14$ \\
$N_{\mathit{mill}}=7$         & $2.8$       & $1.1$     & $.79$ \\
$N_{\mathit{logger}}=0$       & $1.6$       & $.70$     & $0$ \\
$N_{\mathit{logger}}=8$       & $3.0$       & $.91$     & $.02$ \\
$\mu_{\mathit{logger}}=2.921$ & $4.4$       & $2.3$     & $.02$\\
$\mu_{\mathit{logger}}=4.243$ & $3.8$       & $1.3$     & $.02$ \\
$\mu_{\mathit{diff}}=.169$    & $3.1$       & $1.1$     & $.01$ \\
$\mu_{\mathit{diff}}=.587$    & $2.1$       & $.79$     & $.07$\\
$\sigma_v^2=.122$             & $2.9$       & $.89$     & $.01$ \\
$\sigma_v^2=.646$             & $5.6$       & $4.5$     & $.25$ \\
$\alpha=.505$                 & $2.6$       & $.86$     & $.02$ \\
$\alpha=.872$                 & $2.4$       & $.52$     & $.01$ \\
$K=.39$                       & $.55$       & $.17$     & $.02$ \\
$K=3.72$                      & $4.2$       & $1.3$     & $.02$ \\
$K=16$                        & $17$       & $3.8$     & $.02$ 
\end{tabular}
\end{center}

\caption{Welfare loss as a percentage of the first best in timber auction calibration, rounded to nearest two significant digits.
         Parameter values given are the ones that differ from the baseline.}\label{timbertable}

\end{table}

\Cref{timbertable} represents our results for all the parameter settings we study as the percent loss in welfare compared to the first-best allocation rounded to two significant digits, consistent with our numerical error which is at this scale of resolution.  For every parameter setting the Dutch auction has substantially lower welfare loss than the other two formats. In all but three cases, none including the baseline, the loss is within a tenth of a percent of the optimum and except in a few cases it is at least an order of magnitude smaller than the loss from either of the other formats.  In many cases it is more than two orders of magnitude smaller.  Furthermore, as we show in \Cref{sec:app-calibrations-extra}, results on revenue are generally consistent with these conclusions. The Dutch auction obtains higher revenue than the second-price auction in all of these cases and higher than the sequential BKRS sequential mechanism in all but two cases. 

However, both the second price and sequential auctions perform quite well in these settings, implying that the gains from the Dutch (and from the sequential over the second price) auction are modest in absolute terms, though as \citeauthor{sweeting} note they are fairly large by the standard of relative welfare and revenue gains from other design changes usually considered in the empirical auctions literature. In most cases the sequential mechanism performs significantly better than the second-price auction and this absolute gain is usually bigger than that of the Dutch over the sequential mechanism.  However, the relative decline in loss compared to the first-best is almost always greater for the Dutch over sequential than sequential over second price and it is possible that the gains of sequential over second price primarily arise from the prohibition on pre-inspection bids in the second-price auction as simulated by \citeauthor{sweeting}

One important reason these gains could be small is that the costs of information acquisition at the baseline are only 2\% of typical values; when we increase this by an order of magnitude the losses grow, especially from the second-price auction, but by much less than an order of magnitude. Losses from the sequential mechanism are still under four percentage points, consistent with our findings in our start-up calibration that even with costs that are about a third of typical values losses from mechanisms that poorly coordinate search are only a few percentage points.

  Thus while the Dutch auction seems clearly to perform dramatically closer to the full optimum than alternative mechanisms, in environments with very small information costs the choice among mechanisms may not make much of a difference and even in cases with large costs the loses are often modest. In any case, any losses from asymmetries in this context appear to be negligible in the vast majority of cases and extremely small in all cases.

\section{Extensions}\label{extensions}

\newcommand{\itemset}{{\mathcal{M}}}
\newcommand{\bidderset}{{\mathcal{N}}}
\newcommand{\assignset}{{\mathcal{A}}}
\newcommand{\gralloc}{{\mathbb{A}}^{\circ}}
\newcommand{\grinspect}{{\mathbb{I}}^{\circ}}
\newcommand{\optalloc}{{\mathbb{A}}^{\ast}}
\newcommand{\optinspect}{{\mathbb{I}}^{\ast}}
\newcommand{\bdualvar}{\alpha}
\newcommand{\idualvar}{\beta}
\newcommand{\tbdualvar}{\tilde{\bdualvar}}
\newcommand{\tidualvar}{\tilde{\beta}}

\newcommand{\reals}{{\mathbb{R}}}
\newcommand{\sigfld}{{\mathcal{F}}}
\newcommand{\given}{\, \| \,}
\newcommand{\stg}{{\frak s}}
\newcommand{\patient}{{\frak r}}
\newcommand{\altstg}{{\frak r}}
\newcommand{\alloc}[1]{{\mathbb{A}^{#1}}}
\newcommand{\allocij}{{\alloc{ij}}}
\newcommand{\allocstg}[1][\stg]{{\mathbb{A}^{#1}}}
\newcommand{\stopstg}{{\stg(\infty)}}
\newcommand{\eps}{{\varepsilon}}
\newcommand{\indic}{{\mathbf{1}}}
\newcommand{\di}{{\mathrm{DI}}}

While our theorem from Section \ref{main} is fairly general, our results may still be extended beyond it in a number of directions.  Because it is cumbersome to incorporate these together into a maximally general model, given the different extents to which the full force of our analysis carries over in each case, we now consider a series of disjoint extensions.  In many cases similar results may be obtained when various extensions are combined, but in the interests of space we do not discuss this further here.  Furthermore and for the same reasons,  all set-off formal proofs in this section have been moved to Appendix \ref{extapp}, except for those in \Cref{multistage} which occupy Appendix \ref{app:multistage}.

\subsection{Many stages of inspection}\label{multistage}

In our baseline model we assumed that inspection is a binary choice.  However our results all extend to a model with many stages of inspection, so long as complete inspection is mandatory for obtaining the object.  In this subsection we describe such an environment.

Bidders may inspect the object in \emph{inspection stages}, at each stage
investing additional resources to acquire more information. Inspection stages are discrete, 
taking values in $\mathbb{N} \cup \{\infty\}$;
the special value $k=\infty$
denotes the completion of all stages of inspection.
At the beginning of the procedure, each bidder $i$
is in inspection stage $0$.
Each stage of
inspection is instantaneous and the bidder may choose to
advance the inspection stage at any time,
including advancing any number of stages at the same
moment in time. When a bidder acquires the item, she
immediately completes all (countably many) stages of
inspection for it.\footnote{%
  To model only a finite number of stages, let all
  subsequent stages have no additional cost and
  reveal no information.}

The information that bidder $i$ gains  upon reaching the $k^{\text{th}}$ inspection stage 
is represented by a signal $s^k_{i}$.  Each bidder $i$ draws a private type $\theta_i \in \Theta_i$ and  bidders know their own type at all times.
The bidder's valuation for the item is a random variable
$v_i$ whose value is completely determined by the information
learned while inspecting.
The cost of inspection is represented by a sequence of random variables
$0 = c^0 \leq c^1 \leq c^2 \leq \cdots \leq c^\infty$, where $c^k$ represents
the cumulative cost of reaching the $k^{\mathrm{th}}$ inspection stage; we will
assume that this cumulative cost is one of the pieces of information incorporated
into the signal $s^k_{i}$. 
We will also assume that $\E[v^+] < \infty$ and $\E[c^{\infty}] < \infty$.  

\begin{prop}\label{manystages}
Lemma \ref{equivwelfare}, Theorem \ref{maintheorem} and thus Corollary \ref{maincorollary} hold in the model with many inspection stages, with an appropriately defined generalized notion of strike price and covered call value.
\end{prop}

Informally, the generalization of strike price is a (random) sequence of prices, 
one for each stage of inspection, each representing the strike price of a call
option whose fair value would exactly cover the expected remaining inspection
costs of a decision maker, given that the decision maker has already reached
the given inspection stage and invested the attendant costs. The generalized
covered call value is then defined as the minimum of this sequence of strike
prices.  Intuitively this is the value that remains to an individual who, at every stage of inspection prior to inspecting sells the relevant call option to fund her future information acquisition. We leave to the appendix the formal definition, as it is somewhat complicated.  
But, as in \Cref{main}, it preserves all relevant welfare quantities and thus has the same substantive implications as in the case with a single stage of inspection.

\subsection{Many objects}\label{multiunit}

Our baseline analysis focuses on the case of a single object\footnote{%
  Actually, our baseline analysis also applies---with only minor
  modifications in the proofs---to the case in which 
  a fixed number of identical items are for sale or, even more
  generally, to the case in which certain subsets of agents may be
  chosen as winners of the auction, the collection of these feasible
  sets of winners constitutes a matroid, and $v_i$ represents agent
  $i$'s value for winning. The Dutch auction is modified so that:
  (i) an agent may join the set of winners at any time, by paying the
  current clock value, (ii) agents are eliminated from the auction
  whenever they do not belong to any feasible superset of the 
  current set of winners, (iii) the auction ends when every agent
  has either been included in the set of winners or eliminated.
  Agents are informed when they are eliminated, but they are not 
  informed when other agents join the set of winners.
}, but our basic insight extends to a context with multiple items. 
This context requires an auction format, with a uniform descending price across objects, that is novel to our knowledge.\footnote{A patent on this mechanism is pending with the United States Patent and Trademark Office under the name ``Descending Counter Value Matching with Information Sharing''. \label{patent}}

A set $\itemset$ of $m$ heterogeneous items 
is for sale to a set $\bidderset$ of $n$ bidders.
Bidders have unit demand; each wants to acquire at most one item.
As in our baseline model we will denote a generic bidder by $i$.
We will use the symbol $j$ to denote a generic item.   Bidders are endowed, as in the baseline model, with a type $\theta_i$ determining a vector of costs $c_i\left(\theta_i\right)$ with typical element $c_{ij}$ (from which we drop the argument) denoting the cost of inspecting object $j$ for bidder $i$.  Type $\theta_i$ is drawn for each bidder according to an arbitrary joint distribution.  Given $\theta_i$, bidder $i$'s value $v_{ij}$ for object $j$ is drawn independently across $i$ and $j$ according to distribution $F_{\theta_i, j}$.

\begin{definition}
The {\em uniform descending auction} is a sales procedure that proceeds as follows. A clock begins at $\infty$ and continuously decreases and a set of items that are available is maintained; both are displayed publicly. At any time, any bidder may claim any available item at the clock price, at which point both the bidder and item are removed instantly.  The clock then continues to descend until it reaches~$0$.
\end{definition}

The equilibrium of this mechanism has not been analyzed in previous literature to our knowledge even in a context without information acquisition and when information acquisition is incorporated into the model, there is no known computationally tractable procedure that is guaranteed to achieve more than half of the welfare of an optimal allocation, making it difficult to imagine full efficiency is tractably achievable.  As a result, in this setting, an analog of Theorem \ref{maintheorem} is less meaningful.  We instead prove an analog of the second point of Corollary \ref{maincorollary}, that any equilibrium will quite tightly approximate the welfare guarantee of the best tractable assignment protocol we are aware of under complete knowledge of all {\em ex ante} known information, namely the greedy matching algorithm that is guaranteed to achieve half of optimal welfare.  Our argument in the appendix draws heavily on the literature on smoothness in computational mechanism design \citep{lucier,roughgarden, syrgkanise, syrgkanis, hartline} which is used to derive approximate efficiency results of this form.

\begin{thm}\label{manyobjects}
In any Bayes-Nash Equilibrium of the uniform descending auction, expected welfare is at least 
$\frac12 \left( 1 - e^{-2} \right) \approx 0.43$ of the first-best.
\end{thm}

Thus the basic intuition behind our result in the single item case that guarantees reasonable performance carries over to a setting with many objects for sale.

\subsection{Approximate best response}\label{rationality}

Determining a best response may be difficult in our setting, especially in the extension of the previous subsection,  because deciding both when to investigate which object and what to bid upon investigation requires a rich set of strategic calculations about the likely behavior of other agents.  Luckily, our results do not depend sensitively on agents being able to play perfect best-responses.  This is an attractive property of the approximation arguments like those we use to establish Theorem \ref{manyobjects} \citep{roughgarden15robust}.  In particular, we will say that bidders $\alpha$-best respond if they obtain at least a fraction $\alpha \leq 1$ of the greatest utility they could achieve through any strategy.  Our bound declines continuously in $\alpha$.

\begin{prop} \label{prop:bounded-rationality}
In the Dutch auction with a single item, if all bidders $\alpha$-best respond, then expected welfare is at least a factor $1-e^{-\alpha}$ of optimal. In the uniform descending auction with multiple items, if all bidders 
$\alpha$-best respond, then expected welfare is at least a factor $\frac12 \left( 1 - e^{-2 \alpha} \right)$ of optimal.
\end{prop}

For reasonably large $\alpha$ this decay is quite gradual indeed.  
For example, when $\alpha = 0.7$ the expected welfare of the 
Dutch auction with a single item is at least half of optimal, 
as compared with the approximation guarantee of $1-\frac{1}{e}
\approx 63\%$ when players exactly best-respond to one another. 
The expected welfare of the uniform descending auction with multiple
items is at least 37\% of optimal when $\alpha=0.7$, as compared
with the approximation guarantee of $1 - \frac{1}{e^2}\approx $43\% when 
players exactly best-respond to one another.

\subsection{Common  values}\label{common}

\citet{milgromweber} show, in a context without information acquisition costs, that ascending price auctions quite generally lead to greater revenue than a descending price auction if individuals have {\em common values} in the sense that the true but unknown value of the object is common across all individuals who all have a signal of the object's value.  This raises the question of whether, in a common values version of our model, an ascending price mechanism might out-perform a descending price mechanism.  In this subsection we show that, to the contrary, the bad potential performance of ascending price mechanisms and the good guaranteed performance of a descending price mechanism extend to a simple common values context.

Consider our baseline environment with the following modification to the value and information structure.  Every individual has a common prior on the true value of the object, $v$, and an individual-specific cost $c_i$ of determining this common value that is private knowledge.  Expenditure of this cost is required to obtain the item as above.  

Because the object is of equal value to all individuals, the only scope for efficiency is in ensuring that the individual who investigates the object has the minimum possible cost for doing so.  In an ascending or simultaneous second-price auction, there is always an equilibrium that implements this efficiency as proven by \citet{VCGinfo}. However, as \citet{tomoeda} shows, there are often very inefficient equilibria.  In our case this is particularly easy to illustrate.  

Suppose all costs are known and that there are two bidders, Alice and Bob.   Alice has cost $\mathbb E[v]-\epsilon$ and  Bob has cost  $\epsilon$, where $\mathbb E[v]\gg\epsilon>0$.  The best welfare is $\mathbb E[v]-\epsilon$ and is achieved in an equilibrium where Alice believes Bob will query and claim the object at price $0$ after submitting a bid of $\mathbb E[v]$ and facing no competition.  However suppose that Bob believes Alice will, with certainty, query and then submit a bid of $v$.  Bob then has no incentive to bid anything above $0$ as there is no way he can avoid taking losses if he queries.  Alice makes positive profits of $\epsilon$ in this case and would lose these by not querying.  Thus this state constitutes an equilibrium.  As $\epsilon \rightarrow 0$ it remains such and thus the price of anarchy is arbitrarily large in this example.

Conversely, the following theorem shows that the descending procedure is robustly approximately efficient in this case.

\begin{prop}\label{commonprop}
In the common values environment described in this section, the descending price, single-item auction achieves at least $1-\frac{1}{e}$ of the welfare of the social optimum.
\end{prop}

 This suggests that it may be possible to extend our results into richer common values environments, for example where bidders have private signals of the valuation or have both a common and private value component.  This is an interesting direction for future research.

\subsection{Optional inspection}\label{optional}

In our baseline model inspection is mandatory for an individual to consume an object.  In many cases this seems realistic: in acquiring a company it is often mandated by law that due diligence is complete, purchasing a house usually requires inspection and assessments and hiring job market candidates in academia rarely  is possible without a campus interview.  However, in some contexts, such as buying a car, inspection may be optional  prior to consumption.  

This possibility raises interesting challenges for a descending price mechanism.  While knowing that the price of the object is not too high (bounding it from above) may be critical for deciding between inspection and forgoing the object entirely (deciding between ``no'' and ``maybe''), it is when the price of the object is very low that it will be optimal to consume it without bothering to inspect.  For example, suppose there are two bidders. 
Bidder $A$ has a cost of inspection equal to 1 and a 
value equal to $v$ with probability $p > 1/v$ and to 0 
with the remaining probability. Bidder $B$ has no inspection
cost and value $w$ sampled from a distribution will full support
on $[0,\infty)$. Since it is cost-free to inspect $B$'s value, $w$,
the first-best procedure starts by doing so. If $w$ is greater than bidder $A$'s
strike price (namely, $v - 1/p$), then it is optimal to simply award
the item to $B$. If $w$ is sufficiently close to 0, then it is optimal
to award the item to $A$ without expending the cost of inspecting
$A$'s value. If $w$ lies in a non-empty interval between these two extremes, 
it is optimal to inspect $A$'s value
and then award the item to the bidder with the higher value. Thus,
in order to implement the first-best procedure, 
before $A$ decides to inspect her value it is necessary
to know both a {\em lower} and an {\em upper} bound on
$B$'s value. While we spare the reader an equilibrium
construction, it should be clear that a descending procedure
will not be able to signal that $B$'s value is above
the threshold value necessary to spur inspection by $A$
in the first-best procedure.

More generally, for a bidder to choose between inspection and immediate purchase (deciding between ``maybe'' and ``yes'') it may be crucial to know that the price of the object will not be too low (bounding it from below).  It thus seems plausible that the greatest efficiency might be achieved by neither an ascending nor a descending price mechanism, but rather a hybrid where potential prices converge from both ends.  On the other hand, there is no known computationally tractable algorithm for achieving the first-best in this setting \citep{doval}.

However, at least in the case where inspection occurs in a single stage and there is a single object, the inefficiency due to inspecting when one might instead have immediately claimed the object cannot ``compound'' in the way inefficiency due to not considering the object at all can.  Intuitively the reason is that once an individual has claimed the object, no other efficiency loss is possible from any other individual.  Any given individual can never lose more than half of the welfare achievable by the optimal policy by simply choosing to ignore one of her two options. 
More formally consider our baseline model with the lone modification that bidders now have the option to claim and receive $\mathbb E_{\theta_i}\left[v_i\right]$ without first inspecting.  

\begin{prop}\label{prop:optional}
In the optional inspection model, any Bayes-Nash equilibrium of the Dutch auction achieves at least $\frac{1}{2}-\frac{1}{2e}$ of the welfare of the first-best.
\end{prop}

Thus the descending price procedure at least approximates optimality in this context.

\subsection{Posted pricing}\label{posted}

The fundamental intuition behind our result is that, to perform well, a mechanism must allow bidders to exercise in the money by guaranteeing them, at the point they must decide on inspection, that they can acquire the object at a price with an upper bound.  A mechanism that is in some ways even simpler than ours and allows this is to sequentially approach bidders and offer them a posted price.  This obviously requires the seller to know a reasonable price to post.  When the bidders' types are independent and their distributions are known to the seller it is possible to calculate such a price that can approximate optimal welfare.

\begin{thm}[``Prophet inequality'', \cite{krengel1978semiamarts,samuel1984comparison}]  
\label{thm:prophet}
Let $X_1,\dots,X_n$ be independent nonnegative random variables. Assume
that the distribution of the random variable $X^* = \max_i X_i$ has no point masses,
and let $\pi$ be the median of this distribution.
Let $Z$ equal the lowest-index $X_i$ that exceeds $\pi$, or $0$ if none do.
Then $\E Z \geq \frac{1}{2} \E X^*$.
\end{thm}

\begin{corollary}\label{prophetcorr}
When bidders are fully independent, the sequential posted-price procedure, with price $\pi$ set to the median of the distribution of the maximum covered call value, obtains half the optimal expected welfare.
\end{corollary}

\begin{remark}
A different procedure that arranges the bidders in random order and 
sequentially offers bidder-specific posted prices can be shown to obtain
at least $1 - \frac{1}{e}$ fraction of the optimal expected welfare, using
the ``prophet secretary inequality'' of \citet{esfandiari2015prophet}.
\end{remark}

While escaping the extreme failures of standard procedures, the sequential posted price mechanism sacrifices the price discovery benefits of auctions.  Unlike the descending procedure, it is not fully efficient even under symmetry and without information acquisition.  Any posted price will exclude some bidders who may turn out to be the highest value ones in any given sample or will allocate the item to some bidders who turn out not to have the highest value in a given sample. 
% Furthermore when values are correlated (say when there is aggregate uncertainty) again even in the case without 
% information acquisition it is know that sequential posted pricing may perform arbitrarily poorly (CITATION).  
The descending price auction maintains the price discovery features of an auction while providing the guarantees necessary to assure exercise in the money.

\section{Beyond Auctions}\label{matching}
All of our formal results concern auctions, that is, transferable utility environments where bidders value objects that are indifferent about which bidder obtains them.  However a large literature on ``market design'' in recent years has drawn links from insights about auctions to mechanisms for matching pairs of agents with preferences on both sides of the market, and also to allocation problems where transfers are not available.

For example, \citet{hatfield} show that the underlying structure of ascending-price auctions is closely connected to that of the ascending match value (lowest value matches made first) procedure for transferable utility matching proposed by \citet{crawfordknoer} and \citet{kelsocrawford} and the deferred acceptance algorithm of \citet{galeshapley}.  Quantitative welfare guarantees in incomplete information environments are often not available in these alternative settings because of bilateral monopoly problems \citep{myersonsatt} and the lack of transfers.  However the structural connections between auction and matching mechanisms have been used to suggest the attractiveness of commonly-used matching procedures and vice versa.  This suggests that some of the insights about auctions in our results may have analogs in matching settings.  In this section we briefly and informally speculate about these settings.  We hope future research will more formally quantify the welfare properties of these proposals.

\subsection{Two-sided matching and the Marshallian Match}\label{twosided}

Most of the centralized matching procedures advocated by economists \citep{rothpopular} involve participants simultaneously submitting preference lists.  This may lead to large welfare losses for reasons analogous to those we highlight in the auction context.  For example, medical students entering into the National Resident Matching Program must submit preferences over all hospitals they might match with before they have any clear sense of which hospitals are in their choice set.  It may be difficult for such students to know which hospitals are worth investigating more deeply, leading to potentially large welfare losses. In fact two-sided matching in many markets not organized by centralized matching appears to have features reminiscent of a descending price auction; namely matches clear sequentially, from high value matches downward, and participants are given ``guarantees'' at various points in time about their options prior to the process of information acquisition.  Two examples are illustrative.

First,  in college admissions, it is common for students to ``apply early'' to a favorite choice and to therefore receive a more thorough review of their application.  Then schools commit to accepting their top choices of students, prior to these students flying out to the schools for ``second visits''.  Only after students have engaged in this information acquisition and made decisions do schools fill their remaining slots from their wait lists. Second, in the economics job market, information acquisition proceeds in stages, from application review to interviews to fly-outs to offers to second visits to applicant choices and then to second round offers.  This process progressively narrows the pool of interest and increases the degree of potential match commitment before more costly information acquisition takes place.  Even within rounds, especially on fly-outs and offers, top schools tend to ``go before'' lower ranked schools to avoid lower ranked schools wasting a fly-out or offer on an applicant who is likely to have better offers.

Inspired by these examples, we propose and have submitted a patent application (see Footnote \ref{patent} above) on a two-sided mechanism with transferable utility.  There is a clock that descends from $\infty$ to $0$. Employers maintain bids on any subset of employees they are interested in hiring.  Employees maintain asks of any employer they are willing to work for. When the clock descends to the value of the bid-ask spread on any edge, the employer-employee pair along that edge are matched and removed.  The wages the employer pays the employee are the average of the bid and the ask.  When the clock reaches $0$ the match ends and all unmatched participants leave unmatched.

We refer to this mechanism as the {\em Marshallian Match} following \citeauthor{marshallianpath}'s \citeyearpar{marshallianpath} discussion of \citeauthor{marshall}'s \citeyearpar{marshall} theory of equilibrium.  In particular \citeauthor{marshallianpath} argue that both in \citeauthor{marshall}'s theory and in laboratory experiments, the highest value potential trades tend to happen first.  While they consider homogeneous good markets where transactions are just defined by a uniform-across-partner willingness-to-pay or willingness-to-accept, the natural extension of this logic to matching seems to be a mechanism like that we describe.\footnote{Our mechanism may also be applied to settings with homogeneous values across matches for an given individual, like that studied by \citeauthor{marshallianpath}; then each employer and employee need report only a single number.  A natural application would be to two-sided spectrum auctions (though one where homogeneity is only on the sellers' side and partially on the buyers' side) such as the on-going incentive auctions in the United States \citep{radio}.  In this setting sufficient competition may overcome the \citet{myersonsatt} results and allow a more positive result.}

\subsection{Non-transferable utility}\label{NTU}

In some settings where information acquisition costs are significant, social factors prohibit the use of monetary transfers and make the combination of envy-freeness and Pareto efficiency a more attractive standard to judge a mechanism than its \citeauthor{kaldor}-\citeauthor{hicks} efficiency.   For example, in the allocation of dormitory rooms among college students, determining preferences may require costly visits to the relevant dormitories, but most schools allocate rooms by lotteries for purposes of fairness rather than by auction.

The most canonical form of such allocation is (random) serial dictatorship.  Individuals are ordered (uniformly randomly in the common envy-free form) and each individual has the pick of all items that remain upon that individual being reached.  In an environment without information acquisition a simple implementation is for individuals to order all items and the randomization to occur, and matches be made, through a centralized computer algorithm.  However, with information acquisition the sequencing seems likely to be more important, so individuals do not waste resources investigating objects that will not turn out to be in their choice sets.  It seems plausible that a mechanism that introduced actual sequencing and thus  clarified students' effective choice sets prior to their investigating their rooms could improve all students' utilities by an unbounded factor.\footnote{The implications in a two-sided setting may be even richer, as it may be desirable for students to ``accumulate offers'' before investigating any of them rather than to hold on to only a single offer as in \citet{galeshapley}'s celebrated deferred acceptance algorithm. }  In practice full sequencing may be too slow, but many schools do batch students into dormitory assignment rounds  with equal priority, allowing them a reasonably clear sense of their choice set.

Despite these benefits of sequencing, it is well-known that random serial dictatorship may be very inefficient \citep{agarwalsomaini} from an ex-ante perspective when compared with a market in probability shares in different objects such as that proposed by \citet{hz}.  To our knowledge the only implementation of the \citeauthor{hz} mechanism thus far  proposed involves the static reporting of preferences to a centralized algorithm.  Our analysis suggests that developing implementations based on a descending price process, with irrevocable purchases of probability shares, may have advantages if information acquisition is costly.\footnote{Kleinberg and Weyl have submitted a patent application on a ``Descending Price Auction for a Divisible Good'' that could potentially be extended to many goods to form the basis of such an implementation.}  Again this is an interesting direction for future research.

\section{Conclusion}\label{conclusion}

In this paper we argue that traditional market designs do a poor job coordinating sequential information acquisition while a simple uniform descending price design does better.  This conclusion runs largely contrary to existing literature which has focused on ascending price or simultaneous mechanisms \citep{crawfordknoer, milgromascending,ausubel}.  While our calibrations suggest that in many realistic cases the benefits of descending price mechanisms outweigh their disadvantages, our formalism neglects many important aspects of the information acquisition process  that could modify our conclusions.

Most importantly we assumed that information acquisition, while costly, occurs instantaneously, and thus that the process is not costly to other participants.  In reality, an economics job market fly-out takes a day of real time, and the meeting where interviews are cheaper occurs only over the course of three days.  As a result, the rich and full sequencing from the highest-value matches downward that we suggest is infeasible and/or undesirable in that setting.  It still seems plausible, and consistent with practice, that it would be optimal for some amount of clearing to take place sequentially from the top down, perhaps in stages.  However good mechanisms would have to trade-off sequencing against speed and a reasonable trade-off would be more complicated and parameter-dependent to strike.  In this context, therefore, our results offer more a qualitative insight and direction for future research rather than direct guidance on a market design.

Relatedly, we assume that all information arrives only {\em endogenously} in response to the expenditure of costs by bidders.  In reality, over the course of the time assignment takes place, information may arrive {\em exogenously}.  Creating a market that allows for sequencing of endogenous information may therefore (assuming a deadline) require that some matches are made before all exogenous information arrived.  This is a leading concern raised by \citet{gun} and \citet{niederle} and would tend to favor a market that was (nearly) simultaneous once the maximum exogenous information was revealed.  A model exploring trade-offs between the sequencing that promotes efficient endogenous information acquisition and the simultaneity that maximizes the value of exogenous information would therefore be important to determine an appropriate market design.

Finally, our work at present leaves a number of more technical ends loose.  For example, we have imposed a unit demand structure, which can be relaxed to a matroid and potentially also to some other forms of substitutability, and perhaps even to allow some complementarity at some corresponding cost to the quality of our bounds.

%\newpage{}
%
\bibliographystyle{ecta}
\bibliography{dynamic}

\begin{center}{\LARGE Appendix}
\end{center}

\appendix
In \Cref{lower}, we give proofs of negative results from \Cref{negative}.
In \Cref{mainapp}, we give the formal proof of \Cref{maintheorem} from \Cref{main}.
In \Cref{numericsapp}, we give a high-level overview of the numerical methods used in our calibrations.
In \Cref{extapp}, we give proofs of our results on extensions in \Cref{extensions}.
In \Cref{app:revenue}, we give some revenue guarantees for the Dutch auction.
In \Cref{sec:app-calibrations-extra}, we give a detailed description of the numerical methods used in our calibrations.
In \Cref{app:multistage}, we give full proofs of our results for multistage inspection.

\section{Negative Result Proofs}\label{lower}
\begin{proof}[Proof of Proposition \ref{prop:simultaneous-flaw}]
We show its welfare on any strategy profile can be arbitrarily small relative to the optimal.  Let $p_i$ denote the probability assigned by $i$'s strategy to inspecting and entering.
The probability that $i$ inspects and discovers $v_i = M$ is $\frac{p_i}{M}$, and is independent of all other bidders.
Total welfare is the sum of welfare from bidders who inspect before bidding and those who don't.
The welfare from the latter is at most the expected value minus cost of one of these bidders, which is $1-c$.
The welfare from inspecting bidders is at most
\begin{align*}
 \text{welfare}
  &\leq M \Pr[\text{some $i$ inspects and has $v_i = M$}] - c \sum_i p_i  \\
  &= M \left(1 - \prod_i \left(1 - \frac{p_i}{M}\right)\right) ~ - ~ c \sum_i p_i  \\
  &\leq M \left(1 - e^{-\sum_i p_i / (M-p_i)}\right) - c \sum_i p_i  \leq M \left(1 - e^{-\sum_i p_i / (M-1)}\right) - c \sum_i p_i. 
\end{align*}
We used the inequality $1-y \geq e^{-y/(1-y)}$ for $y < 1$, with $y = \frac{p_i}{M}$.
Now let $x = \sum_i p_i$, the expected number of bidders to inspect.
We observe that this welfare is maximized by choosing $x$ such that $c = \frac{M}{M-1}e^{-x/(M-1)}$.
As we will take $M \to \infty$, we simplify the exposition by supposing the maximizer be the $x$ where $c = e^{-x/M}$; this is $x = M\ln\frac{1}{c}$.
Then
\begin{align*}
 \text{welfare}
  &\leq M \left(1 - c\right) - c \cdot M\ln\frac{1}{c}  = M\left( 1 - c - c\ln\frac{1}{c} \right).
\end{align*}
This gives a total welfare of at most $M\left(1 - c - c\ln\frac{1}{c}\right) + 1-c$, and we are interested in its ratio to the optimal welfare, which approaches $M(1-c)$ as $n,M \to \infty$ as discussed.
This ratio approaches
\begin{align*}
 \frac{\text{second-price}}{\text{optimal}}
  &\to \frac{M\left(1 - c - c\ln\frac{1}{c}\right)}{M(1-c)} = \frac{1 - c -c\ln\frac{1}{c}}{1-c}  .
\end{align*}
This can be made arbitrarily small as $c$ is chosen arbitrarily close to $1$.
\end{proof}

\begin{proof}[Proof of Proposition \ref{prop:ascending-flaw}]
The ascending auction is formally specified as follows.
A global price $t$ is initially set to some value $t_0$ and increases continuously.
A bidder may instantaneously inspect at any time.
A bidder may choose to drop out at any time $t$, including the time of inspection after the inspection has occurred.
Other bidders observe this decision at all times $t' > t$.
After the second-to-last bidder drops out at time $t$, the last remaining bidder obtains the item and pays a price of $t$.
If the winning bidder has not yet inspected, she must pay the inspection cost  (as with all mechanisms in this paper).
If multiple bidders simultaneously choose to drop out at the same time $t$ leaving no one remaining, the item is awarded to one of those who dropped out at $t$, chosen uniformly at random, who pays price $t$.
We restrict to equilibria where, once a bidder inspects, they drop immediately if the current price exceeds their realized value, else they drop when the price reaches their realized value. (This is achieved by a trembling-hand or dominant-strategies refinement.)

All bidders are symmetric and have the following value distribution:
 \[ v = \begin{cases} M  & \text{with probability $p$}  \\
                      L  & \text{with probability $0.5$}  \\
                      0  & \text{otherwise.} \end{cases} \]
The cost $c$ of all bidders is equal to $\E v - L = p M - \frac{L}{2}$.
The strike price may be calculated by $p\left(M - \sigma\right) = c$, hence $\sigma = \frac{L}{2p}$.
(Observe that this calculation is correct as long as $\sigma > L$, which is true because $p < 0.5$.)
In particular, it will suffice for our purposes to set $M = \frac{1}{p^2}$ and $L = p$, and take $p \to 0$.
However, a variety of parameter ranges also would give the result.

The proof proceeds in two parts.
First, we show that, in equilibrium, no bidders wait until after time $L$ before inspecting.
Intuitively, this will follow because whichever bidders inspect last are getting negative expected value (the price exceeds their expected gain), causing an unraveling.
Second, we show that if all bidders inspect and/or drop before time $L$, welfare cannot be high.
Intuitively, this is almost identical to the simultaneous second price case, as bidders have no information about who has found or failed to find a black swan.

For the first part of the proof, assume for contradiction that we are in an equilibrium and there is a realization of the game where the price reaches some $t > L$ with $k \geq 1$ bidders having not yet inspected.
We prove by induction on the number of remaining bidders that all have negative expected utility for this outcome, giving the contradiction.
A key claim we will use, to be proven later, is: \emph{$(*)$ For any $k \geq 1$, if the last $k$ bidders inspect simultaneously, then they all have negative expected utility.}
In particular, if $i$ is the sole remaining bidder; then her expected utility is $\E v - c - t < \E v - c - L \leq 0$, so it is negative.

Now suppose there are $k \geq 2$ remaining bidders and there is a nonzero chance that some bidder $j$ is inspecting.
We claim that waiting until any later time $t' > t$ to act is strictly dominated for all other $i \neq j$.
If $j$ drops before $t'$, then there are $k-1$ remaining bidders and by inductive hypothesis all have negative expected utility.
If no bidders drop before $t'$, $i$ gets no higher utility for dropping at $t'$ than at $t$, and no higher utility for inspecting at $t'$ than at $t$.
Hence the only possible case where all bidders are best-responding is when they always inspect simultaneously at time $t$; in this case, claim $(*)$ says that all have negative expected utility.

Therefore, in equilibrium, no inspections occur at any time $t > L$: Thus conditioned on the time equalling $L$, it is strictly dominated to wait longer before deciding to inspect or to drop.

We now prove claim $(*)$, which relies on our choice of $M$, $L$, and $p$ (although many choices would suffice).
The case $k=1$ is covered above, so let $k \geq 2$.
The expected utility for any one bidder $i$ is upper-bounded as follows.
Bidder $i$ always pays the inspection cost $c$.
If any other bidder $j\neq i$ finds a black swan ($v_j = M$), then $i$'s net gain is zero.
Otherwise, if $i$ finds the black swan, $i$ gets $M-t$ (where $t$ is the current time, hence the price the winner will pay).
Otherwise (so nobody, including $i$, finds a black swan), $i$ has a $1/k$ chance of dropping out last, hence will pay price $t$ and gain her value.
We will lower-bound $t$ by $L$, as this only decreases the price $i$ pays when she wins and hence only increases her utility.
This allows us to drop the case where $i$ wins with value $v_i = L$, as the value cancels with the price paid.
Putting these observations together,
\begin{align*}
 \text{utility}
  &< (1-p)^{k-1}\left(p(M-L) + \frac{1}{k}\left(\frac{1}{2}-p\right)(-L)\right) ~ - ~ c  \\
  &= (1-p)^{k-1}\left(p(M-L) - \frac{L}{k}\left(\frac{1}{2}-p\right)\right) ~ - ~ pM + \frac{L}{2} .
\end{align*}
Note the strict inequality results from $t < L$.
In particular, plugging in $k=1$ gives that utility is strictly negative for a single bidder winning the item at time $t > L$ having not yet inspected.
Next, by taking the derivative of this bound on utility with respect to $k$, we get
\begin{align*}
  &= (1-p)^{k-1} \ln(1-p) \left(p(M-L) - \frac{L}{k}\left(\frac{1}{2}-p\right)\right) + (1-p)^{k-1} \frac{L}{k^2}\left(\frac{1}{2}-p\right)  \\
  &= (1-p)^{k-1} \left[ L\left(\frac{1}{2}-p\right)\left(\frac{1}{k^2} + \frac{1}{k}\ln\frac{1}{1-p}\right) - p(M-L)\ln\frac{1}{1-p}\right]  \\
  &\leq L\left(\frac{1}{2}-p\right)\left(1 + \ln\frac{1}{1-p}\right) - p(M-L)\ln\frac{1}{1-p}
\end{align*}
using that this expression is decreasing in $k$ and $k \geq 1$.
We can now use some rough bounds, particularly $p < \frac{1}{2}$ and $\ln\frac{1}{1-p} \geq p$, to say that the derivative of our bound on utility is at most
$$\leq L - p^2(M-L) \leq 0$$
if we take, for example, $M = \frac{1}{p^2}$ and $L = p$.
This ensures that the derivative of the bound is always negative, hence the bound is decreasing in $k$.
But the bound on utility was already $0$ for $k=1$, hence utility is negative for all $k$.
This completes the proof of claim $(*)$.

For the second part of the proof, we now know that in equilibrium all inspections occur at or before time $L$.
We can use a similar strategy to the second price case.
Let $x$ be the expected number of bidders who inspect in equilibrium.
Because each bidder who inspects has a $\frac{1}{2}+p$ chance of $v_i \neq 0$, the expected number of bidders with $v_i \neq 0$ is exactly $x\left(\frac{1}{2} + p \right)$.
We claim that conditioned on $i$ inspecting and discovering $v_i \neq 0$, in equilibrium at time $\leq L$, the probability that $v_i = M$ is exactly $p/\left(\frac{1}{2}+p\right)$ independent of everything else.
This follows because the cases $v_i = M$ and $v_i = L$ are indistinguishable up to time $L$ ($i$ always stays in).
Hence, given that $k$ bidders inspect and find $v_i \neq 0$, the probability that none of them has $v_i = M$ is exactly  $\left(1 - \frac{p}{\frac{1}{2}+p}\right)^k$. 
Therefore, in a realization of the game where exactly $\hat{x}$ bidders inspect and $k$ of them find $v_i \neq 0$, expected welfare is at most
\begin{align*}
 &M \left(1 - \left(1 - \frac{p}{\frac{1}{2}+p}\right)^k\right) + L - \hat{x} \cdot c .
\end{align*} \\
Now, use that $1-y \geq e^{-y/(1-y)}$ for $y < 1$, letting $y = \frac{p}{\frac{1}{2}+p}$.
(Here, $y/(1-y) = 2p$.)
\begin{align*}
  &\leq M\left(1 - e^{-2pk}\right) + L - \hat{x} \cdot c . 
\end{align*}
Now, we are using $x$ to denote the expected number of inspections, so $x = \E \hat{x}$.
Meanwhile, using that $1-e^{-x}$ is a convex function and the constraint that $\E k = x\left(\frac{1}{2}+p\right)$, we get that expected welfare is at most
$$\E M \left(1 - e^{-2p k} \right) + L - x \cdot c  \leq M \left(1 - e^{-2p \E k} \right) + L - x \cdot c  = M \left(1 - e^{-x(p + 2p^2)}\right) + L - x \cdot c .$$
It will simplify exposition of this proof to write this bound more roughly as
 \[ M \left(1 - e^{-xp}\right) - x \cdot c  + \littleo(\sigma) . \]
This will not change the conclusion because the optimal welfare approaches $\sigma$.\footnote{To see this, note that optimal welfare is the maximimum covered call value $\kappa_i = \min\{v_i,\sigma\}$; with enough bidders, almost certainly some bidder has $\kappa_i = \sigma$.}
This simplification holds because, first, $p \to 0$ in our example and the $2p^2$ term has a negligible effect; and second, the optimal welfare is $\sigma$ and $L = 2p\sigma = \littleo(\sigma)$, so the additive $L$ does not change the welfare ratio.

As in the simultaneous second price case, we can maximize this expression over all $x$.
Recall that $c = p(M-\sigma)$; the maximum occurs when $\sigma = M\left(1 - e^{-px}\right)$, or $x = \frac{1}{p}\ln\frac{1}{1-\frac{\sigma}{M}}$.
So
 \[ \text{welfare} \leq \sigma - (M-\sigma)\ln\frac{1}{1-\frac{\sigma}{M}} + \littleo(\sigma) , \]
whereas the optimal welfare approaches $\sigma$ as the number of bidders diverges.
Thus the ratio approaches
\begin{align*}
 \frac{\text{welfare}}{\text{optimal}}\leq 1 - \left(\frac{M}{\sigma} - 1\right) \ln\frac{1}{1-\frac{\sigma}{M}} + \littleo(1) =    1 - \frac{M}{\sigma} \ln\frac{1}{1-\frac{\sigma}{M}} + \littleo(1)  \to 0 \text{ as $\frac{M}{\sigma} \to \infty$.}
\end{align*}
Recall that in particular we chose $M = \frac{1}{p^2}$ and $L = p$, which gave $\sigma = \frac{1}{2}$.
This satisfies the requirements we needed of $\frac{L}{\sigma} \to 0$ and $\frac{\sigma}{M} \to 0$ as $p \to 0$.
\end{proof}

\begin{proof}[Proof of Proposition \ref{prop:sequential-flaw}]
The optimal policy is for the $R$ types in the population to sequentially try to find a black swan and only allow an $S$ type to take the object once all of the option value of the $R$ types has been exhausted, or one has found a black swan. As $n$ grows large enough that $n\pi p_R$ is large, this policy achieves welfare of approximately $\frac{v_R}{2}$, as the cost $c_R$ must be paid roughly $\frac{1}{p_R}$ times before a black swan of value $v_R$ is found.  On the other hand any policy in which a type $S$ investigates prior to any investigation by $R$ types cannot achieve welfare of more than $\left(1-p_S\right)\epsilon+p_S \frac{v_R}{2}\approx \epsilon$ because $v_S>\frac{v_R}{2}$, so that once a safe option has been successfully investigated the cost of this investigation is already sunk and thus it is not worth investigating the risky option.

However, it is easy to see that in any sequential equilibrium of the BKRS procedure this is exactly what happens with probability $1-\pi$.  In particular the seller cannot distinguish between $R$ and $S$ types and thus will draw an $S$ type first with probability $1-\pi$.  Such a type will earn surplus $\left(1-p_S\right)\epsilon-r$ from investigating her value and bidding $r$ to become the incumbent, assuming that this leads her to win with probability $1$.  But it will do so, because any future buyer will know that if she enters bidding will proceed, in the knockout auction, at least up to $v_S>\frac{v_R}{2}, \epsilon$, wiping out any gains that either type $R$ or $S$ could gain from investigating their values in a later round.  Thus the social welfare of the BKRS procedure can be arbitrarily close to a fraction $0$ of the first-best social welfare.

\end{proof}

%%%%%%%%%%%%%%%%%%%%%%%%%%%%%%%%%%%%%%%%%%%%%%%%%%%%%%%%%%%%%%%%%
\section{Equivalence Theorem}\label{mainapp}

This section fleshes out the proof sketch
of \Cref{maintheorem} provided in \Cref{main}.
It will be helpful to express the proof in terms
of a sample space with $2n$ bidders, whose types and
values are jointly distributed as follows.
Bidders $1,\ldots,n$ have types
$\{\theta_i\}_{i=1}^n$ distributed according to $\mathcal P$,
and bidder $i$ for $i = n+1,\ldots,2n$ has the covered
counterpart type $\ccpt{\theta}_{i-n}$. The values
$v_1,\ldots,v_{2n}$ are coupled so that
$v_{i+n} = \kappa_i$ for $i=1,\ldots,n$
(such a coupling is possible because $\theta_{i+n}=\ccpt{\theta}_i$)
and the pairs $\{(v_i,v_{i+n})\}_{i=1}^n$ are mutually independent.

\begin{claim} \label{clm:1}
Against any profile of strategies for bidders $k+1,\ldots,k+n-1$,
any best response of bidder $k+n$ is functionally equivalent to
a normalized strategy.
\end{claim}
\begin{proof}
Since bidder $k+n$ can inspect her value without cost, each of her strategies
is functionally equivalent to one in which she inspects her value at the earliest
possible moment. Thus, up to functional equivalence, a strategy for bidder $k+n$
is completely described by a function $b(v_{k+n},\theta_{k+n})$ indicating the 
price at which she
will claim the item, given that her value is $v_{k+n}$, her type is $\theta_{k+n}$, 
and no other bidder has yet claimed the item. Furthermore,
if $v < v'$ but $b(v,\theta_{k+n}) > b(v',\theta_{k+n})$,
then $b$ cannot be a best response to a strategy profile of
bidders $k+1,\ldots,k+n-1$, unless both of the bids 
$b(v,\theta_{k+n}), \, b(v',\theta_{k+n})$ have zero
probability of winning against that strategy profile. 
In that case $b$ is functionally
equivalent to a strategy that always bids
zero when the probability of winning is zero,
and the latter strategy is represented by a 
monotonically non-decreasing bid function.
\end{proof}

In the next claim and the following one, 
we say a strategy of bidder $k$ ``almost surely exercises
in the money'' if the event that bidder $k$ inspects her value,
finds that $v_k > \sigma_k$, yet does not claim the item, has
probability zero.
\begin{claim} \label{clm:2}
Fix any strategy profile $b_{-k}$ for bidders $k+1,\ldots,k+n-1$.
If $b_k$ is any strategy for bidder $k$ that almost surely exercises in the money,
then $b_k$ and $\lambda(\mu(b_k))$ are functionally equivalent. If $b_{k+n}$
is a normalized strategy for bidder $k+n$ then $b_{k+n}$ and
$\mu(\lambda(b_{k+n}))$ are functionally equivalent.
\end{claim}

\begin{proof}
\newcommand{\insp}{{\mathrm{insp}}}
From the definitions of $\lambda$ and $\mu$ it is clear that,
if bidder $k$ obtains the item at price $t$ when using
strategy $b_k$, then $k+n$ also obtains the item at price $t$
when using $\mu(b_k)$ to simulate $b_k$, and that $k$ obtains the
item at price $t$ when using $\lambda(\mu(b_k))$ to simulate
$\mu(b_k)$. 

To complete the proof that 
$b_k$ and $\lambda(\mu(b_k))$ yield functionally equivalent outcomes,
we must prove that bidder $k$ inspects the item under strategy $b_k$
if and only if she inspects the item under strategy $\lambda(\mu(b_k))$. 
Let $b_k^\insp(\theta_k)$ denote the 
clock value at which bidder $k$ inspects the item when using
strategy $b_k$; by our assumption that $b_k$ 
almost surely exercises in the money, we know that it 
always claims the item at price $b_k^\insp(\theta_k)$ 
if $v_k > \sigma_k$. Therefore, when bidder $k+n$ uses
strategy $b=\mu(b_k)$ to simulate $b_k$, she must always claim the
item at price $b_k^\insp(\theta_k)$ if $v_k > \sigma_k$. 

Now, recall that the mapping $\lambda$ is defined in terms of the
function $b(v)$ that specifies the price at which bidder $k+n$ 
claims the item if $v_{k+n}=v$. When $v_k > \sigma_k$ we have
$v_{k+n}=\sigma_k$, so according to the reasoning in the previous paragraph
it must be the case that $b(\sigma_k) = b_k^\insp(\theta_k)$. By the 
definition of $\lambda$, then, it follows that strategy $\lambda(b)=\lambda(\mu(b_k))$
inspects the item when the clock value is $b_k^\insp(\theta_k)$, exactly
as $b_k$ does. This completes the proof that $b_k$ and 
$\lambda(\mu(b_k))$ are functionally equivalent.

 Since bidder $k+n$ has
zero inspection cost, this immediately implies that 
$b_{k+n}$ and $\mu(\lambda(b_{k+n}))$
are functionally equivalent.
\end{proof}

\begin{claim} \label{clm:3}
Fix any strategy profile $b_{-k}$ for bidders $k+1,\ldots,k+n-1$.
Let $b_k$ be any strategy of bidder $k$ and let $b_{k+n}$ be any normalized
strategy of bidder $k+n$. Denoting by $u_k, u_{k+n}$ the expected
utilities of bidders $k$ and $k+n$, we have
\begin{align}
\label{eq:mainthm-3.1}
  u_k(\lambda(b_{k+n}), b_{-k}) & = u_{k+n}(b_{k+n},b_{-k}) \\
\label{eq:mainthm-3.2}
  u_k(b_k,b_{-k}) & \leq u_{k+n}(\mu(b_k),b_{-k}).
\end{align}
Furthermore, equality is attained in Equation~\eqref{eq:mainthm-3.2}
if and only if $b_k$ almost surely exercises in the money.
Best responses of bidder $k$ to $b_{-k}$ almost surely
exercise in the money.
\end{claim}

\begin{proof}
For $i \in \{k,k+n\}$ let $\allocsubi$ and $\inspecti$ denote the indicator random variables
of the events that bidder $i$ inspects her value in a Dutch auction 
against bidders $k+1,\ldots,k+n-1$ playing strategy profile $b_{-k}$,
and let $t_i$ denote the amount that $i$ pays. We have
\begin{align*}
  u_k &= \allocsub{k} v_k - \inspect{k} c_k - t_k \\
  u_{k+n} &= \allocsub{k+n} v_{k+n} - \inspect{k+n} c_{k+n} - t_{k+n} 
                   = \allocsub{k} \kappa_k - t_k 
\end{align*}
where we have used the facts that $\allocsub{k+n}=\allocsub{k}$ and
$t_{k+n}=t_k$, along with $v_{k+n}=\kappa_k$ (by 
our construction of the coupling) and $c_{k+n}=0$ (by the definition
of the covered counterpart). \Cref{lemma:upper-bound-policy} now
implies all of the conclusions stated in the claim, except for the statement
that best responses of bidder $k$ to $b_{-k}$ almost surely exercise in 
the money. To prove that statement, we reason as follows.
Let $u_k^*$ and $u_{k+n}^*$ denote the 
maximum expected utility attainable by bidders $k$ and $k+n$, respectively,
when bidding against strategy profile $b_{-k}$.
Equation~\eqref{eq:mainthm-3.1} implies $u_k^* \geq u_{k+n}^*$,
while Equation~\eqref{eq:mainthm-3.2} implies $u_k^* \leq u_{k+n}^*$,
therefore the two must be equal to one another. Furthermore, any strategy
of bidder $k$ that does not almost surely exercise in the money fails to 
attain equality in Equation~\eqref{eq:mainthm-3.2} and therefore must
fail to be a best response. Contrapositively, every best response of bidder
$k$ must almost surely exercise in the money.
\end{proof}

\begin{claim} \label{clm:4}
The functions $\lambda$ and $\mu$ induce mutually inverse bijections
on functional equivalence classes of best responses for bidders $k$
and $k+n$.
\end{claim}

\begin{proof}
First we must show that $\lambda$ and $\mu$ preserve the relation
of functional equivalence.
If $b_{k+n}$ and $b'_{k+n}$ are functionally equivalent 
normalized strategies of bidder $k+n$
then $\lambda(b_{k+n})$ and $\lambda(b'_{k+n})$ are 
functionally equivalent because this transformation leaves other bidders' behavior unchanged and bidder $k$ inspects at the same clock value no matter whether she
is using $\lambda(b_{k+n})$ or $\lambda(b'_{k+n})$; in fact this clock
value is equal to the value at which bidder $k+n$ claims the item if her 
value is $\sigma_k$ and she is using either $b_{k+n}$ or $b'_{k+n}$. 
If $b_k$ and $b'_k$ are  functionally equivalent strategies of bidder $k$
then $\mu(b_k)$ and $\mu(b'_k)$ are functionally equivalent as the transformation leaves other bidders' behavior unchanged.

Having established that $\lambda$ and $\mu$ are well-defined
mappings between functional equivalence classes, we know from 
\Cref{clm:2} that the composition $\lambda \circ \mu$ is equal to 
the identity on the set of functional equivalence classes of strategies that
almost surely exercise in the money, and that $\mu \circ \lambda$
is equal to the identity on the set of functional equivalence classes of normalized
strategies. Those two sets contain, respectively, the sets of best responses
of bidders $k$ and $k+n$, by an application of \Cref{clm:1} (for bidder $k+n$)
and \Cref{clm:3} (for bidder $k$). 
\end{proof}
 
To conclude the proof of \Cref{maintheorem},
from \Cref{clm:4} we know that for any fixed
profile of strategies $b_{-k}$ of bidders $k+1,\ldots,k+n-1$, there is
a bijection between functional equivalence classes
of best responses of bidders $k$ and $k+n$; in particular this holds
when $b_{-k}$ represents the profile of strategies used by bidders
$k+1,\ldots,k+n-1$ in an equilibrium of the Dutch auction with either
of the bidder sets $\{k,k+1,\ldots,k+n-1\}$ or $\{k+1,k+2,\ldots,k+n\}$.
The functional equivalence classes of
equilibria of the Dutch auction with these two bidder sets are
therefore in bijective correspondence. Furthermore,  this correspondence preserves the auctioneer's
revenue and the utility of each bidder except possibly
bidders $k$ and $k+n$.
% , as the transformation leaves the actions and allocations of all other bidders are unchanged.
\Cref{clm:3} ensures that those bidders'
expected utility is also preserved.
% , since we know (via the proof
%  of \Cref{clm:4}) that best responses of bidder $k$ almost surely
%  exercise in the money.
Composing together these equilibrium
correspondences for $k=1,\ldots,n$, we finally obtain the mapping
from equilibria of the Dutch auction with bidders $1,\ldots,n$
(whose types are jointly distributed according to ${\mathcal P}$)
to equilibria of the Dutch auction with bidders $n+1,\ldots,2n$
(whose types are jointly distributed according to $\ccpt{\mathcal P}$)
as asserted in the theorem statement.

\section{Numerical Methods}\label{numericsapp}

Here we give a high-level overview of our numerical methods, with more details given in \Cref{sec:app-calibrations-extra}.

In the start-up calibration of Subsection \ref{startup} we must solve, in each scenario we study, for the expected welfare of three mechanisms: the first-best, the Dutch auction and the simultaneous second-price auction.  To do so, for every parameter setting, we repeatedly:
\begin{enumerate}
\item Draw  a value of $\left\{\left(V_i^0,C_i^0\right)\right\}_{i=1}^N$.
\item Construct a distribution of covered call values for each bidder based on these.
\item Solve for the equilibrium of the Dutch auction using Richard Katzwer's ``AuctionSolver'' software package \citep{AuctionSolver} to find the equilibrium of the asymmetric first-price auction with bidder values distributed according to covered call values. By \Cref{maintheorem} we know the equilibrium welfare in this environment
is equal to that of the Dutch auction in our environment.
\item Solve for the equilibrium of the second-price auction using smoothed best response iteration and some analytical tricks through a program we wrote that we describe in \Cref{sec:app-calibrations-extra}.
\item Sample $\left\{ \left(V_i^1,C_i^1,V_i^2,C_i^2\right)\right\}_{i=1}^N$ repeatedly, calculate the corresponding covered call values, and feed this information to the above-calculated equilibria to obtain average welfare for each mechanism, as well as the first-best (which is just the highest covered call value).
\end{enumerate}
We then average values over all these samples to construct overall Monte Carlo estimates of our three objects of interest, though the number of sampling iterations is small given that each iteration requires a full solution for equilibrium and thus in practice we consider the average of only three auctions generated by draws of $\left\{\left(V_i^0,C_i^0\right)\right\}_{i=1}^N$ for each set of parameter values rather than a precise estimate of average welfare.  This allowed us only to get precision to a single significant digit  and thus we report estimates rounded to this accuracy.  We provide further details of our numerical methods in \Cref{sec:app-calibrations-extra}.

   Solving for the equilibrium of the BKRS mechanism is fairly complex, likely one reason that \citeauthor{sweeting} consider a  set-up simpler than ours in the previous calibration.  Luckily, \citeauthor{sweeting} report the welfare of the simultaneous second-price auction and the BKRS mechanism, meaning that we only need to compute the first-best welfare and that of the descending auction.  Furthermore because there are only two types of bidders, the outer loop of our numerical method above is unnecessary.  We thus simply construct the distribution of covered call values and the corresponding first-price auction equilibrium using AuctionSolver  and then sample from the covered call value distribution to construct both first-best welfare and welfare under the Dutch auction.

\section{Extensions}\label{extapp}

\subsection{Many objects}\label{multiobjectapp}

Our argument uses an approach from the literature on smoothness in computational mechanism design \citep{lucier,roughgarden, syrgkanise, syrgkanis, hartline}.  This relies on showing inequalities based on a particular, simple deviation that each bidder could make relative to equilibrium play.  For player $i$,
the deviation strategy $b_i'$:
\begin{enumerate}
\item Samples a random number $r \in [e^{-2},1]$ with density $f(r)= \frac{1}{2r}$. For an object with strike price $\sigma$ and value $v$,
define the {\em shaded strike price} and {\em shaded value} to be $(1-r)\sigma$ and $(1-r)v$, respectively.
%\label{step:r}

\item Takes no action at any time when the clock value $t$ is strictly greater than the highest remaining shaded strike price of an uninspected object and the highest remaining shaded value of an inspected object.

\item When $t$ becomes weakly less than either of these, the deviation claims any object that has been inspected and has shaded value of $t$ and if no such object exists it inspects, in random order, any object with shaded strike price $t$, claiming that object if its shaded value is weakly above $t$ and otherwise continuing inspection until all objects of shaded strike price $t$ have been inspected. 

\end{enumerate}

\begin{lem} \label{lem:deviation-non-exposed}
The deviation strategy $b_i'$ always exercises in the money. It acquires
item $j$ at price $(1-r) \kappa_{ij}$ if the item is still available,
has non-negative value, and bidder $i$ has not yet acquired another item; 
otherwise it does not acquire item $j$.
\end{lem}
\begin{proof}
By construction, any object that is inspected and has a value above its strike price is immediately claimed,
i.e.~$b'_i$ always exercises in the money. The price at which an object is claimed is, by construction, 
equal to the lesser of its shaded strike price and its shaded value, i.e.~$(1-r) \kappa_{ij}$.
The only circumstance that prevents $i$ from claiming $j$ at this price is the event that 
either $i$ already claimed another object, or that $j$ was already sold to another bidder.
% Because the clock terminates at time $0$, the strategy never claims an object of strictly negative value.
\end{proof}
 
To state the key 
lemma that underpins our analysis of the equilibria of the descending-clock
auction, we introduce the following notation. If $b$ is a profile of strategies, 
$p_i(b)$ denotes the price paid by bidder $i$ in the descending-clock auction,
and $p^j(b)$ denotes the price paid for item $j$.
% , and $p_{ij}(b)$ denotes the amount
% that $i$ pays for $j$, which is zero if $i$ does not receive $j$ and $p_i(b)$ otherwise. 
Finally, $\kappa_i(b)$ denotes the covered call value of the item $i$ receives
in the descending-clock auction, or $\kappa_i(b)=0$ if there is no such item;
in other words, $\kappa_i(b) = \sum_{j \in \itemset} \allocij(b) \kappa_{ij}$.
All of these quantities should be interpreted
as random variables on the sample space defined by the random realization of costs,
types, and values.

\begin{lem} \label{lemma:general-smoothness}
For any $i,j$, any strategy profile $b$, and any realizations of all types and values,
 \[ \E \left[ \kappa_i \left(b_{-i},b_i' \right) - p_{i} \left(b_{-i},b_i' \right) \mid \{(\theta_j,v_j)\}_{j=1}^n \right] + \tfrac12 p_i(b) 
+ \tfrac12 p^j(b) \geq \tfrac{1}{2} \left(1 - e^{-2} \right) \kappa_{ij} . \]
(The conditional expectation on the left side integrates over the random choice of $r$
in strategy $b'_i$.)
\end{lem}
\begin{proof}
The structure of the deviation strategy $b_i'$ guarantees that 
$\kappa_i(b_{-i},b_i') - p_i(b_{-i},b_i') = r \kappa_i(b_{-i},b_i')$ pointwise.
When $i$ does not receive an item this is because both sides are equal to zero;
when $i$ receives an item it is because strategy $b_i'$ dictates that $i$ always
pays $1-r$ times the covered call value for the item it acquires.
Accordingly, the inequality asserted by the lemma is equivalent to
\begin{equation} \label{eq:gs.0}
  {\E}_r \left[ r \kappa_i(b_{-i},b_i') \right] + \tfrac12 p_i(b)
  + \tfrac12 p^j(b) \geq \tfrac{1}{2} \left( 1 - e^{-2} \right) \kappa_{ij}, 
\end{equation}
where we have introduced the notation ${\E}_r[\cdot]$ as
shorthand for $\E[\cdot \mid \{(\theta_j,v_j)\}_{j=1}^n]$.

%If $p(b) > \left( 1 - e^{-2} \right) \kappa_{ij}$,
%then at least one of the last two terms on the left side
%of~\eqref{eq:gs.0} is greater than the right side; the 
%inequality follows because all three terms on the left are
%non-negative. 
%
%Let $t(r)$ denote the
%clock value at which bidder $i$ obtains an item in
%the outcome of strategy profile $(b_{-i},b'_i)$ when
%mixed strategy $b'_i$ samples $r$ as the random number
%selected in its initial step. Since the price paid by $b'_i$
%is always $1-r$ times the covered call value of the item
%acquired, we know that 
%$\kappa_i(b_{-i},b'_i) = r t(r)$.

%If $p(b) \leq \left( 1 - e^{-2} \right) \kappa_{ij}$,
%then in equilibrium, when the clock value reaches 
%$\left( 1 - e^{-2} \right) \kappa_{ij}$,
%item $j$ is unclaimed and bidder $i$ has not claimed any item.

\begin{claim} \label{clm:gs1}
Let $p(b) = \max \{ p_i(b), p^j(b)\}$. 
If the random number $r$ sampled in the
initial step of deviation strategy $b'_i$ 
satisfies $r < 1 - p(b)/\kappa_{ij}$ then
$\kappa_i(b_{-i},b'_i) \geq \kappa_{ij}$.
\end{claim}
\begin{proof}
The proof is by contradiction: if $\kappa_i(b_{-i},b'_i) <
\kappa_{ij}$ then player $i$ did not claim item $j$ 
when the clock value was at $(1-r) \kappa_{ij}$,
which could only mean that some other bidder $i'$ had already
claimed item $j$ at that time. However, the inequality
$r < 1 - p(b)/\kappa_{ij}$ implies 
$(1-r) \kappa_{ij} > p(b) \geq p_i(b)$,
meaning that the behavior of bidder $i$ in strategy profile
$(b_{-i},b'_i)$ remains indistinguishable from her behavior
in the equilibrium profile $b$ until after the clock value
descends below $(1-r) \kappa_{ij}$. Consequently, the
bidder $i' \neq i$ who claims item $j$ at a price greater than or 
equal to $(1-r)\kappa_{ij}$ when the strategy profile is
$(b_{-i},b'_i)$ must also do so when the strategy profile
is $b$. This implies $p^j(b) \geq (1-r) \kappa_{ij}$, contrary
to our assumption that $p(b) < (1-r) \kappa_{ij}$. 
\end{proof}

Let $\rho = p(b)/\kappa_{ij}$.
\Cref{clm:gs1} justifies the inequality
\begin{equation} \label{eq:gs.1}
    {\E}_r \left[ r \kappa_i(b_{-i},b_i') \right] \geq
    \left( \int_{e^{-2}}^{1 - \rho} (r \kappa_{ij}) \, f(r) \, dr \right)^+ =
    \left( \int_{e^{-2}}^{1-\rho} \frac{\kappa_{ij}}{2} \, dr \right)^+ =
    \tfrac12 \left( 1 - e^{-2} - \rho \right)^+ \kappa_{ij}.
\end{equation}
If $\rho < 1-e^{-2}$ then $1 - e^{-2} - \rho$ is positive, 
so we obtain
\begin{align}
  {\E}_r \left[ r \kappa_i(b_{-i},b_i') \right] \geq 
  \tfrac12 \left( 1 - e^{-2} \right) \kappa_{ij} - \tfrac12 \rho \kappa_{ij} = 
  \tfrac12 \left( 1 - e^{-2} \right) \kappa_{ij} - \tfrac12 p(b).
\end{align}
In light of the fact that $p(b) = \max\{p_i(b),p^j(b)\} \leq p_i(b) + p^j(b)$,
we conclude that Inequality~\eqref{eq:gs.0} holds in this case.
On the other hand, if $\rho \geq 1 - e^{-2}$, then
$p(b) \geq \left( 1 - e^{-2} \right) \kappa_{ij}$,
so at least one of the last two terms on the left side
of Inequality~\eqref{eq:gs.0} is greater than or equal
to the right side. Since all three terms on the 
left side are non-negative, we see that~\eqref{eq:gs.0} holds
in this case as well. 
\end{proof}

\begin{proof}[Proof of  \Cref{manyobjects}]
%The proof blends the smoothness proof of \autoref{thm:single-item-poa} 
%with the dual-variable accounting scheme that underpins the proof
%of \autoref{thm:greedy-approx-opt}.
Let $c^b_{ij}$ denote the inspection cost paid by bidder $i$ for 
item $j$ in strategy profile $b$.
Let $u_i(b) = \sum_j [\allocij(b) v_{ij} -  c^{b}_{ij} - p_{ij}(b)]$
denote bidder $i$'s utility in the outcome of strategy profile $b$.
Then, in any equilibrium $b$,
\begin{equation} \label{eq:gd.1}
 \E\left[ \text{welfare}\right] = \E \left[ \sum_i u_i(b) + p_i(b) \right]  
          \geq \E \left[ \sum_i u_i(b_{-i},b_i') + p_i(b) \right] 
\end{equation}
because each $i$ prefers the equilibrium strategy $b_i$ to the deviation $b_i'$. 
Now, the deviation strategy $b_i'$ always exercises in the money, so 
\begin{align*} % \nonumber
 \E\left[ u_i(b_{-i},b_i') \right] & = 
 \E\left[ \sum_j \left(\allocij(b_{-i},b_i') v_{ij} 
                             - c^{(b_{-i},b_i')}_{ij}  
                             - p_{ij}(b_{-i},b_i') 
                      \right) 
    \right] \\ & =
 \E\left[ \sum_j \allocij(b_{-i},b_i') \kappa_{ij} -  p_{ij}(b_{-i},b_i') \right] =
 \E\left[ \kappa_i(b_{-i},b_i') - p_i(b_{-i},b_i') \right].
% \label{eq:gd.2}
\end{align*}
Substituting this into~\eqref{eq:gd.1}, we obtain
\begin{equation*}
   \E\left[ \text{welfare}\right] \geq \E \left[ \sum_i (\kappa_i(b_{-i},b_i') - p_i(b_{-i},b_i')) + \sum_i p_i(b) \right].
\end{equation*}
The sum of prices paid by bidders equals the sum of prices paid {\em for} items, so we may 
rewrite $\sum_i p_i(b)$ as $\frac12 \sum_i p_i(b) + \frac12 \sum_j p^j(b)$, yielding
\begin{align} \label{eq:gd.2}
   \E\left[ \text{welfare}\right] 
   &\geq 
   \E \left[ \sum_i (\kappa_i(b_{-i},b_i') - p_i(b_{-i},b_i')) 
             + \frac12 \sum_i p_i(b) + \frac12 \sum_j p^j(b) \right].
\end{align}
At this point we may define random variables 
\begin{align*}
    \bdualvar_i &= \E \left[ \kappa_i(b_{-i},b_i') - p_i(b_{-i},b_i') 
                              \mid \{ (\theta_j, v_j) \}_{j=1}^n \right] + \tfrac12 p_i(b) \\
    \idualvar_j &= \tfrac12 p^j(b)
\end{align*}
and rewrite~\eqref{eq:gd.2} as
\begin{equation} \label{eq:gd.3}
  \E\left[ \text{welfare} \right] \geq \E \left[ \sum_i \bdualvar_i + \sum_j \idualvar_j \right].
\end{equation}
\Cref{lemma:general-smoothness}
ensures that for all bidders $i$ and items $j$, the inequality
\begin{equation} \label{eq:gd.4}
  \bdualvar_i + \idualvar_j \geq \tfrac12 \left( 1 - e^{-2} \right) \kappa_{ij}
\end{equation}
holds pointwise. For the first-best policy, define $\optalloc_{ij}$ and $\optinspect_{ij}$
as before to be, respectively, the indicators of 
of the events that bidder $i$ receives item $j$ and that $i$ inspects $j$. 
We know from Lemma~\ref{lemma:upper-bound-policy} that
\begin{equation} \label{eq:gd.5}
   \E\left[ \text{first-best} \right] = 
   \sum_{ij} \E\left[ \optalloc_{ij} v_{ij} - c^*_{ij} \right] \leq
   \sum_{ij} \E \optalloc_{ij} \kappa_{ij}.
\end{equation}
Using~\eqref{eq:gd.4} and the fact that the first-best policy allocates at 
most one item to each bidder and at most one bidder to each item, we have
\begin{align} \nonumber
  \tfrac12 \left( 1 - e^{-2} \right) \sum_{ij} \E \optalloc_{ij} \kappa_{ij} & \leq
  \E\left[ \sum_{ij} \optalloc_{ij} (\bdualvar_i + \idualvar_j) \right] \\
  \nonumber & =
  \E\left[ \sum_i \bdualvar_i \left( \sum_j \optalloc_{ij} \right) \right] +
  \E\left[ \sum_j \idualvar_j \left( \sum_i \optalloc_{ij} \right) \right] \\
  & \leq
  \E\left[ \sum_i \bdualvar_i + \sum_j \idualvar_j \right].
\label{eq:gd.6}
\end{align}
The theorem follows immediately, by combining~\eqref{eq:gd.3} with~\eqref{eq:gd.6}.
\end{proof}

Intuitively there are two reasons why the approximation (often referred to as the ``price of anarchy'' in the relevant literature) we obtain to optimal welfare is worse in the general than in the single-item case (point 2 of Corollary \ref{maincorollary}).  First, the greedy assignment algorithm which a descending auction approximately coordinates is not optimal, only approximately so.  Second, with multiple items there are more cases to be considered to ``cover'' the potential welfare impacts of a deviation, as deviations may not only change {\em whether} an individual is assigned an object but also {\em which} object she is assigned.\footnote{Rather than the two cases 
that arise in the smoothness proof of, for example, \citeauthor{lucier}
(either bidder $i$ achieves a high expected utility when deviating,
or some bidder pays a high price in equilibrium)
there are now {\em three} cases to consider:
either bidder $i$ achieves a high expected utility when deviating,
or some bidder pays a high price for item $j$ in equilibrium,
or bidder $i$ pays a high price for some {\em other} item in
equilibrium.}  
When one combines these two sources of inefficiency, they cause our 
price-of-anarchy bound
to degrade from $1-\frac{1}{e} \approx 0.63$ in the 
single-item case to 
$\frac{1}{2} \left( 1 - e^{-2} \right) \approx 0.43$
in the general case with multiple items being sold.
Note, however, that in the multiple-item case the
greedy procedure is only guaranteed to attain at least
half of the first-best welfare, so the fact that every
equilibrium of the uniform descending procedure
attains at least $43\%$ of the first-best welfare can
be interpreted as indicating that the efficiency loss due
to strategic behavior is quite mild.

\subsection{Approximate best responses} \label{app:approxresponse}

%The proofs of 
%\Cref{prop:bounded-rationality,commonprop,prop:optional}
%all use smoothness arguments analogous to the
%proof of \Cref{manyobjects} in \Cref{multiobjectapp}.
%First we describe the deviation strategies that must be 
%analyzed in each case. Then we give the proofs.
%When agents play $\alpha$-approximate best 
%responses (\Cref{prop:bounded-rationality}) 
%the deviation
%strategy $b'_i$ is exactly the same as described in 
%\Cref{manyobjects}, except that the
%random number $r$ used to define the
%shaded strike price and shaded value
%has density $f(r) = \frac1{\alpha r}$ and support
%$[e^{-\alpha},1]$ in the single-item case,
%or density $f(r) = \frac{1}{2 \alpha r}$ and
%support $[e^{-2 \alpha}, 1]$ in the 
%multiple-item case. 
%For the common values environment
%(\Cref{commonprop})
%the deviation strategy $b'_i$ sets
%$w = \E[v_i^+] - c_i$, 
%samples $r \in [e^{-1},1]$
%with density $f(r) = \frac1r$, 
%inspects the item when the descending 
%clock is at $(1-r)w$, and claims the item
%immediately upon inspection unless its value
%is negative. Finally, when inspection
%is optional (\Cref{prop:optional}), the deviation 
%samples $r \in [e^{-1},1]$ with 
%density $f(r) = \frac1r$ and then
%selects one of the following two
%alternatives with equal probability:
%use shaded strike price $(1-r) \sigma_i$
%and shaded value $(1-r) v_i$ as in 
%the proof of \Cref{manyobjects},
%or acquire the item without inspection
%as soon as the descending clock reaches
%$(1-r) \E[ v_i \mid \theta_i ]$.

\begin{proof}[Proof of \Cref{prop:bounded-rationality}]
We present the proof of the single-item case, and we indicate how to
modify the proof to pertain to the multi-item case.

Let $i^* = \arg\max_i \kappa_i$.
For a strategy profile $b$, let $p(b)$ be the price paid and let 
$u_i(b) = \allocsubi(b)[v_i - p(b)] - \inspecti(b)c_i$ be $i$'s utility.
If $b$ is a profile of $\alpha$-best responses, then
\begin{align} 
 \text{welfare} &\geq \E \Big[ p(b) + \sum_i u_{i}(b) \Big]  
\label{eq:pbr.1}
  \geq \E \Big[ p(b) + \alpha \sum_i u_{i}(b_{-i},b_{i}') \Big]
\end{align}
for any profile of deviation strategies, $\{b_{i}'\}$.
In particular, consider the mixed strategy $b_{i}'$ that samples
a random variable $r \in [e^{-\alpha},1]$ with density 
$f(r) = \frac{1}{\alpha r}$ (independent of $\theta_i$)
and claims the item at price $(1-r) \kappa_{i}$ if it is available.
This is achieved by inspecting when the clock value is
at $(1-r) \sigma_{i}$ and claiming at $(1-r) v_{i}$, or immediately if 
$v_i \geq \sigma_{i}$. We have
\begin{align} \label{eq:pbr.2}
  u_i(b'_i) &= \allocsubi(b_{-i},b'_i) (\kappa_i - p(b)) 
                 = \allocsubi(b_{-i},b'_i) (r \kappa_i)
              \geq 0,
\end{align}
where the first equation holds 
because $b'_i$ always exercises in the money,
and the second holds because strategy $b'_i$
is designed to always pay price $(1-r)  \kappa_i$ 
upon winning the item.
Combining~\eqref{eq:pbr.1} and~\eqref{eq:pbr.2},
and letting $i^*$ denote the identity of the bidder
with the highest covered call value, we have
\begin{align} \label{eq:pbr.3}
   \text{welfare} &\geq \E \Big[ p(b) + \alpha u_{i^*}(b_{-i^*},b'_{i^*}) \Big]
         = \E \Big[ p(b) + \alpha \allocsub{i^*}(b_{-i^*},b'_{i^*}) (r \kappa_{i^*}) \Big].
\end{align}
We claim that the following inequality holds for all type profiles $\theta$
and strategy profiles $b$:
\begin{align} \label{eq:pbr.4}
    \E \Big[ p(b) + \alpha \allocsub{i^*}(b_{-i^*},b'_{i^*}) (r \kappa_{i^*}) \Big| \theta, b \Big]
    \geq (1 - e^{-\alpha}) \kappa_{i^*}.
\end{align}
Since we are conditioning on both $\theta$ and $b$, the only remaining 
randomness is due to the random sampling of $r$ by strategy $b'_{i^*}$,
which in turn may influence the allocation of the item and the price paid. 
Let $p$ denote the price at which a bidder other than $i^*$ will claim
the item when $i^*$ is excluded from the auction and the other bidders'
strategy profile is $b_{-i^*}$. If  
$p > (1 - e^{-\alpha}) \kappa_{i^*}$ then 
$\allocsub{i^*}(b_{-i^*},b'_{i^*}) = 0$ 
and $p(b) = p$ so the validity of Inequality~\eqref{eq:pbr.4} is obvious. 
Otherwise, setting $\kappa = \kappa_{i^*}$ for notational convenience,
observe that bidder $i^*$ wins the item if and only if $(1-r) \kappa > p$,
i.e. $r < 1 - p/\kappa$.
Thus, we have
\begin{align*} 
  \E \Big[ \alpha \allocsub{i^*}(b_{-i^*},b'_{i^*}) (r \kappa_{i^*}) \Big| \theta, b \Big]
  &=
  \alpha \int_{e^{-\alpha}}^{1-p/\kappa} r \kappa \, f(r) \, dr 
  =
  \int_{e^{-\alpha}}^{1-p/\kappa} \kappa \, dr \\
  &= 
  \kappa \left[ 1 - p/\kappa - e^{-\alpha} \right] =
  \left(1 - e^{-\alpha}\right) \kappa - p.
\end{align*}
Observing that $p(b) \geq p$ pointwise, we may add $\E[p(b)|\theta,b]$ to the left
side and $p$ to the right side, obtaining Inequality~\eqref{eq:pbr.4}.

Finally, combining Inequalities~\eqref{eq:pbr.3} and~\eqref{eq:pbr.4}
we find that the welfare attained by any profile of $\alpha$-best responses
satisfies
\begin{align} \label{eq:pbr.6}
  \text{welfare} &\geq
  \left(1 - e^{-\alpha}\right) \E[\kappa_{i^*}].
\end{align}
In light of Lemma 1 which says
that $\E[\kappa_{i^*}]$ is an upper bound on the expected
welfare attained by any policy, this concludes the proof of the
proposition.

The proof of the approximation factor $\frac12 \left(1 - e^{-2\alpha}\right)$ 
in the multi-item case is a modification of the proof of Theorem 3.
A key difference is that we are only assuming that $b$ is a profile
of $\alpha$-best responses, so Inequality~(7) must be 
relaxed to
\begin{equation} \label{eq:pbr.7}
 \E\left[ \text{welfare}\right] = \E \left[ \sum_i u_i(b) + p_i(b) \right]  
          \geq \E \left[ \alpha \sum_i u_i(b_{-i},b_i') + p_i(b) \right] .
\end{equation}
To analyze the quantity on the right side, one considers for each player $i$ 
the deviation strategy $b'_i$ that:
\begin{enumerate}
\item Samples a random number $r \in [e^{-2\alpha},1]$ with density $f(r)= \frac{1}{2 \alpha r}$. 
For an object with strike price $\sigma$ and value $v$,
define the {\em shaded strike price} and {\em shaded value} to be $(1-r)\sigma$ and $(1-r)v$, respectively.
%\label{step:r}
\item Takes no action at any time when the clock value $t$ is strictly greater than the highest remaining shaded strike price of an uninspected object and the highest remaining shaded value of an inspected object.
\item When $t$ becomes weakly less than either of these, the deviation claims any object that has been inspected and has shaded value of $t$ and if no such object exists it inspects, in random order, any object with shaded strike price $t$, claiming that object if its shaded value is weakly above $t$ and otherwise continuing inspection until all objects of shaded strike price $t$ have been inspected. 
\end{enumerate}
The key inequality that underpins the analysis of welfare
in an equilibrium strategy profile $b$ is 
\begin{equation} \label{eq:pbr.8}
  \alpha \E \left[ \kappa_i \left(b_{-i},b_i' \right) - p_{i} \left(b_{-i},b_i' \right) 
              \mid \{(\theta_j,v_j)\}_{j=1}^n \right] + \tfrac12 p_i(b) 
   + \tfrac12 p^j(b) \geq \tfrac{1}{2} \left(1 - e^{-2 \alpha} \right) \kappa_{ij} . 
\end{equation}
The proof of this inequality is exactly parallel to the proof of 
Lemma~4, with appropriate modifications
due to the fact that the density of $r$ is now $f(r) = \frac{1}{2 \alpha r}$
rather than $\frac{1}{2r}$. Compared with the proof of 
Theorem~3, we have an additional factor of $\alpha$ on the 
right side of~\eqref{eq:pbr.7} which matches the additional factor of
$\alpha$ on the left side of~\eqref{eq:pbr.8}. This allows us to combine
the two inequalities to derive the claimed approximation guarantee, in a
manner that exactly parallels the way that Inequality~(7)
was combined with Lemma~4 to derive the
approximation guarantee in Theorem~3. 
\end{proof}

\subsection{Common values} \label{app:commonvalues}

\begin{proof}[Proof of \Cref{commonprop}]
Let $i$ be the bidder who turns out to have the lowest cost. 
Note that the first-best welfare is achieved when bidder $i$ 
inspects the common value, $v$, and acquires the item if and
only if $v \geq 0$. This implies that the first-best welfare is
equal to $\E[v^+] - c_i$. Denote this quantity by $w$.

Now let $b$ be any equilibrium profile of strategies and let
$b'_i$ be the mixed strategy for $i$ 
that samples a random $r \in [\frac1e,1]$ with density $f(r)=1/r$ and 
inspects the item's value, $v$, when the descending clock is at $(1-r) w$,
unless the item has already been claimed.
Upon inspecting the item and finding that $v \geq 0$, 
strategy $b'_i$ immediately claims the item at the current
price of $(1-r) w$. If $v<0$ then $b'_i$ never claims the item.

All bidders have non-negative 
utility in equilibrium, and bidder $i$'s utility in equilibrium is at least as great as her
utility when playing $b'_{i}$, so
\begin{align} \label{eq:common.1}
  \E \left[ \text{welfare} \right] & =
  \E \left[ p(b) + \sum_j u_j(b) \right]  \geq
  \E \left[ p(b) + u_{i}(b_{-i},b'_i) \right].
\end{align}
Condition on the profile of bidders' costs, $c$, so that the 
only remaining random variables are the item's value, $v$,
and the random choice of $r$ in mixed strategy $b'_i$. We claim
that 
\begin{align} \label{eq:common.2}
  \E \left[ p(b) + u_i(b_{-i},b'_i) \mid c \right] \geq \left( 1 - \tfrac1e \right) w
\end{align}
Let $p$ denote  the price at which a bidder other than $i$ will claim
the item when $i$ is excluded from the auction and the other bidders'
strategy profile is $b_{-i}$. If $p > \left(1 - \tfrac1e \right) w$ then
bidder $i$ does not inspect the item's value, and another bidder wins
the item at price $p$. We thus have
$p(b) + u_i(b_{-i},b'_i) = p > \left(1 - \tfrac1e \right) w$
which establishes~\eqref{eq:common.2}. Otherwise, we have
$p \leq \left(1 - \tfrac1e \right) w$ and bidder $i$ inspects the
item if and only if $r < 1 - p/w$. Note that conditional on the value
of $r$ and on the event that bidder $i$ inspects the item, her 
utility is 
\begin{align*}
  \E[v^+] - c_i - (1-r) w \cdot \Pr(v \geq 0) &=
  w - (1-r) w \cdot \Pr(v \geq 0) \geq r w.
\end{align*} 
Integrating over the random choice of $r$,
we obtain 
\begin{align} \label{eq:common.3}
  \E \left[ u_i(b_{-i},b'_i) \mid c \right]  
  &=
  \int_{1/e}^{1 - p/w} r w \, f(r) \, dr 
  =
  \left( 1 - \tfrac{p}{w} - \tfrac{1}{e} \right) w
  = 
  \left( 1 - \tfrac{1}{e} \right) w - p.
\end{align}
Since $p(b) \geq p$ we may add $\E[p(b)|c]$ to the
left side of~\eqref{eq:common.3} and $p$ to the right side,
obtaining~\eqref{eq:common.2}

Finally, combining Inequalities~\eqref{eq:common.1} and~\eqref{eq:common.2}
we obtain $\E[\text{welfare}] \geq \left(1 - \frac1e\right) w$, as claimed.
\end{proof}

\subsection{Optional inspection}\label{app:optionalinspection}

\begin{proof}[Proof of \Cref{prop:optional}]
For bidder $i$, let $\bar{v}_i =
\E[v_i \mid \theta_i]$ denote the expected value of
acquiring the item given $i$'s type, $\theta_i$. 
For any procedure, let us represent its net welfare as the sum of 
two terms: $w_i$, the {\em inspected welfare}, is the net 
contribution from bidders who inspect their value, while
$w_u$, the {\em uninspected welfare}, is the net contribution
from bidders who do not inspect their value. More precisely,
\begin{align*}
  w_i &= \sum_i \inspecti \cdot \left[ \allocsubi v_i - c_i \right] \\
  w_u &= \sum_i (1 - \inspecti) \cdot \bar{v}_i.
\end{align*}
Conditional on the profile of realized types, $\{\theta_i\}$, the
procedure that maximizes $\E[w_i]$ is Weitzman's optimal search
procedure, which achieves $\E[w_i] = \E[\max_i \kappa_i]$.
The procedure that maximizes $\E[w_u]$ simply allocates the item
to the bidder $i$ with maximum $\bar{v}_i$. Therefore, letting
$w_i^*$ and $w_u^*$ denote the values of $w_i$ and $w_u$
for the first-best procedure, we have
\begin{align} \label{eq:optional.1}
  \E[\text{first-best}] &= \E[w_i^* + w_u^*] \leq
  \E[\max_i \kappa_i] + \E[\max_i \bar{v}_i].
\end{align}
We claim that any equilibrium of the Dutch auction 
achieves
\begin{align}
\label{eq:optional.2}
  \E[\text{welfare}] &\geq \left(1 - \tfrac1e\right) \E[\max_i \kappa_i] \\
\label{eq:optional.3}
  \E[\text{welfare}] &\geq \left(1 - \tfrac1e\right) \E[\max_i \bar{v}_i] 
\end{align}
The proposition will follow by summing these two inequalities and
combining with~\eqref{eq:optional.1}.

The proofs of Inequalities~\eqref{eq:optional.2}
and~\eqref{eq:optional.3} use a ``smoothness'' argument
very similar to the arguments used in proving 
Theorem~3 and  Propositions~5 and 6. 
Let $b$ be any equilibrium profile
of strategies, and define deviation 
strategies $b'_i$ and $b''_i$ for bidder $i$ to be the following
mixed strategies: sample a random $r \in [\frac1e,1]$ with
density $f(r)=1/r$. In deviation $b'_i$, bidder $i$ inspects the item when
the descending clock is at $(1-r)\sigma_i$ and claims it either 
immediately (if $v_i \geq \sigma_i$) or when the descending 
clock reaches $(1-r)v_i$. In deviation $b''_i$, bidder $i$ acquires
the item without inspection when the descending clock reaches
$(1-r) \bar{v}_i$. Note that,
conditional on any type profile $\{\theta_j\}$,

For any profile of types $(\theta_1,\ldots,\theta_n)$
let $i'$ denote the bidder with maximum covered call value, and
let $i''$ denote the bidder with maximum conditional expected value.
The equilibrium welfare satisfies
\begin{align}
\label{eq:optional.4}
  \E[\text{welfare}] &=
  \E \left[ p(b) + \sum_i u_i(b) \right] \geq
  \E \left[ p(b) + \sum_i u_i(b_{-i},b'_i) \right] \geq
  \E \left[ p(b) + u_{i'}(b_{-i'},b'_{i'}) \right] \\
\label{eq:optional.5}
  \E[\text{welfare}] &=
  \E \left[ p(b) + \sum_i u_i(b) \right] \geq
  \E \left[ p(b) + \sum_i u_i(b_{-i},b''_i) \right] \geq
  \E \left[ p(b) + u_{i'}(b_{-i''},b''_{i''}) \right]
\end{align}
where the second inequality on each line follows because
the expected utility of bidder $i$ using either of the 
deviations $b'_i, b''_i$ is non-negative. 
As in the proof of Proposition~6, conditional
on any profile of types $\{\theta_i\}$, each 
of the inequalities 
\begin{align}
\label{eq:optional.6}
  \E \left[ p(b) + u_{i'}(b_{-i'},b'_{i'}) \mid \{\theta_i\} \right] &\geq \left( 1 - \tfrac1e \right) \kappa_i \\
\label{eq:optional.7}
  \E \left[ p(b) + u_{i''}(b_{-i},b''_{i''}) \mid \{\theta_i\} \right] &\geq \left( 1 - \tfrac1e \right) \bar{v}_i
\end{align}
holds pointwise. The justification of both inequalities recapitulates
the justification for Inequality~\eqref{eq:common.2} in the proof
of Proposition~6, except for notational changes to account
for the minor differences in the deviations being considered.
Combining~\eqref{eq:optional.4} and~\eqref{eq:optional.5}
with~\eqref{eq:optional.6} and~\eqref{eq:optional.7} we 
obtain \eqref{eq:optional.2} and~\eqref{eq:optional.3}, from
which the proposition follows.
\end{proof}

\subsection{Posted pricing}\label{postedapp}

\begin{proof}[Proof of Corollary \ref{prophetcorr}]

Each bidder $i$'s best response, when offered $\pi$, is to inspect iff $\pi \leq \sigma_i$ and subsequently claim the item iff $\pi \leq v_i$; in other words, $i$ accepts the posted-price offer if and only if $\pi \leq \kappa_i$.
These strategies exercise in the money, so the welfare of the sequential posted-price procedure is equal to the expected covered call value of the winner, while the optimal welfare is upper-bounded by the expected maximum covered call value.
Letting $X_i$ be the covered call value of bidder $i$, and noting that these random variables $\{X_i\}$ are independent, Theorem \ref{thm:prophet} implies the result.

\end{proof}

\section{Revenue Guarantees} \label{app:revenue}
Here we show that, under some common assumptions, the Dutch auction with per-bidder reserves has approximately optimal revenue in the single-item setting with inspection costs.
We will rely on material from main text Section 4.
We will first show that the Dutch auction with per-bidder reserves is still invariant to search costs, \emph{i.e.} ``functionally equivalent'' to an auction without costs of inspection.
We may then apply known results on first-price auctions with reserves, under suitable assumptions.

\begin{definition}
The Dutch auction with per-bidder reserve prices $\{r_i\}_{i=1}^N$ is a sales procedure in which a clock begins at $\infty$ and continuously decreases.
Any bidder $i$ may stop the clock at any time weakly exceeding $r_i$, ending the auction, paying the current clock value, and obtaining the item.
\end{definition}

Recall Definition 3: for a given type $\theta_i \in \Theta_i$, its \emph{covered counterpart} $\theta_i^{\circ}$ has zero inspection cost and value distributed according to the distribution of $\kappa_i$, the covered-call value for $\theta_i$.
Given a prior $\mathcal{P}$ on bidders in our model, $\mathcal{P}^{\circ}$ is the corresponding prior on covered counterparts.

Recall Definition 5: two auction outcomes are termed \emph{functionally equivalent} if the same bidder gets the item and at the same price, and the set of inspection costs paid is the same.
This is extended to functional equivalence of strategies and of equilibria.

\begin{theorem} \label{theorem:dutch-reserves-bijection}
There is a mapping from equilibria of the Dutch auction with per-bidder reserves in our model with types jointly distributed according to $\mathcal{P}$ to equilibria of the Dutch auction with per-bidder reserves with the covered counterpart distribution $\mathcal{P}^{\circ}$ where bidders know their values without inspection.
This mapping preserves bidder expected utility and auctioneer expected revenue and induces a bijection on functional equivalence classes of equilibria.
\end{theorem}
\begin{proof}
Fix a bidder $i$ and any strategies of the other bidders.
We will construct a mapping between (A) best-responses of $i$ in the Dutch auction with per-bidder reserves and inspection and (B) best-responses of $i$'s covered counterpart in the Dutch auction with per-bidder reserves and no inspection.
This will be a bijection of functional equivalence classes of best responses, preserving bidder utilities and auctioneer revenue, proving the theorem.

First, we will construct a bijection between setting A above and (A') the Dutch auction without reserves and with inspection, where bidders other than $i$ continue to play their strategies from the auction with reserves and we have added a ``reserve bidder'' who always bids $r_i$.
We claim this is a bijection of best-responses of bidder $i$ between $A$ and $A'$, for the following reason. Every strategy in $A'$ that attempts to inspect or claim below $r_i$ is dominated by doing nothing, since the reserve bidder always bids $r_i$.
All other strategies in A' are also available in A, and all strategies in A are available in A'.
Furthermore, every strategy has the same utility for $i$ in both cases, so the mapping of strategies in A to the same strategy in A' induces a bijection between functionally equivalent best-responses; it also preserves all bidder utilities and auctioneer revenue.

Now, consider setting B', obtained from B in the same way that A' is obtained from A: by removing the reserve prices and adding a ``reserve bidder'' who always bids $r_i$.
Because these both involve a Dutch auction without reserves, and B' is the inspection-less covered counterpart of A', Claim 4 asserts a bijection between functional equivalence classes of best-responses of bidder $i$ in A' and $i$'s covered counterpart in B'; this bijection preserves all bidders' utilities and auctioneer revenue.

Finally, we construct a bijection between $i$'s covered counterpart's best responses in B' and in B.
Again, any strategy of $i$ in B' that plans to inspect or claim at a time below $r_i$ is dominated, as the reserve bidder always claims the item at this time; all other strategies are equivalent to the same strategy played in B, and preserve all bidder utilities and auctioneer revenue.

By composing these bijections, we obtain the desired bijection between A and B.
\end{proof}

\subsection{Revenue guarantees for Dutch with reserves}
We now develop some revenue guarantees for the Dutch auction with per-bidder reserves in the standard independent private values setting (\emph{i.e.} each bidder knows her value without inspecting).
These will show that the expected revenue obtained by the auction is at least some constant fraction of the total welfare generated in the auction, under some assumptions on value distributions.
This implies that expected revenue is at least the same fraction of the optimal obtainable revenue.\footnote{Ideally, one could obtain a better bound by comparing to the revenue of the optimal auction.
Unfortunately, with inspection costs, we do not have a better upper-bound on the optimal revenue than the total welfare.}

The bounds we prove for this independent private values setting will then immediately transfer to the case with inspection costs via Theorem \ref{theorem:dutch-reserves-bijection}.

\paragraph{Preliminaries.}
The \emph{virtual value function} associated with a distribution $F$ having density $f$ is
 \[ \phi_F(x) = v - \frac{1-F(x)}{f(x)} . \]
For bidder $i$, let $F_i$ denote her distribution of covered call values and $\phi_i$ the associated virtual value function.
(Recall that the covered call value is $\kappa_i = \min\{v_i,\sigma_i\}$, the minimum of value and strike price.)

A common assumption in proofs of approximate revenue maximization is that each bidder's value distribution be \emph{regular}, meaning its corresponding virtual value function has nonnegative derivative.
Here we assume a somewhat stronger, parameterized condition on the distribution of covered calls, \emph{$\rho$-concavity.}
\begin{definition}
 $F$ is \emph{$\rho$-concave} for $\rho \geq -1$ if $\frac{d\phi}{dx} \geq 1 + \rho$.
\end{definition}
Regularity is equivalent to $-1$-concavity, while the common \emph{monotone hazard rate} assumption is $0$-concavity.
We will assume that there is a $\rho > -1$ such that each bidder's covered call distribution $F_i$ is $\rho$-concave, obtaining bounds parameterized by $\rho$.

\paragraph{Results.}
The following lemma is not an original idea here but represents a common use case of $\rho$-concave distributions, particularly in \citet{anderson2003efficiency}.
\begin{lemma} \label{lemma:rho-concave-approx}
Let value distributions $F_1,\dots,F_n$ all be $\rho$-concave for $\rho > -1$.
Let $R$ be the optimal expected revenue of any auction with private values drawn indpendently from these distributions, and $W$ the optimal expected welfare.
Then
 \[ R \geq \left(\frac{2}{e+4}\right) \frac{\rho+1}{\rho+2} W . \]
\end{lemma}
\begin{proof}
Note that the optimal achievable revenue, from Myerson's auction, is the expected maximum virtual value, \emph{i.e.}
 \[ R = \E \max_i \phi_i(\kappa_i)^+ . \]
Meanwhile, welfare is the expected maximum value:
 \[ \E \max_i \kappa_i . \]
Pick a bidder $i$ and fix any realizations of $\kappa_{-i}$.
Let $s = \phi_i^{-1}\left(\max_{j\neq i} \phi_j(\kappa_j)^+ \right)$.
Intuitively, $s$ is the ``reserve price'' faced by bidder $s$ induced by the opponents' bids, in Myerson's auction.
The revenue of the optimal auction is
\begin{align}
 \text{revenue} &= \E_{\kappa_{-i}} R_i(\kappa_{-i})  \label{eqn:rho-concave-rev-expectation}  \\
 \text{where }~R_i(\kappa_{-i}) &= \int_{k=s}^{\infty} f_i(k) \phi_i(k) dk.  \nonumber
\end{align}
We now apply two bounds of \citet{anderson2003efficiency}.
First, in this single-bidder environment, consumer surplus is upper-bounded by revenue divided by $\rho + 1$.
In our notation, let the welfare generated by $i$ be $W_i(\kappa_{-i}) = \int_{k=s}^{\infty} f_i(k) k dk$.
Then
\begin{align*}
 R_i(\kappa_{-i}) &\geq \left(\rho + 1\right) \left(W_i(\kappa_{-i}) - R_i(\kappa_{-i})\right)  & \text{\citep{anderson2003efficiency}} \\
 \implies R_i(\kappa_{-i}) &\geq \frac{\rho+1}{\rho+2} W_i(\kappa_{-i}) .
\end{align*}
Second, the difference in welfare generated by $i$ between the welfare-optimal mechanism and the revenue-optimal one (the ``dead weight loss'') is bounded.
Letting $W_i^*(\kappa_{-i}) = \int_{k=\max_{j\neq i}\kappa_j}^{\infty} f_i(k) k dk$ be the first term:
\begin{align}
  W_i^*(\kappa_{-i}) - W_i(\kappa_{-i})
  &\leq \left((1+\rho)^{1/\rho} + \frac{\rho+2}{\rho+1}\right) R_i(\kappa_{-i})   & \text{\citep{anderson2003efficiency}}  \nonumber \\
 \implies W_i^*
  &\leq \left((1+\rho)^{1/\rho} + 2\frac{\rho+2}{\rho+1}\right) R_i(\kappa_{-i})  \nonumber \\
  &\leq \frac{e+4}{2} \frac{\rho+2}{\rho+1} R_i(\kappa_{-i}) .  \label{eqn:rho-concave-upper-bound}
\end{align}
Here, if $\rho = 0$, then $(1+\rho)^{1/\rho} := e$; and we use that $(1+\rho)^{1/\rho} \leq \frac{e}{2} \frac{\rho+2}{\rho+1}$.
Now the result follows by plugging Inequality \ref{eqn:rho-concave-upper-bound} to (\ref{eqn:rho-concave-rev-expectation}), obtaining
\begin{align*}
 \text{revenue}
  &\geq \frac{2}{e+4}\frac{\rho+1}{\rho+2} \E_{\kappa_{-i}} W_i^*(\kappa_{-i})  \\
  &= \frac{2}{e+4}\frac{\rho+1}{\rho+2} \text{welfare} .
\end{align*}
\end{proof}

\begin{lemma}[\cite{POArevenue}] \label{lemma:hartline-revenue}
When bidders' value distributions are regular (including any $\rho$-concave distribution for $\rho \geq -1$), the first-price auction with per-bidder reserves has expected revenue at least a fraction $\frac{e-1}{2e}$ of the optimal (Myerson's) expected revenue.

The reserves are set to each bidder's inverse virtual value of zero.
\end{lemma}

\begin{theorem} \label{theorem:first-price-revenue}
In equilibrium of the first-price auction with per-bidder reserves of $\{\phi_i^{-1}(0)\}_{i=1}^N$ and all value distributoins $\rho$-concave for $\rho > -1$, expected revenue is at least a fraction $0.09 \frac{\rho+1}{\rho+2}$ of total welfare.
\end{theorem}
\begin{proof}
We have
\begin{align*}
 \text{expected revenue}
   &\geq \frac{e-1}{2e} \text{optimal revenue}
     & \text{Lemma \ref{lemma:hartline-revenue}}  \\
   &\geq \frac{e-1}{2e} \left(\frac{2}{e+4}\right) \frac{\rho+1}{\rho+2} \text{welfare}
     & \text{Lemma \ref{lemma:rho-concave-approx}}  \\
   &= \frac{e-1}{e(e+4)} \frac{\rho+1}{\rho+2} \text{welfare}  \\
   &\geq 0.09 \frac{\rho+1}{\rho+2} \text{welfare} .
\end{align*}
\end{proof}

\begin{corollary} \label{corollary:dutch-revenue-approx}
In the setting with inspection costs, suppose that bidder covered-call value distributions are $\rho$-concave for $\rho > -1$ and have virtual value functions $\{phi_i\}_{i=1}^N$.
The expected revenue of the Dutch auction with per-bidder reserves of $\{\phi_i^{-1}(0)\}_{i=1}^N$ is at least $0.09 \frac{\rho+1}{\rho+2}$ fraction of the optimal expected revenue.
\end{corollary}
\begin{proof}
By Theorem \ref{theorem:dutch-reserves-bijection}, the revenue and welfare are equal to that of a first-price auction with values distributed as the covered-call values.
By Theorem \ref{theorem:first-price-revenue}, revenue in the first-price auction exceeds the given fraction of total welfare, so this is the case in the Dutch auction with the same reserves and inspection costs.
Finally, total welfare is at least the optimal expected revenue.
\end{proof}

In calibrations, we tend to find a significant improvement in revenue for the descending-price mechanism over a simultaneous sealed-bid auction when inspection costs are high, \emph{i.e.} in the startup acquisition calibrations.  See Subsection \ref{addresults} below for details.  For timber sales, where inspection costs are much lower, improvement over the second-price or BKRS procedures is more modest.

\section{Calibrations - In Detail} \label{sec:app-calibrations-extra}

In this appendix we discuss in detail our methods for solving for approximate equilibria of the simultaneous second-price and Dutch auctions and provide additional numerical results.  All discussions below are thus formulated in terms of the general model of \Cref{model} in the text, rather than the particular formulation in Section 5, unless specifically stated. We begin by  describing our procedure at a relatively high level here, highlighting only elements of potential interest to a general economic theorist.   Then we describe some additional results from our calibrations. Finally we fill in additional details about our approach that can be used to replicate our methods from scratch.

\subsection{Simultaneous second-price auction}\label{SPAapp}

Our discussion here corresponds to step 4 of the pseudo-algorithm we describe in \Cref{numericsapp}.

We focus on equilibria in which, at any decision point at which agents have a weakly dominant truthful strategy they adopt this strategy, given that \citet{VCGinfo} show this leads to the best possible outcome in a simultaneous setting.  In particular, any bidder who inspects her value chooses to bid this value and any bidder who chooses not to inspect her value bids the value she would in expectation derive from being awarded the object, $\tilde v_i\equiv \mathbb E_{\theta_i} \left[ v_i \right]-c_i$.  

Next note that a) because the probability of paying the cost is strictly higher under inspection than not, it becomes strictly less attractive to inspect when $c_i$ is higher holding fixed $\theta_i$, b) for $c_i=0$ inspection is weakly dominant and c) for $c_i=\infty$ non-inspection is weakly dominant.  Thus for any value of $\theta_i$ there is a unique threshold value of $\overline c_i \left(\theta_i\right)$ such that those with $c_i\leq \overline c_i\left(\theta_i\right)$ will inspect and those with $c_i> \overline c_i\left( \theta_i\right)$ will not inspect.

A given strategy induces a distribution of bids made by $i$ whose cumulative distribution function we call $G_i$.
A sufficient statistic in determining $i$'s best response is the distribution of the highest competing bid, call it $b_{-i}^{\max}$ made by all other bidders, which is determined by $\{G_j\}_{j=1}^N$.  Let $G_{-i}$ be the cumulative distribution function of this variable and let $g_{-i}$ be its density (which we now assume exists).  The payoff to inspecting is
$$\int_{v_i=0}^{\infty}\int_{b_{-i}^{\max}=0}^{v_i} \left(v_i-b_{-i} \right)f\left(v_i\right) g_{-i}\left(b_{-i}\right)db_{-i}dv_i ~~ - ~~ c_i$$
and the payoff to not inspecting is
$$\int_{b_{-i}^{\max}=0}^{\tilde v_i} \left(\tilde v_i-b_{-i}\right)  g_{-i}\left(b_{-i}\right)db_{-i}.$$
We can solve for $\overline c_i$ by simply finding the value of $c_i$ that equates these two expressions for any given value of $\theta_i$.  

Given these facts, we iteratively take the following steps, starting with a bid distribution $\{G_i\}_{i=1}^N$ where each bidder bids her value truthfully.
We compute a best response set of theshold costs $\overline c_i\left(\theta_i\right)$, on a grid of possible values of $\theta_i$.
We then compute the induced best-response distribution $G_i^{\textit{br}}$ on a grid of possible bids.
We update the bid distribution by $G_i(x) = (1-\lambda)G_i(x) + \lambda G_i^{\textit{br}}(x)$, where $\lambda \in (0,1)$ is a smoothing parameter.
We then repeat for bidder $i+1$, wrapping around at $N$.
Once this reaches an approximate fixed point, we stop and return the approximate equilibrium strategy $\{\overline c_i\}_{i=1}^N$.
%To determine a best response for bidder $i$ to a given $G_{-i}$ we simply do this on a grid of possible values of $\theta_i$.  For any given profile of $\left\{\overline c_i \right\}_{i=1}^N$ we can determine $\left\{G_{-i},g_{-i}\right\}_{i=1}^{N}$ by drawing random samples of $\left(\theta_i, c_i\right)$, feeding these through the strategies described above and determining the two highest bids, then constructing from this a smoothed empirical distribution of highest bid by any other bidder for each bidder $i$.  To iterate this process we take an input vector $\left\{\overline c_i^{t-1} \right\}_{i=1}^N$ from iteration $t-1$, calculate a best response $\left\{\overline c_i^{\star, t} \right\}_{i=1}^N$ and set $c_i^t=\lambda c_i^{t-1}+(1-\lambda)c_i^{\star,t}$, where $\lambda\in (0,1)$ is a smoothing parameter. Once this mapping is an approximate fixed point we stop and return an approximate equilibrium value of $\left\{ c_i \right\}_{i=1}^{N}$. 

\subsection{Dutch auction}\label{Dutchapp}

Our discussion here corresponds to step 5 of the pseudo-algorithm we describe in \Cref{numericsapp}.

AuctionSolver takes as inputs a number of bidders and a value distribution function for each bidder, which may be described in terms of a number of common mathematical functions, and returns a set of approximate equilibrium bid functions, one for each bidder.  Thus, in order to take advantage of his code, we must approximate the distribution of covered call values using functions expressible in terms of the mathematical functions AuctionSolver allows.  Log-normal distributions appear to offer a very close fit in our specifications except at very low values, which are irrelevant as all equilibrium quantities are functions of first- and second-order statistics.  We thus use a least squares non-linear fit finder to fit a log-normal cumulative distribution function to the empirical cumulative distribution function of covered call values generated by taking a very large number of samples from the $\left(\theta_i, c_i\right)$ distribution for each bidder.  We then feed these log-normal distributions into AuctionSolver (which allows the necessary mathematical functions) and return corresponding bid functions for each bidder.

\subsection{Evaluating welfare} \label{welfareapp}

Our discussion here corresponds to step 6 of the pseudo-algorithm we describe in \Cref{numericsapp}.

With these equilibria in hand, we draw a large number of samples of $\left\{\left(\theta_i, c_i\right)\right\}_{i=1}^N$.  For each of these samples:
\begin{enumerate}
\item To calculate the welfare of the simultaneous second-price auction we directly run the sample through the equilibrium calculated in Subsection\ref{SPAapp} above, tally all the inspection costs and the value of the highest bidder and take the net between these, calling this the welfare in that sample.
\item To calculate the welfare of the Dutch auction we calculate the covered call value for each bidder from $\left(\theta_i, c_i\right)$ by solving for the strike price using Definition~1, use the equilibrium functions, determine the highest bidder and then call the welfare in that sample her covered call value.
\item To calculate first best welfare we follow the same procedure as in the Dutch auction, but simply use the highest covered call value rather than determining the winner.
\end{enumerate}
In the timber calibration, we determine the welfare of the simultaneous second-price auction and the BKRS mechanism using their code.

\subsection{Additional Results}\label{addresults}

\paragraph{Revenue.}
In calibrations, we tend to find a significant improvement in revenue for the Dutch auction in the startup acquisition calibrations.  One explanation for this arises from the fact that, as we discuss shortly, many more inspections in the Dutch auction occur prior to bidding.  Inspections prior to bidding are sunk costs that do not discourage aggressive bids as a result.  Under the Dutch auction, equilibrium always has all inspections occurring prior to bids.  On the other hand, in the equilibrium of the second-price auction, most inspections occur only after bidding has already occurred.  Such inspection costs reduce the effective value a bidder places on the item even in expectation and thus lower optimal bids.  Furthermore the lack of ex-ante expectations dampens the spread of bids and thus makes the highest order statistics that determine revenue lower.  Together these effects mean that delayed inspections have a quite large impact on revenue, even when their impact on efficiency is modest.

This account is further confirmed by the fact that the revenue gains of the Dutch auction are much more modest, and sometimes even negative, compared to the form of the second-price auction studied by \citeauthor{sweeting}, where inspections must occur prior to bids, and even more so relative to the BKRS mechanism.  Another possible explanation for the difference between these cases is that inspection costs are much higher in the start-up than the timber calibration.  In fact, we suspect these two differences are complementary in producing the observed discrepancy.

The results are for the same parameter settings discussed in previously in the paper.
The startup calibrations results are given in Table \ref{startuprevenuetable} and timber in Table \ref{timberrevenuetable}.

\begin{table}
\begin{center}
\begin{tabular}{cc}
Parameter values   & Second price revenue as \% of Dutch \\
\hline \hline
Baseline           & $50$  \\
$\sigma_v^2=.5$    & $50$  \\
$\sigma_v^2=1.5$   & $60$  \\
$\mu_c=-1$      & $50$  \\
$\mu_c=-.25$        & $70$  \\
$ \sigma_c^2=0$    & $60$  \\
$ \sigma_c^2=1$    & $60$  \\
$ \rho=-.5$        & $60$  \\
$ \rho=0$          & $60$  \\
$ \rho=1$          & $60$  \\
$ (\alpha_0,\alpha_1) = (.5,.25)$ & $60$  \\
$(\alpha_0,\alpha_1) =  (.5,0.05)$   & $50$  \\
$(\alpha_0,\alpha_1) =  (.1,.1)$  & $60$  \\
$(\alpha_0,\alpha_1) =  (.1,.7)$  & $70$  \\
$N=2$              & $70$   \\
$N=10$             & $60$
\end{tabular}
\end{center}
\caption{Revenue of the second-price auction as a percentage of that obtained by the Dutch auction in the startup acquisition calibration, to one significant figure. Parameter values given are the ones that differ from the baseline.}\label{startuprevenuetable}
\end{table}

\begin{table}
\begin{center}
%%%% THE COMMENTED-OUT TABLE HAS THREE DIGITS
%\begin{tabular}{ccc}
%Parameter values              & Second Price & Sequential \\
%\hline \hline
%Baseline                      &  $97.6$ & $  99.4$  \\
%$N_{\mathit{mill}}=1$         &  $96.4$ & $  99.0$  \\
%$N_{\mathit{mill}}=7$         &  $99.7$ & $ 101.2$  \\
%$N_{\mathit{logger}}=0$       &  $99.9$ & $ 100.04$ \\
%$N_{\mathit{logger}}=8$       &  $96.7$ & $  98.9$  \\
%$\mu_{\mathit{logger}}=2.921$ &  $94.7$ & $  98.7$  \\
%$\mu_{\mathit{logger}}=4.243$ &  $98.6$ & $  99.3$  \\
%$\mu_{\mathit{diff}}=.169$    &  $97.7$ & $  99.3$  \\
%$\mu_{\mathit{diff}}=.587$    &  $98.0$ & $  99.7$  \\
%$\sigma_v^2=.122$             &  $96.4$ & $  99.2$  \\
%$\sigma_v^2=.646$             &  $99.1$ & $  99.7$  \\
%$\alpha=.505$                 &  $96.7$ & $  99.0$  \\
%$\alpha=.872$                 &  $98.4$ & $ 100.0$  \\
%$K=.39$                       &  $99.8$ & $  99.8$  \\
%$K=3.72$                      &  $95.0$ & $  98.9$  \\
%$K=16$                        &  $79.0$ & $  94.4$  \\
%\end{tabular}
\begin{tabular}{ccc}
Parameter values              & Second Price (\% of Dutch) & Sequential (\% of Dutch) \\
\hline \hline
Baseline                      &  $98$  & $  99$  \\
$N_{\mathit{mill}}=1$         &  $96$  & $  99$  \\
$N_{\mathit{mill}}=7$         &  $100$ & $ 101$  \\
$N_{\mathit{logger}}=0$       &  $100$ & $ 100$  \\
$N_{\mathit{logger}}=8$       &  $97$  & $  99$  \\
$\mu_{\mathit{logger}}=2.921$ &  $95$  & $  99$  \\
$\mu_{\mathit{logger}}=4.243$ &  $99$  & $  99$  \\
$\mu_{\mathit{diff}}=.169$    &  $98$  & $  99$  \\
$\mu_{\mathit{diff}}=.587$    &  $98$  & $ 100$  \\
$\sigma_v^2=.122$             &  $96$  & $  99$  \\
$\sigma_v^2=.646$             &  $99$  & $ 100$  \\
$\alpha=.505$                 &  $97$  & $  99$  \\
$\alpha=.872$                 &  $98$  & $ 100$  \\
$K=.39$                       &  $100$ & $ 100$  \\
$K=3.72$                      &  $95$  & $  99$  \\
$K=16$                        &  $79$  & $  94$  \\
\end{tabular}
\end{center}
\caption{Revenue of the second-price and sequential mechanisms in the timber auction calibration, expressed as a percentage of that obtained by the Dutch auction, to two significant figures.
         Standard error due to sampling is far smaller.
         Parameter values given are the ones that differ from the baseline.}\label{timberrevenuetable}
\end{table}

\paragraph{Features of the second-price auction.}
In our startup acquisition calibrations, we observed some features that might help explain when second price performs well versus poorly.
For the numbers reported below, we report results from just one scenario from each set of parameters, but we believe these to likely be typical results.
Recall that each scenario is solved with relatively high precision.

\begin{table}
\begin{center}
\begin{tabular}{ccc}
Parameter values   & Welfare loss as \% of first best\\
\hline \hline
Baseline           & $7$  \\
$\sigma_v^2=.5$    & $0.4$  \\
$\sigma_v^2=1.5$   & $10$  \\
$\mu_c=-1$         & $9$  \\
$\mu_c=-.25$       & $5$  \\
$ \sigma_c^2=0$    & $7$  \\
$ \sigma_c^2=1$    & $7$  \\
$ \rho=-.5$        & $8$  \\
$ \rho=0$          & $7$  \\
$ \rho=1$          & $5$  \\
$ (\alpha_0,\alpha_1) = (.5,.25)$ & $2$  \\
$(\alpha_0,\alpha_1) =  (.5,05)$   & $10$  \\
$(\alpha_0,\alpha_1) =  (.1,.1)$  & $2$  \\
$(\alpha_0,\alpha_1) =  (.1,.7)$  & $2$  \\
$N=2$              & $2$   \\
$N=10$             & $10$
\end{tabular}
\end{center}
\caption{Welfare loss of a hypothetical ``first-best without inspection'' procedure: Imagine always assigning to the bidder with largest expected value minus cost; she is the only one who inspects.
         Displays the difference compared to the first best as a percent of first best, rounded to one significant figure.}\label{startupsneverinspecttable}
\end{table}
The first is to look at what would happen if bidders never inspected, but instead always bid according to expected value minus cost, with only the winner inspecting (Table \ref{startupsneverinspecttable}).
Perhaps surprisingly, in many of our scenarios this procedure would achieve quite high welfare.
This makes it less surprising that the simultaneous second-price auction achieves good welfare as well.
The welfare of this hypothetical ``never-inspect'' procedure tends to match second price in terms of trends (although its welfare is often significantly worse).
When $\sigma_v^2$ (the variance of the value) is increased to $1.5$ or the heterogeneity of bidders is decreased ($\alpha_0 = \alpha_1 = .1$), the welfare of this procedure drops dramatically.
The welfare of second price also drops in these cases, but not as much.

\begin{table}
\begin{center}
\begin{tabular}{cccc}
Parameter values   & Second price & Descending & First best \\
\hline \hline
baseline (auction 1) & $24$ & $27$ & $27$ \\
baseline (auction 2) & $27$ & $29$ & $29$ \\
baseline (auction 3) & $28$ & $30$ & $30$ \\
\end{tabular}
\end{center}
\caption{Percentage of bidders who eventually inspect in equilibrium of the second-price and Dutch auctions as well as in the first best procedure. Results shown are for the equilibria of three different auctions drawn according to the baseline parameters. There are $5$ bidders. Rounded to two significant figures; errors due to sampling are significantly lower.}\label{startupsfracinspecttable}
\end{table}
Another is to consider the fraction of bidders that inspect in equilibrium, on average.
For the Dutch  auction, this fraction closely (almost exactly) follows that of the optimal procedure.
Perhaps surprisingly, this fraction is also quite similar for the second-price auction (see Table \ref{startupsfracinspecttable}).

However, the way in which these numbers are reached is significantly different.
In Dutch and first best, the winner has always inspected prior to deciding to claim the item; this is not true in the second-price auction.

\begin{table}
\begin{center}
\begin{tabular}{cccc}
Parameter values    & Winner inspection rate  & Winner matches first best rate \\
\hline \hline
baseline (auction 1)  & $52$  & $85$  \\
baseline (auction 2)  & $58$  & $83$  \\
baseline (auction 3)  & $59$  & $82$  \\
$\sigma_v^2=.5$ (auction 1)  & $1$  & $97$  \\
$\sigma_v^2=.5$ (auction 2)  & $2.6$  & $97$  \\
$\sigma_v^2=.5$ (auction 3)  & $2.8$  & $95$  \\
$(\alpha_0,\alpha_1) = (0.1,0.1)$ (auction 1)  & $73$  & $78$  \\
$(\alpha_0,\alpha_1) = (0.1,0.1)$ (auction 2)  & $64$  & $74$  \\
$(\alpha_0,\alpha_1) = (0.1,0.1)$ (auction 3)  & $77$  & $79$
\end{tabular}
\end{center}
\caption{Percentage of winners, in equilibrium of the second-price auction, who inspected prior to submitting a bid (first column) and who were the same as the winner in the first best procedure (second column).
The three auctions for each set of parameters are drawn from those parameter sets (giving $V_i^0, C_i^0$ for each bidder $i$), then solved. Rounded to two significant figures. The first column is relatively low and the second relatively high, meaning that in our calibrations, second price is able to achieve reasonably high welfare and often find the ``right'' winner even from bidders who bid ``blind'' without inspecting.}\label{startupsfracwinnerinspectedtable}
\end{table}
Because of this, another way to evaluate the second-price auction is to look at the fraction of times in equilibrium that the winner had inspected her value before bidding (see Table \ref{startupsfracwinnerinspectedtable}).
This can be compared with, for instance, the fraction of times that the winner was the same as the winner in the optimal procedure.
In general this confirms that in the startups calibration, second price is able to find reasonably high welfare assignments even without always inducing the right bidders to inspect.
In the baseline scenario, despite a significantly low inspection rate of the winner, the overall welfare loss is relatively small.
This suggests that, in our baseline scenario, good welfare can result even when often picking a bidder with high expected value minus cost without their inspecting first.

It also suggests how problems can arise for second price.
An ideal scenario for second price arises with low variance ($\sigma_v^2 = 0.5$), where the winner is often correct despite not inspecting.
This says that bidders are sufficiently informed of their values so that they can correctly sort without inspecting.
But in the case with low heterogeneity ($\alpha_0 = \alpha_1 = 0.1$), even though the winner inspection rate goes up, the rate of matching the first best winner goes down.
This suggests a problem with coordinating the right bidders to inspect.

\subsection{Additional Numerical Details}
Below, we fill in some further details on our methods.
We hope that an outside party can use the information in our paper and this appendix to independently replicate our results writing code from scratch.
Our code itself contains instructions on exactly how we ran the code and interfaced with outside tools in order to produce the numbers reported in this paper.

\subsubsection{Startup Acquisition}
In this case, for each parameter setting considered, we repeated the following steps:
\begin{enumerate}
  \item Randomly draw a ``scenario'' using these parameters.
        This corresponds to drawing $V_i^0$ and $C_i^0$ for each bidder $i$.
  \item Compute an equilibrium of the Dutch auction using AuctionSolver, as described (to avoid duplication) below in Subsection \ref{sec:online-app-auctionsolver}.
  \item Compute an equilibrium of the second-price auction, described next.
  \item Simulate a large number (typically one million or ten million) of random realizations of that scenario.
        For each realization, compute the outcome under the first-best, Dutch equilibrium, and second price equilibrium.
\end{enumerate}
Because interfacing with AuctionSolver is relatively time-intensive, we only drew a small number of scenarios for each set of parameters.
In that sense, we cannot guarantee high accuracy on our results for a given set of parameters.
However, for each \emph{scenario} we drew, we were able to compute the welfare and revenue to quite high precision (sources of error are described below).

In a couple of cases, AuctionSolver was unable to accept or solve inputs from a scenario.
This typically occurred because the covered call value distributions were too high or heavy-tailed, making the program freeze when attempting to process them.
In these cases (less than 10\% of the total number of scenarios drawn), we re-drew the scenario randomly and tried again.
Other than that, the scenarios were not filtered in any way.
After finalizing the code and setting the parameters, we used the first random scenarios we drew for the results reported.
The particular scenarios we drew are saved in the code so that they can be double-checked.

\paragraph{Solving second price.}
In \Cref{calibrations}, we described our high-level approach for solving the second-price auction.
Here we give some additional details.
One point to note is that in almost all of our simulations, there was exactly one bidder of each ``kind''; that is, bidders were all asymmetric.
In a case where there are multiple bidders of a certain kind, the code uses the fact that their best-responses are identical to only look for an equilibrium that is symmetric within a given kind.
In the rest of the exposition, we will write as though there is only one bidder of each kind.

We discretized the bid space into at least 1000 possible values.
The range of bids was placed between $0$ and high on the highest value distribution's CDF, at least $0.9999$.
This large upper bound was chosen because bidders will bid their true values when they choose to inspect, so we cannot only consider relatively low bids.
We also discretized each bidder's space of possible $V_i^1$; recall that this is the type they observe at the start of the auction, determining their value distribution.
Each bidder's discretized bid space depended on her particular distribution of $V_i^1$, in order to cover most of the probability space, but all had the same number of points (at least 1000).

We also precomputed and stored some useful values that are re-used often, which can take up significant memory storage for large grid sizes but greatly speeds up the process.
Code profiling suggested that recomputing these values each loop was using up the majority of computation time.
First, for each possible discretized outcome of $V_i^1$, we computed and stored $i$'s CDF, that is, the probability that her realized $V_i^1$ is at most this value.
Second, for each $i$, for each of the possible discretized $V_i^1$, and for each discretized bid $b$, we precomputed $i$'s expected net utility for bidding her true value when the highest bid of any opponent is $b$.

\paragraph{The outer loop.}
Each loop, we start with a bid distribution and strategy for each bidder $i$.
We compute a new best-response of bidder $1$ along with the corresponding bid distribution.
Then, to avoid over-shooting, we set $i$'s updated bid distribution to be a convex combination of the old one and the best response.
We also track the mean error between the old and new bid distributions.
After repeating this for each bidder $i$, we check if the overall mean error was small enough, and if so, terminate the process.

\paragraph{Computing a best-response for bidder $i$.}
$i$'s opponents' strategies may be summarized as a distribution $G_{-i}$ of the ``highest opposing bid''.
We first compute $G_{-i}$ given the current strategies of the opponents.
We then compute $i$'s best-response as follows.
For each possible discretized outcome of $V_i^1$, from lowest to highest, we compute $c^*$, the cost threshold below which $i$ inspects and above which she bids her expected value minus cost.

We compute $c^*$ by defining $f$ to be the function equal to (utility from just bidding expected value minus cost) - (utility from inspecting and bidding value).
Then we find the root of $f$.

This rootfinding is done in two steps: a discretized approximate search, then a continuous rootfinding using a scientific library.
To see why we use a discretized search first, recall that we discretized the allowable bids.
If the bidder does not inspect, each of these discretized bids corresponds to an expected value minus cost.
Because the expected value is fixed for this choice of discretized $V_i^1$, this results in a list of discretized costs.
Each corresponds to a threshold at which the non-inspecting bidder would choose the next discretized bid.
We can quickly find the two such costs that sandwich $c^*$; then the scientific library only has to search over the small range in between to find it exactly.

We attempt to speed up this search in a couple of ways.
We observe that the expected utility for inspecting is a constant positive utility that does not depend on $c$, minus $c$.
So, we precompute this constant positive amount rather than recomputing it for every evaluation of $f$.
Most importantly, we ``seed'' the search with an initial guess for $c^*$ equal to the $c^*$ found for the previous discretized value of $V_i^1$.
Because $V_i^1$ has changed very little (if the discretization is good), $c^*$ should also not change much.
This allows the discretized search to run very quickly just by linearly searching up or down from the guess.

\paragraph{Updating the bid distribution of $i$.}
In tandem with computing $c^*$ for each discretized $V_i^1$, we simultaneously update the distribution of $i$'s bids.
This is useful for computing $G_{-j}$ for other bidders $j$.

Computing this distribution is straightforward: For each discretized $V_i^1$, let $p_1$ be the probability that the bidder draws $V_i^1$ and let $c^*$ be the corresponding computed cost threshold.
Let $p_2$ be the probability that the bidder draws $c \leq c^*$ conditioned on this realization of $V_i^1$.
Then we know that the bidder inspects with probability $p_1p_2$, in which case, she bids her value.
We can thus put a point mass on each discretized bid equal to $p_1p_2$ times the discretized probability of value equalling that bid conditioned on $V_i^1$.
Meanwhile, with probability $1-p_2$, the bidder chooses not to inspect and bids expected value minus cost.
Her expected value, given $V_i^1$, is a fixed constant.
Thus, for each discretized bid from $0$ up to her expected value, we can add a point mass equal to $p_1$ times the probability that her cost $c$ is in the range such that $\E[v_i] - c$ equals this bid.

\subsubsection{Timber Sales}
For these calibrations, we were able to use the code of Roberts and Sweeting directly to obtain the revenue and welfare of the sequential and second-price mechanisms.
For the Dutch, we followed the same procedure as in the startups case; this is described next.

One difference to the startups case is that here we only consider specific ``scenarios'' (\emph{i.e.} $V_i^0$ and $C_i^0$ are always $0$ and do not play a role).

\subsubsection{Using First-Price to Solve the Dutch Auction} \label{sec:online-app-auctionsolver}
These are the steps we took to calculate the welfare and revenue of the Dutch auction for a given scenario.
One of the main results of the paper is that equilibrium strategies, welfare, and revenue of the Dutch auction can be completely determined by solving a first-price auction in a setting without inspection costs.
Therefore, we always used the following procedure to compute welfare and revenue of the Dutch auction:
\begin{enumerate}
  \item For each kind of bidder $k$, fit a distribution $F_k$ to the covered-call value distribution of $k$.
        For instance, in the timber sales setting, there are at most two kinds of bidders in each auction, mills and loggers.
        All bidders of the ``mills'' kind have the same distribution of covered call values; the same holds for the ``loggers'' kind.
        We fit the distributions by drawing a large number of samples (typically 100,000) to approximate the true CDF, then finding the $(\mu,\sigma)$ for which the lognormal distribution had the least squared error to this CDF.
        We also visually inspected the difference between each fitted distribution and true one to ensure it was a close fit, particularly at the higher quantiles where the fit matters more for the outcome of the auction.

        We chose to use lognormal distributions because they are simple (with only two parameters) yet seem to fit the settings we checked very closely.
        They are also one of the distributions built-in to the AuctionSolver tool (discussed below).
        For other settings of parameters, for instance when there is a significant chance of a negative covered-call value, some other class of distributions may be a better choice.
        It is also worth pointing out that the closeness of fit matters more at higher values which are more likely to be involved in winning the auction.
        For instance, at the opposite extreme, fitting the covered-call distribution well below zero is completely unnecessary as such bidders do not even enter a bid.

        This step is a source of potential error if the fitted distributions $F_k$ are not very close to the true ones, as then we will be solving for equilibrium of a different set of bidders than the true ones.

  \item We used Richard Katzwer's AuctionSolver tool to solve for the equilibrium of a first-price auction.
        In this auction, each bidder of kind $k$ has value distribution $F_k$.
        In order to use the tool, we needed to truncate the value distributions entered for the auction.
        We usually attempted to truncate them quite high, \emph{e.g.} at a CDF of 0.9999 or higher.
        However, in some cases, this gave the tool difficulty in solving the auction, so we truncated somewhat lower.

        The tool also had some difficulty loading larger or heavier-tailed distributions, so we typically scaled down all value distributions $F_k$ by a constant before entering them into AuctionSolver.
        Luckily, this can be accomplished for the lognormal distribution by subtracting a fixed amount from the $\mu$ parameter of each distribution.

        After solving for an equilibrium in AuctionSolver\footnote{We almost universally used the default settings except to try to increase grid size for more precision; our code describes in detail the exact steps we took in using AuctionSolver.}, we used it to print out a discretized bid function $f_k$ for each kind of bidder $k$.
        The output of AuctionSolver takes the form of a list of discretized values and the corresponding bid for each value.
        We constructed $f_k$ by linearly interpolating these points; for instance, if the equilibrium output says that the value $10.0$ maps to bid $5.0$ and value $10.2$ maps to bid $5.1$, then we mapped value $10.1$ to bid $5.05$ and so on.
        Any values above the truncated upper bound are mapped to the upper bound's bid.
        At this point, if we scaled down the value distributions, we now scaled up the output of AuctionSolver by the same factor.
        For instance, suppose we subtracted $1$ from the $\mu$ parameter of every kind of bidder, then ran AuctionSolver.
        This implies that each bidder's value is a factor $e$ smaller in the AuctionSolver results, so those bids are also a factor $e$ smaller.
        To construct $f_k$, we just multiply all values and bids by $e$, then linearly interpolate.

        In using AuctionSolver, we introduce two potential sources of numerical error.
        First, AuctionSolver discretizes the equilibrium bid functions in order to solve for equilibrium.
        We tried to use large grid sizes to avoid error here and it does not seem likely that this could noticeably impact results.
        Second, it upper-bounds the value distributions.
        One might be concerned because this entails solving for equilibrium with slightly different value distributions than the true ones, which are unbounded.
        We tried to use large upper bounds to mitigate this concern.
        One might also be concerned that this upper bound could slightly disrupt welfare results, although the chances of two bidders both exceeding the upper bound in the same auction is very small.
        We also guarded against this by breaking any ties in bids so that the lower-covered-call bidder won the tiebreak, so that we can only underestimate the welfare.
        Meanwhile, the impact on revenue of the upper bound should only be to decrease it.

  \item We simulated a large number of auctions, typically one million or ten million.
        For each, we calculated the eventual bid of each participant, which is $f_k$ applied to their covered-call value (as shown in this paper), and computed the welfare and revenue.
        The average welfare can be computed as the average covered call value of the winner.
        We then averaged the results of these trials.
        By computing the sample variance of these trials, we were also able to compute the standard error, which is significantly smaller than the precision to which we report the results.
        In that sense, the numerical error here is extremely small.
        But again, there is some possibility of error if the original fitted distributions $F_k$ were not good fits to the true covered call distributions.
        If that were the case, we would be simulating strategies that are not actually an equilibrium.
        Again, we attempted to mitigate this concern by visually inspecting the goodness of the lognormal fit.
\end{enumerate}

\newcommand{\qstg}{{\frak q}}
\newcommand{\awardi}{{\mathbb{A}_i^{\ast}}}

\section{Multi-Stage Inspection} \label{app:multistage}

In this section we present full proofs of our results for the case of multistage inspection, culminating in \Cref{lem:general-ubp-restated}, which generalizes our main ``amortization lemma'' (\Cref{lemma:upper-bound-policy}) to the setting of multi-stage
inspection.

\subsection{A more formal model} \label{genmodel-appendix}

This subsection presents, in greater formality, the model
of multi-stage inspection presented in main text 
Section 6.1.  In fact,
the model we present here also
generalizes from the single-item
auctions contemplated in main text Section 6.1 
to auctions with any finite
number of items.

We will equip our 
probability space $(\Omega,\sigfld,\mu)$ with filtrations
$\{\sigfld_{i,j}^{k,\tau}\}$. The subscripts $i$ and $j$ range over bidders
and items, respectively. For the remainder of this section we 
focus on the stopping problem that a single bidder faces
when deciding when to advance the stage of inspection for
a single item and when to acquire it. Thus we are treating $i$ 
and $j$ as fixed for the remainder of this section. Accordingly,
we will omit the double subscript $i,j$ and denote
$\sigma$-fields by $\sigfld^{k,\tau}$. The superscripts $k$ and $\tau$ refer
to the stage of inspection and the time, respectively. We 
think of $k \in \mathbb{N} \cup \{\infty\}$ as a counter that
increases when the bidder endogenously decides to advance
to the next stage of inspection; the special value $k=\infty$
denotes the completion of all stages of inspection, which is
a prerequisite for acquiring an item. 
We think of $\tau \in \reals_+$ as
representing the ``clock time'' which advances exogenously;
as $\tau$ increases the bidder may receive decision-relevant
information. For example, the bidder may be notified that 
a competing bidder has acquired an item. 
In the notation of main text Section 6.1, $\sigfld_{i,j}^{k,\tau}$ denotes the
$\sigma$-field generated by the signals $s_i^1,\ldots,s_i^k$,
along with any exogenous signals that arrive during the time
interval $[0,\tau)$.

The $\sigma$-fields
$\sigfld_{i,j}^{k,\tau}$ satisfy the relation
$\sigfld_{i,j}^{k,\tau} \subseteq \sigfld_{i,j}^{k',\tau'}$
whenever $k < k'$ and $\tau < \tau'$. 
%% In the notation of Section 6.1, $\sigfld^{k,\tau}_{i,j}$
%% corresponds to the $\sigma$-field generated by the random
%% variables $\{s^{k'}_{i,j}\}_{k' \leq k}$ and
%% $\{\tilde{s}^{\tau'}_{i,j}\}_{\tau' \leq \tau}$.
We will assume that $\sigfld^{\infty,\tau}$ is the
$\sigma$-field generated by 
$\bigcup_{k \in \mathbb{N}} \sigfld^{k,\tau}$,
and we will use the analogous notation 
$\sigfld^{k,\infty}$ to denote the
$\sigma$-field generated by 
$\bigcup_{\tau \in \reals_+} \sigfld^{k,\tau}$.
The information
a bidder learns by inspecting an item's value is 
conditionally independent of the information
learned by waiting as time passes; formally,
for any $k \in \mathbb{N}, \tau \in \reals_+$ and
any event $E \in \sigfld^{k,\infty}$,
we have
\begin{equation} \label{eq:cond-indep}
  \Pr \left[ E \given \sigfld^{\infty,\tau} \right] = 
  \Pr \left[ E \given \sigfld^{k,\tau} \right].
\end{equation}

The bidder's valuation for the item is represented by a
$\sigfld^{\infty,0}$-measurable function $v$. The cost of
inspection is represented by a stochastic process $(c^k)$ adapted to
the filtration $\{\sigfld^{k,0}\}_{k=0}^{\infty}$, and at every point
of the sample space the sequence $0=c^0,c^1,c^2,\ldots$ is
non-decreasing and converges to a finite limit, $c^{\infty}$. The
value of $c^k$ should be interpreted as the combined cost that the
bidder must pay to reach the $k^{\mathrm{th}}$ inspection stage. Our
assumption that $v$ is $\sigfld^{\infty,0}$-measurable and that $c^k$
is $\sigfld^{k,0}$-measurable means that the bidder's uncertainty
about the item's value and about future inspection costs may
potentially diminish when she advances to a higher inspection stage,
but it does not diminish as clock time progresses.  We will assume
that $\E[v^+] < \infty$ and $\E[c^{\infty}] < \infty$.  

An inspection policy is a rule for varying the inspection
stage of an item over time, and deciding when to acquire the item,
based on information learned in the past and present. More
formally, it is an ordered pair $(\stg,\allocstg)$, where 
$\stg : \Omega \times \reals_+ \to \mathbb{N} \cap \{\infty\}$ 
denotes the rule for varying the inspection stage and 
$\allocstg : \Omega \times \reals_+ \to \{0,1\}$ 
is the indicator of the ($\sigfld$-measurable) event that the policy decides
to acquire the item at time $\tau$ or earlier. 
With a slight abuse of notation, we will refer 
to such an inspection policy simply as $\stg$ rather than as $(\stg,\allocstg)$.
Inspection policies must satisfy the following properties.
\begin{enumerate}
\item For each $\omega \in \Omega$, $\stg(\omega,\tau)$ 
  and $\allocstg(\omega,\tau)$ are
  non-decreasing functions of $\tau$.
%\item For each $t \in \reals_+$, $\stg(\omega,t)$ is a $\sigfld$-measurable 
%  function of $\omega$.
\item For all $\tau \in \reals_+, \, k \in \mathbb{N}$, we have
  $
     \{\omega \in \Omega \mid \stg(\omega,\tau) > k\} \in \sigfld^{k,\tau}.
  $
\item For all $(\omega,\tau) \in \Omega \times \reals_+$, 
if $\allocstg(\omega,\tau) = 1$ then  $\stg(\omega,\tau) = \infty$.
\end{enumerate}
The second property means that if the inspection policy has decided to
advance beyond stage $k$ at time $\tau$ or earlier, the decision must be
based on information obtained during the first $k$ stages of
inspection and the first $\tau$ units of clock time.  
The third property means that if the inspection policy decides to acquire 
the item, it must complete all stages of inspection.
We will generally
omit the argument $\omega$ from $\stg$, interpreting $\stg(\tau)$ for any
fixed $\tau$ as a random variable defined on $(\Omega,\sigfld)$ and
taking values in $\mathbb{N} \cup \{\infty\}$.  
Since $\stg(\tau)$ is a non-decreasing function of $\tau$, the limit
$\lim_{\tau \to \infty} \stg(\tau)$ is a well-defined 
$(\mathbb{N} \cup \{\infty\})$-valued random
variable, which we will denote by $\stg(\infty)$;
it represents the final inspection stage reached
by $\stg$. Similarly $\lim_{\tau \to \infty} \allocstg(\tau)$
is a well-defined $\{0,1\}$-valued random variable
and we will denote it simply by $\allocstg$; it is the 
indicator random variable of the event that the policy
acquires the item. 
%We define the random
%variable $\allocstg$ to take the value 1 if and only if 
%% there exists a finite time $\tau \in \reals_+$ such that $\stg(\tau)=\infty$. 
%$\stg(\infty) = \infty$. Informally,
%$\allocstg$ is the indicator of the event that the bidder acquires the
%item. 
The random variable $c^{\stg} = c^{\stg(\infty)}$ 
represents the cost of executing the policy $\stg$.

We can combine two inspection policies $\qstg$ and $\altstg$ by
forming their pointwise minimum, $\qstg \wedge \altstg$, or their
pointwise maximum, $\qstg \vee \altstg$. Formally these are defined
by
\begin{align*}
  (\qstg \vee \altstg)(\omega,\tau) = 
  \qstg(\omega,\tau) \vee \altstg(\omega,\tau) 
  & 
  \quad \mbox{and} \quad
  \alloc{\qstg \vee \altstg}(\omega,\tau) = 
  \alloc{\qstg}(\omega,\tau) \vee \alloc{\altstg}(\omega,\tau) \\
  (\qstg \wedge \altstg)(\omega,\tau) = 
  \qstg(\omega,\tau) \wedge \altstg(\omega,\tau) 
  & \quad \mbox{and} \quad
  \alloc{\qstg \wedge \altstg}(\omega,\tau) = 
  \alloc{\qstg}(\omega,\tau) \wedge \alloc{\altstg}(\omega,\tau).
\end{align*}
The verification that $\qstg \vee \altstg$ and $\qstg \wedge \altstg$ satisfy
the definition of an inspection policy is left to the
reader.

\subsection{Generalized strike price} \label{genstrike-appendix}
This subsection generalizes the notion of strike price to the setting
of multi-stage inspection. For each inspection stage $k < \infty$ we
will define a strike price $\sigma^k$ which will be a
$\sigfld^{k,0}$-measurable random variable.  Informally $\sigma^k$
represents the value of an outside option such that a bidder who has
already sunk the cost of reaching stage $k$ is indifferent between
stopping immediately and accepting a payoff of $\sigma^k$, versus
continuing to apply the optimal policy that undertakes at least one
more stage of inspection, given that the policy can stop at any future
time and obtain the outside-option payoff of $\sigma^k$.

Proving the existence and uniqueness of a 
$\sigfld^{k,0}$-measurable function $\sigma^k$ that
achieves this indifference property turns out 
to be a bit subtle. We begin with the following
definition.
\begin{definition} \label{def:discourages}
A $\sigfld$-measurable function $\sigma$
{\em discourages inspection at stage $k$} if
for every inspection policy $\stg$ satisfying
$\stg(\infty) > k$ pointwise, the 
relation
$$
    \E \left[ 
         \allocstg v - c^\stg 
         \given \sigfld^{k,\infty}
       \right]  \leq
   \E \left[
         \allocstg \sigma - c^k 
         \given \sigfld^{k,\infty}
   \right]
$$
holds pointwise almost everywhere. The set of 
all $\sigfld^{\infty,0}$-measurable functions 
% in $L^1(\mu,\sigfld)$ 
that discourage inspection 
at stage $k$ is denoted by $\di^k$.
The set of functions $\sigma \in \di^k$ that 
satisfy $\sigma \leq v^+$ pointwise is 
denoted by $\di^k_0$.
\end{definition}
Informally, to say that a random variable 
$\sigma$ discourages
inspection at stage $k$ means a bidder
who has an outside option worth $\sigma$
and is currently at inspection stage $k$
must weakly prefer a policy that stops immediately
and claims the payoff of $\sigma$ over one
that may perform additional stages of inspection.

The following property of $\di^k$ will be
useful in the sequel.
\begin{lemma} \label{lem:pointwise-min}
For any finite subset $\{\sigma_1,\ldots,\sigma_n\} \subset
\di^k$ the pointwise minimum 
$\sigma(\omega) = \min_{1 \leq i \leq n} \{\sigma_i(\omega)\}$
also belongs to $\di^k$.
\end{lemma}
\begin{proof}
Let $U_i = \{ \omega \in \Omega \mid \sigma_i(\omega) = \sigma(\omega) \}$
and note that the sets $U_1,\ldots,U_n$ are $\sigfld^{k,\infty}$-measurable
and that they cover $\Omega$. The relation
\[
  \E \left[ \allocstg v - c^\stg \given \sigfld^{k,\infty} \right] 
\leq
  \E \left[ \allocstg \sigma_j - c^k \given \sigfld^{k,\infty} \right]
=
  \E \left[ \allocstg \sigma - c^k \given \sigfld^{k,\infty} \right]
\]
holds pointwise a.e.\ on $U_j$. The lemma follows by
patching together these relations, since the sets $U_j$ 
cover $\Omega$.
\end{proof}
 
Since the strike price is informally defined
as the {\em minimum} value of an outside option that
discourages inspection at stage $k$, it
is natural to try defining $\sigma^k$ as the pointwise 
infimum of the functions $\sigma \in \di^k$.
Unfortunately, except at sample points 
that are point-masses of the probability
measure $\mu$, this pointwise infimum will
be equal to $-\infty$ because we can modify 
the value of any $\sigma \in \di^k$ on any 
measure-zero set without altering its membership 
in $\di^k$. Therefore, we have to define
$\sigma^k$ more indirectly using the following
lemma.

\begin{lemma} \label{lem:sigmak}
There exists a $\sigfld^{k,\infty}$-measurable
function $\sigma^k$, taking values in 
$\reals \cup \{ -\infty \}$, such that for all 
$U \in \sigfld^{k,\infty}$, 
\begin{equation} \label{eq:def-sigmak}
  \int_U \sigma^k \, d \mu =
  \inf \left\{ \left. \int_U \sigma \, d \mu 
         \; \right| \;
               \sigma \in \di^k
       \right\} = 
  \inf \left\{ \left. \int_U \sigma \, d \mu 
         \; \right| \;
               \sigma \in \di^k_0
       \right\} .
\end{equation}
The function $\sigma^k$ is unique almost surely,
meaning that any two such functions are equal
except on a set of measure zero. 
\end{lemma}
\begin{proof}
First note that the set $\di^k_0$ is non-empty
because, for example, the function $v^+$ 
belongs to $\di^k_0$: 
% it belongs to $L^1(\mu,\sigfld)$
% by our assumption that $\E[v^+] < \infty$, and
a bidder whose outside option is to obtain 
payoff $v^+$ at no cost will always (at least weakly) prefer 
that option to paying the cost of additional
stages of inspection followed
by attaining a reward which is at best equal to $v^+$.

To prove the existence of $\sigma^k$ we will consider
the set function
$$
  \nu^k(U) = \inf \left\{ \left. \int_U \sigma \, d \mu 
         \; \right| \;
               \sigma \in \di^k
       \right\} = \inf \left\{ \left. \int_U \sigma \, d \mu 
         \; \right| \;
               \sigma \in \di^k_0
       \right\}.
$$
Note that the infimum defining $\nu^k(U)$ is the same
regardless of whether $\sigma$ ranges over the full set $\di^k$ 
or its subset $\di^k_0$; this is because for every $\sigma \in \di^k$
the pointwise minimum $\sigma' = \sigma \wedge v^+$ belongs to
$\di^k_0$ and satisfies 
$\int_U \sigma' \, d\mu \leq \int_U \sigma \, d\mu$.
We will prove that $\nu^k(U)$
is a countably additive measure on $\sigfld^{k,\infty}$ and that it is satisfies 
$\nu^k(U) < \infty$ for all $U$. From the definition of
$\nu^k$ it is clear that $\nu^k(U)=0$ whenever $\mu(U)=0$.
An application of the Radon-Nikodym Theorem then implies
the existence of the function $\sigma^k$ asserted in the
lemma statement.

We now argue that $\nu^k$ is
countably additive. Suppose $U_1,U_2,\ldots$ are disjoint
sets in $\sigfld^{k,\infty}$ and let $U$ denote their union.
Since $v^+ \in \di^k$ we see that $\nu^k(U_i) \leq \int_{U_i} v^+ \, d\mu$
for all $i$. Hence the sum of the non-negative elements
in the set $\{\nu^k(U_i) \mid i=1,2,\ldots\}$ is finite,
being bounded above by $\int_{\Omega} v^+ \, d\mu$.
It follows that the sum $\sum_{i=1}^{\infty} \nu^k(U_i)$
has a well-defined value in $\reals \cup \{-\infty\}$ 
independent of the ordering of the summands: either
the sum of the negative elements in $\{\nu^k(U_i) \mid i=1,2,\ldots\}$ is finite,
in which case $\sum_{i=1}^{\infty} \nu^k(U_i)$
converges absolutely to a finite value, or else the
sum of the negative elements in $\{\nu^k(U_i) \mid i=1,2,\ldots\}$ is
$-\infty$, in which case the partial sums 
$\sum_{i=1}^{n} \nu^k(U_i)$ converge
to $-\infty$ as $n \to \infty$, irrespective of the
ordering of summands.

We must prove that $\sum_{i=1}^{\infty} \nu^k(U_i) = \nu^k(U)$. 
For any $\eps>0$ we can choose $\sigma \in \di^k$ such that
$\int_U \sigma \, d\mu < \nu^k(U) + \eps$.
Without loss of generality we may assume $\sigma \leq v^+$ 
pointwise, since Lemma~\ref{lem:pointwise-min} justifies
replacing $\sigma$ with the pointwise minimum 
$\sigma \wedge v^+$ if necessary. Now, arguing as in the
preceding paragraph, we may conclude that the sum
$\sum_{i=1}^{\infty} \int_{U_i} \sigma \, d\mu$ is
well-defined irrespective of the order of summands,
and that it is equal to $\int_U \sigma \, d\mu$.
It follows that
\[
   \sum_{i=1}^{\infty} \nu^k(U_i) \leq
   \sum_{i=1}^{\infty} \int_{U_i} \sigma \, d\mu =
   \int_U \sigma \, d\mu < \nu^k(U) + \eps.
\]
As $\eps>0$ was arbitrarily small, we may conclude
that $\sum_{i} \nu^k(U_i) \leq \nu^k(U)$. 

To prove the reverse inequality, for $i=1,2,\ldots$
choose $\sigma_i \in \di^k$ such that 
$\int_{U_i} \sigma_i \, d\mu < \nu^k(U_i) + 2^{-i} \eps$.
For $n=1,2,\ldots$ let $\sigma_{(n)}$ be the 
pointwise minimum of $\sigma_1,\sigma_2,\ldots,\sigma_n,v^+$.
Using the definition of $\nu^k$,
\begin{align*}
  \nu^k(U) \leq \int_U \sigma_{(n)} \, d\mu 
    & \leq
    \sum_{i=1}^{n} \int_{U_i} \sigma_i \, d\mu +
    \sum_{i=n+1}^{\infty} \int_{U_i} v^+ \, d\mu \\
    & <
    \sum_{i=1}^{n} \nu^k(U_i) + (1 - 2^{-n}) \eps +
    \sum_{i=n+1}^{\infty} \int_{U_i} v^+ \, d\mu.
\end{align*}
We may choose $n$ large enough that 
    $ \sum_{i=n+1}^{\infty} \int_{U_i} v^+ \, d\mu < \eps$
and
    $ \sum_{i=n+1}^{\infty} \nu^k(U_i) < \eps $.
Then
$$ \nu^k(U) < \sum_{i=1}^{n} \nu^k(U_i) + 2\eps <
              \sum_{i=1}^{\infty} \nu^k(U_i) + 3\eps. $$
As $\eps>0$ was arbitrarily small we conclude that
$\nu^k(U) \leq \sum_{i=1}^{\infty} \nu^k(U_i)$.
Having already proved the reverse inequality, 
we may conclude that $\nu^k$ is a countably
additive measure. By the Radon-Nikodym Theorem,
there exists a $\sigfld^{k,\infty}$-measurable 
function $\sigma^k$ taking values in $\reals \cup \{-\infty\}$
that satisfies~\eqref{eq:def-sigmak}.

Finally, the uniqueness statement in the lemma
follows because if $\sigma^k$
and $\hat{\sigma}^k$ both satisfy~\eqref{eq:def-sigmak}
then $\int_U (\sigma^k - \hat{\sigma}^k) \, d\mu = 0$
for all $U \in \sigfld^{k,\infty}$, implying that 
$\{ \omega : \sigma^k \neq \hat{\sigma}^k \}$ has
measure zero.
\end{proof}

\begin{lemma} \label{lem:sigmak2}
If $\phi$ is a $\sigfld^{k,\infty}$-measurable function taking values
in a bounded non-negative interval $[0,M]$ then
$\E \left[ \phi \sigma^k \right] = 
\inf \left\{ \E \left[ \phi \sigma \right] \mid \sigma \in \di^k \right\} = 
\inf \left\{ \E \left[ \phi \sigma \right] \mid \sigma \in \di^k_0 \right\}.$
\end{lemma}
\begin{proof}
As in Lemma~\ref{lem:sigmak}, the infimum is the same regardless of 
whether $\sigma$ ranges over $\di^k$ or $\di^k_0$. We will henceforth
work only with $\sigma \in \di^k_0$. 

First suppose that $\phi$ is a simple function. Then we may 
represent it as $\phi = \sum_{i=1}^n w_i \indic_{V_i}$ for
disjoint measurable sets $V_i \in \sigfld^{k,\infty}$ and weights
$w_i \in [0,M]$. In that case Lemma~\ref{lem:sigmak} yields
\[
   \E \left[ w_i \indic_{V_i} \sigma^k \right] = 
  \inf \left\{ \left. \E \left[ w_i \indic_{V_i} \sigma \right] \, \right| \,
                   \sigma \in \di^k_0 \right\}
\]
for $i=1,\ldots,n$. Summing over $i$ we may conclude that
\[
  \E \left[ \phi \sigma^k \right] = 
  \sum_{i=1}^n 
  \inf \left\{ \left. \E \left[ w_i \indic_{V_i} \sigma \right] \, \right| \,
                   \sigma \in \di^k_0 \right\} 
  \leq
  \inf \left\{ \left. \E \left[ \phi \sigma \right] \, \right| \,
                   \sigma \in \di^k_0 \right\}.
\]
To prove the reverse inequality, for any $\eps>0$
choose $\sigma_1,\ldots,\sigma_n \in \di^k_0$
such that  for $i=1,\ldots,n$,
$$
  \E \left[ w_i \indic_{V_i} \sigma \right] >
  \E \left[ w_i \indic_{V_i} \sigma_i \right] - \frac{\eps}{n}.
$$
The function $\sigma = \min \{ \sigma_1,\ldots,\sigma_n \}$
belongs to $\di^k_0$ and satisfies 
$ 
  \E \left[ \phi \sigma^k \right] > 
  \E \left[ \phi \sigma \right] - \eps.
$
As $\eps>0$ was arbitrarily small, we conclude that
$\E \left[ \phi \sigma^k \right] \geq 
\inf \left\{ \left. \E \left[ \phi \sigma \right] \, \right| \,
                   \sigma \in \di^k \right\}$
which completes the proof of the lemma for 
the special case of simple functions.

For the general case, we may use the equation
$$\E \left[ \phi \sigma \right] = \E \left[ \phi v^+ \right] - \E \left[ \phi (v^+ - \sigma) \right]$$
to see that the lemma is equivalent to the assertion that
\begin{equation} \label{eq:sigmak2.1}
  \E \left[ \phi (v^+ - \sigma^k) \right] =
  \sup \left\{ \E \left[ \phi (v^+ - \sigma) \right] \mid \sigma \in \di^k_0 \right\}.
\end{equation}
It will be more convenient to work with this form of the lemma
because the function $v^+ - \sigma$ is non-negative for $\sigma \in \di^k_0$.
In particular, the set function $\nu(U) = \int_{U} (v^+ - \sigma^k) \, d\mu$
is a non-negative measure on $(\Omega,\sigfld)$ and therefore
$$
  \E \left[ \phi (v^+ - \sigma) \right] = \int_U \phi \, d\nu
       = \sup \left\{ \int_U \phi' \, d\nu \mid 0 \leq \phi' \leq \phi, \phi' \mbox{ a simple function} \right\}.
$$
Letting $\mathcal{S}$ denote the set of simple functions $\phi'$ that satisfy 
$0 \leq \phi' \leq \phi$, and recalling that~\eqref{eq:sigmak2.1} was already
shown to hold for simple functions, we now find that
\begin{align*}
  \E \left[ \phi (v^+ - \sigma^k) \right] &= 
  \sup \left\{ \left.
            \E \left[ \phi' (v^+ - \sigma^k) \right] 
            \, \right| \,
            \phi' \in \mathcal{S} \right\} \\
&=
  \sup \left\{ \left.
            \E \left[ \phi' (v^+ - \sigma) \right] 
            \, \right| \,
            \phi' \in \mathcal{S}, \, \sigma \in \di^k_0 \right\} \\
&=
  \sup \left\{ \left.
            \E \left[ \phi (v^+ - \sigma) \right] 
            \, \right| \,
            \sigma \in \di^k_0  \right\},
\end{align*}
which completes the proof of~\eqref{eq:sigmak2.1}.
\end{proof}

As a first application of Lemma~\ref{lem:sigmak2}, we will prove
that $\sigma^k$ is $\sigfld^{k,0}$-measurable; this strengthens
Lemma~\ref{lem:sigmak}, which only asserts that 
$\sigma^k$ is $\sigfld^{k,\infty}$-measurable.
The fact that $\sigma^k$ is $\sigfld^{k,0}$-measurable 
can be informally summarized as stating that the value of 
the stage-$k$ strike price only depends
on information learned during the first $k$ stages of
inspection, not on information learned during the 
passage of ``clock time''. This property is intuitive,
since the information learned during the passage of 
clock time is conditionally independent of the bidder's
value and inspection costs. It wil also turn out to be an
important property of $\sigma^k$ when it comes to
analyzing equilibria of the descending-price auction.

\begin{lemma} \label{lem:k0-measurable}
The function $\sigma^k$ is $\sigfld^{k,0}$-measurable.
\end{lemma}
\begin{proof}
Let $\hat{\sigma}^k = \E \left[ \sigma^k \given \sigfld^{k,0} \right]$.
By construction, $\hat{\sigma}^k$ is $\sigfld^{k,0}$-measurable.
Let us prove that it satisfies~\eqref{eq:def-sigmak} for every 
$U \in \sigfld^{k,\infty}$. We will make repeated use of the following
identity: if $f$ is $\sigfld^{k,0}$-measurable and $g$ is $\sigfld$-measurable
then
\begin{equation} \label{eq:measurable.1}
  \E \left[ f g \right] = 
  \E \left[ \E \left[ f g \given \sigfld^{k,0} \right] \right] =
  \E \left[ f \E \left[ g \given \sigfld^{k,0} \right] \right].
\end{equation}

Given any $U \in \sigfld^{k,\infty}$, let 
$\phi = \E \left[ \indic_U \given \sigfld^{k,0} \right]$.
%For any $\sigfld^{k,0}$-measurable function $f$ we have
%\begin{equation} \label{eq:measurable.1}
%  \int_U f \, d\mu = \E \left[ \indic_U f \right] = 
%  \E \left[ \E \left[ \indic_U f \given \sigfld^{k,0} \right] \right] =
%  \E \left[ \E \left[ \indic_U \given \sigfld^{k,0} \right] f \right] = 
%  \E \left[ \phi f \right].
%\end{equation}
Recalling that every function in $\di^k$ is $\sigfld^{k,0}$-measurable,
as is $\hat{\sigma}^k$, we find that
\begin{align*}
  \int_U \hat{\sigma}^k \, d\mu &= \E \left[ \indic_U \hat{\sigma}^k \right] 
  = \E \left[ \phi \hat{\sigma}^k \right] = \E \left[ \phi \sigma^k \right] \\
  &= \inf \left\{ \E \left[ \phi \sigma \right] \mid \sigma \in \di^k_0 \right\} 
    = \inf \left\{ \E \left[ \indic_U \sigma \right] \mid \sigma \in \di^k_0 \right\}
    = \inf \left\{ \left. \int_U \sigma \, d\mu \, \right| \, \sigma \in \di^k_0 \right\},
\end{align*}
where we have applied~\eqref{eq:measurable.1} twice on the first line
and once on the second line.
Since $\hat{\sigma}^k$ satisfies~\eqref{eq:def-sigmak}, the
uniqueness assertion in Lemma~\ref{lem:sigmak} implies that
$\sigma^k = \hat{\sigma}^k$ almost everywhere, and 
consequently (possibly after modifying the values of $\sigma^k$
on a measure-zero set) we may conclude that $\sigma^k$
is $\sigfld^{k,0}$-measurable.
\end{proof}

Related to the issue of $\sigfld^{k,0}$-measurability, we have
the following definition and lemma, which address the question
of when it is possible to simulate an arbitrary inspection policy
with one that performs all of its inspection at time 0.

\begin{definition} 
An inspection policy $\stg$ is {\em prompt} if 
all of its inspection is performed at time 0, 
i.e.~$\stg(\omega,\tau) = \stg(\omega,0)$ for all
$\omega \in \Omega, \tau \in \reals_+$.
\end{definition}

In the following ``prompt simulation lemma'',
the probability space $\overline{\Omega} = \Omega \times [0,1]$ 
is equipped with the product probability measure $\mu \times m$
where $m$ denotes Lebesgue measure on $[0,1]$. We think of
a sample point $(\omega,x) \in \overline{\Omega}$ as consisting
of a sample point $\omega$ from the original probability space,
along with an independent ``random seed'' $x \in [0,1]$ which
may be used for defining a randomized inspection policy. For each $k,\tau$
there are two relevant $\sigma$-fields on $\overline{\Omega}$:
$\overline{\sigfld}^{k,\tau}$ is  the product of $\sigfld^{k,\tau}$ with the Borel
$\sigma$-field on $[0,1]$, whereas $\underline{\sigfld}^{k,\tau}$ is
the $\sigma$-field of all sets of the form $U \times [0,1]$ for
$U \in \sigfld^{k,\tau}$. Thus, a $\overline{\sigfld}^{k,\tau}$-measurable
function is allowed to depend on the random seed $x$, whereas a
$\underline{\sigfld}^{k,\tau}$-measurable function may only depend 
on $\omega$. Note that for every $\sigfld^{k,\tau}$-measurable
function $f$ on $\Omega$ there is a corresponding
$\underline{\sigfld}^{k,\tau}$-measurable function $\underline{f}$
on $\overline{\Omega}$, defined by 
$\underline{f}(\omega,x) = f(\omega)$. In a 
slight abuse of notation, we will ignore the 
distinction between $f$ and $\underline{f}$.

\begin{lemma}[Prompt Simulation Lemma]
\label{lem:prompt}
If $\stg$ is any inspection policy, then there is a prompt
inspection policy $\altstg$ on $\left(\overline{\Omega}, 
\left( \overline{\sigfld}^{k,\tau} \right) \right)$ such that
for all $k \in \mathbb{N} \cup \{\infty\}$,
\begin{equation} \label{eq:prompt.1}
  \Pr \left[ \stg(\infty)=k \given \sigfld^{\infty,0} \right]
=
  \Pr \left[ \altstg(\infty)=k \given \underline{\sigfld}^{\infty,0} \right]
\end{equation}
It follows that 
\begin{equation} \label{eq:prompt.2}
  \E \left[ \allocstg v - c^\stg \given \sigfld^{\infty,0} \right] 
=
  \E \left[ \alloc{\altstg} v - c^{\altstg} \given \underline{\sigfld}^{\infty,0} \right]
\end{equation}
and, for any $\sigma \in \sigfld^{\infty,0}$, 
\begin{equation} \label{eq:prompt.3}
  \E \left[ \allocstg \sigma \given \sigfld^{\infty,0} \right] =
  \E \left[ \alloc{\altstg} \sigma \given \underline{\sigfld}^{\infty,0} \right],
\end{equation}
\end{lemma}
\begin{proof}
The inspection policy $\altstg$ on 
$\overline{\Omega} = \Omega \times [0,1]$
is defined as follows. For any sample point
$(\omega,x) \in \overline{\Omega}$ and
any $\tau \in \reals_+$ we define 
$\altstg((\omega,x),\, \tau)$ to be the
greatest $k \in \mathbb{N} \cup \{\infty\}$ 
such that
$$
  \Pr \left[ \stg(\infty) \geq  k \given \sigfld^{\infty,0} \right] \geq x.
$$
Let us first verify that $\altstg$ satisfies
the definition of a prompt inspection policy,
and then verify that~\eqref{eq:prompt.1}
holds. Promptness of $\altstg$ is trivial,
since the definition of $\altstg((\omega,x), \, \tau)$
has no dependence on $\tau$, so we only need
to check that for all $k$,
$\{ (\omega,x) \mid \altstg((\omega,x), 0) > k\}$
is $\overline{\sigfld}^{k,0}$-measurable. 
From the definition of $\altstg$ we see that 
$\altstg((\omega,x), \, 0) > k$ holds if and only if the inequality
$\Pr \left[ \stg(\infty) > k \given \sigfld^{\infty,0} \right] \geq x$
holds at $\omega$. Let $h(\omega)$ denote the conditional
probability $\Pr \left[ \stg(\infty) > k \given \sigfld^{\infty,0} \right]$.
By definition $h$ is $\sigfld^{\infty,0}$-measurable; we claim that
it is, in fact, $\sigfld^{k,0}$-measurable. Indeed,
the set $E = \{ \omega \in \Omega \mid
\stg(\omega,\infty) > k \}$ belongs to $\sigfld^{k,\infty}$ by the
definition of an inspection policy. Applying
the conditional independence relation~\eqref{eq:cond-indep},
we find that
$$
  \Pr \left[ \stg(\infty) > k \given \sigfld^{\infty,0} \right] = 
  \Pr \left[ \stg(\infty) > k \given \sigfld^{k,0} \right],
$$
so $h(\omega)$ can be equivalently defined as
$\Pr \left[ \stg(\infty) > k \given \sigfld^{k,0} \right]$,
from which it is clear that $h$ is $\sigfld^{k,0}$-measurable.
Consequently the function $h(\omega)-x$
is $\overline{\sigfld}^{k,0}$-measurable, and the set
$W = \{ (\omega,x) \mid h(\omega) - x \geq 0 \}$ belongs
to $\sigfld^{k,0}$. Recalling that $W$ is also the set of
$(\omega,x)$ such that $\altstg((\omega,x), \, 0) > k$,
we see from the relation $W \in \sigfld^{k,0}$ that $\altstg$ 
satisfies the definition of an inspection policy.

To verify~\eqref{eq:prompt.1} recall that the event
$\altstg(\infty) > k$ holds at $(\omega,x)$ if
and only if $h(\omega) - x \geq 0$. As $x$ is independent of $\omega$
and is uniformly distributed in $[0,1]$, the conditional probability
of this event given $\omega$ is simply $h(\omega)$. Thus,
\begin{align}
\nonumber
  \Pr \left[ \altstg(\infty) > k \given \underline{\sigfld}^{\infty,0} \right] 
&=
  \E \left[ \Pr \left[ \altstg(\infty) > k \given \omega \right]
               \, \given \, \underline{\sigfld}^{\infty,0} \right] \\
&=
  \E \left[ h(\omega) \given \underline{\sigfld}^{\infty,0} \right] 
= h(\omega) =
  \Pr \left[ \stg(\infty) > k \given \sigfld^{\infty,0} \right].
 \label{eq:prompt.4}
\end{align}
To derive~\eqref{eq:prompt.1} we simply instantiate 
equation~\eqref{eq:prompt.4} at $k-1$ and $k$, and subtract.

Finally, to verify equations~\eqref{eq:prompt.2}-\eqref{eq:prompt.3}
from the lemma statement,
recall that $v, \sigma,$ and $\{c^k\}_{k \in \mathbb{N} \cup \{\infty\}}$
are $\sigfld^{\infty,0}$-measurable. We therefore have the
equations
\begin{align}
\label{eq:prompt.5a}
  \E \left[ \allocstg v \given \sigfld^{\infty,0} \right]
&=
  \Pr \left[ \stg(\infty)=\infty \given \sigfld^{\infty,0} \right] \cdot v \\
\label{eq:prompt.5b}
  \E \left[ c^\stg \given \sigfld^{\infty,0}  \right] 
&= 
  \sum_{k=0}^{\infty} \Pr \left[ \stg(\infty) = k \given \sigfld^{\infty,0} \right] \cdot c^k \\
\label{eq:prompt.5c}
  \E \left[ \allocstg \sigma \given \sigfld^{\infty,0}  \right]
&=
  \Pr \left[ \stg(\infty)=\infty \given \sigfld^{\infty,0} \right] \cdot \sigma,
\end{align}
Viewing $v, \sigma, \{c^k\}$ as functions defined on $\overline{\Omega}$,
they are $\underline{\sigfld}^{\infty,0}$-measurable, so the same
reasoning justifies
\begin{align}
\label{eq:prompt.6a}
  \E \left[ \alloc{\altstg} v \given \underline{\sigfld}^{\infty,0} \right] &=
  \Pr \left[ \altstg(\infty)=\infty \given \underline{\sigfld}^{\infty,0} \right] \cdot v \\
\label{eq:prompt.6b}
  \E \left[ c^\altstg \given \underline{\sigfld}^{\infty,0} \right] 
&= 
  \sum_{k=0}^{\infty} \Pr \left[ \altstg(\infty) = k \given \underline{\sigfld}^{\infty,0} \right] \cdot c^k \\
\label{eq:prompt.6c}
  \E \left[ \alloc{\altstg} \sigma \given \underline{\sigfld}^{\infty,0} \right]
&=
  \Pr \left[ \altstg(\infty)=\infty \given \underline{\sigfld}^{\infty,0} \right] \cdot \sigma.
\end{align}
Using~\eqref{eq:prompt.1} we find that the right sides 
of~\eqref{eq:prompt.5a}-\eqref{eq:prompt.5c} are equal
to the right sides of the corresponding equations~\eqref{eq:prompt.6a}-\eqref{eq:prompt.6c},
which concludes the proof of~\eqref{eq:prompt.2}-\eqref{eq:prompt.3}.
\end{proof}

Given that $\sigma^k$ is informally defined as the value of an outside option
that makes the bidder indifferent between stopping and 
inspecting at stage $k$, it is intuitive that $\sigma^k$
discourages inspection at stage $k$. The following lemma substantiates
this intuition.

\begin{lemma} \label{lem:strike-discourages}
The random variable $\sigma^k$ discourages inspection at stage $k$.
%Furthermore, letting $W$ denote the set 
%$\{\omega \mid \sigma^k(\omega) > -\infty\}$, for any $\eps>0$
%there exists an
%inspection policy $\stg$ which is pointwise {\em strictly}
%greater than $k$ such that the inequality
%\[
%    \E \left[ \allocstg (v - \sigma^k) - (c^\stg - c^k) \given
%              \sigfld^{k,\infty} \right] > -\eps
%\]
%holds pointwise on $W$.
\end{lemma}
\begin{proof}
For all $V \in \sigfld^{k,\infty}$,
\begin{align*}
\E \left[ \indic_V \left( \allocstg v - c^\stg \right) \right]
& \leq
\inf \left\{ \left. \E \left[ \indic_V \left( \allocstg \sigma
                               - c^k \right) \right] \, \right| \,
       \sigma \in \di^k \right\} \\
& = 
\inf \left\{ \left. \E \left[ \indic_V \left( 
                  \E \left[ \allocstg \given \sigfld^{k,\infty} \right] \,
	       \sigma - c^k \right) \right] \, \right| \,
       \sigma \in \di^k \right\} \\
& =
\E \left[ \indic_V \left( 
                  \E \left[ \allocstg \given \sigfld^{k,\infty} \right] \,
	       \sigma^k - c^k \right) \right] \\
& = 
\E \left[ \indic_V \left( \allocstg \sigma^k - c^k \right) \right],
\end{align*}
where the first line used the definition of $\di^k$, the second and
fourth lines are applications of the law of iterated conditional expectation,
and the third line uses Lemma~\ref{lem:sigmak2}. Since $V$ was an
arbitrary element of $\sigfld^{k,\infty}$ we conclude that
$\E \left[ \allocstg v - c^\stg \given \sigfld^{k,\infty} \right]
  \leq \E \left[ \allocstg \sigma^k - c^k \given \sigfld^{k,\infty} \right],$
i.e., $\sigma^k$ discourages inspection at stage $k$.
\end{proof}

%The following lemma generalizes Lemma~\ref{lem:strike-discourages} by
%stating that for any inspection policy $\altstg$,
%the strike price $\sigma^{\altstg}$, corresponding to the random stage of 
%inspection at which $\altstg$ stops, discourages inspection beyond
%$\altstg$.
\begin{lemma} \label{lem:continuation}
Suppose $k$ is a natural number  and 
$\altstg, \stg$ are any two inspection policies that
satisfy $k \leq \altstg(\infty) \leq \stg(\infty)$ pointwise,
and suppose that $\alloc{\altstg}=\allocstg$ holds
at every sample point where $\altstg(\infty)=\stg(\infty)$.
Let $\sigma^{\altstg}$ denote a random variable 
whose value is equal to $\sigma^{\altstg(\infty)}$
if $\altstg(\infty) < \infty$ and equal to $v$ otherwise.
Then we have
\begin{equation} \label{eq:continuation}
  \E \left[ \allocstg v - c^\stg \given \sigfld^{k,\infty} \right] \leq
%   \E \left[ \allocstg \sigma^{\altstg} - c^\altstg \given \sigfld^{k,\infty} \right] = 
  \E \left[ \alloc{\altstg} v - c^\altstg \given \sigfld^{k,\infty} \right] + 
  \E \left[ \left( \allocstg - \alloc{\altstg} \right) \sigma^\altstg \given \sigfld^{k,\infty} \right].
\end{equation}
\end{lemma}
Note that the lemma generalizes Lemma~\ref{lem:strike-discourages} 
because in the case $\altstg(\infty) \equiv k$ we have
$\alloc{\altstg} \equiv 0,$ so the right side of~\eqref{eq:continuation}
equals $\E \left[ \allocstg \sigma^k - c^k \given \sigfld^{k,\infty} \right]$.
\begin{proof}
The set of sample points where $\stg(\infty) > \altstg(\infty)$
can be partitioned into sets $V_m \in \sigfld^{m,\infty}$ where
$\stg(\infty) > \altstg(\infty) = m \geq k$. On each $V_m$ we have
\begin{equation} \label{eq:continuation.1}
  \E \left[ \indic_{V_m} ( \allocstg v - c^\stg ) 
         \given \sigfld^{k,\infty} \right]
\leq
  \E \left[ \indic_{V_m} ( \allocstg \sigma^{m} - c^{m}) 
         \given \sigfld^{k,\infty} \right]
=
  \E \left[ \indic_{V_m} \left( \allocstg \sigma^{\altstg} - c^{\altstg} 
          \right) \given \sigfld^{k,\infty} \right]
\end{equation}
where the first inequality is because $\sigma^m$ 
discourages inspection at stage $m$. 
%Define an additional set $V_{k-1}$ to consist of
%all sample points where $\altstg(\infty)=\infty$
%and $\alloc{\altstg}=0$. By assumption, $v < \sigma^k$
%pointwise on $V_{k-1}$. Now, letting $V = \bigcup_{m=k-1}^{\infty} V_m$,
Now,
\allowdisplaybreaks
\begin{align*}
  \E \left[ \allocstg v - c^\stg \given \sigfld^{k,\infty} \right] &=
  \E \left[ \indic_{\Omega \setminus V} 
              \left( \allocstg v - c^\stg \right) 
              \given \sigfld^{k,\infty} \right] +
  \sum_{m=k}^{\infty} 
  \E \left[ \indic_{V_m} \left( \allocstg v - c^\stg \right)
       \given \sigfld^{k,\infty} 
      \right] \\
& =
  \E \left[ \indic_{\Omega \setminus V}
              \left( \alloc{\altstg} v - 
                 c^{\altstg} \right)
       \given \sigfld^{k,\infty} \right] +
  \sum_{m=k}^{\infty}
  \E \left[ \indic_{V_m} \left( \allocstg v - c^\stg \right)
       \given \sigfld^{k,\infty} \right] \\
& =
  \E \left[ \alloc{\altstg} v  \given \sigfld^{k,\infty} \right] -
  \E \left[ \indic_{\Omega \setminus V} 
              c^{\altstg}  \given \sigfld^{k,\infty} \right] -
  \sum_{m=k}^{\infty}
  \E \left[ \indic_{V_m} c^{\altstg}  \given \sigfld^{k,\infty} \right] \\
& \qquad +
  \sum_{m=k}^{\infty}
  \E \left[ \indic_{V_m} \left( \allocstg v - c^\stg 
             + c^{\altstg} \right)  \given \sigfld^{k,\infty} \right] \\
& =
  \E \left[ \alloc{\altstg} v -
              c^{\altstg}  \given \sigfld^{k,\infty} \right] +
  \sum_{m=k}^{\infty}
  \E \left[ \indic_{V_m} \left( \allocstg v - c^\stg 
             + c^{\altstg} \right)  \given \sigfld^{k,\infty} \right] \\
& \leq
  \E \left[ \alloc{\altstg} v - c^{\altstg}  \given \sigfld^{k,\infty} \right]
  + \sum_{m=k}^{\infty}
  \E \left[ \indic_{V_m} \allocstg \sigma^{\altstg}  \given \sigfld^{k,\infty} \right] 
  \qquad   \mbox{by inequality~\eqref{eq:continuation.1}} \\
& =
  \E \left[ \alloc{\altstg} v - c^\altstg  \given \sigfld^{k,\infty} \right] +
  \E \left[ \left( \allocstg - \alloc{\altstg} \right) \sigma^{\altstg}  \given \sigfld^{k,\infty}  \right] 
% \\ & =
%  \E \left[ \allocstg \sigma^{\altstg} - c^{\altstg}  \given \sigfld^{k,\infty} \right] 
\end{align*}
which completes the proof.
\end{proof}

Informally, we defined $\sigma^k$ as the value of an outside option 
that makes the bidder indifferent between undertaking at least one more
stage of inspection, or ceasing immediately at stage $k$. Thus far, in 
Lemma~\ref{lem:strike-discourages} we have shown that the bidder 
does not strictly prefer to undertake at least one more stage of 
inspection. We now turn to showing the opposite inequality: that the
bidder does not strictly prefer to cease inspection immediately. We prove
this fact in two steps: Lemma~\ref{lem:patient} shows that if the
value of the outside option is decreased by any constant $\delta>0$
then the bidder strictly prefers to continue advancing the inspection
stage beyond $k$; Lemma~\ref{lem:indifferent} shows that even
if the outside option value is {\em exactly} $\sigma^k$ then
there is an inspection policy that always advances beyond stage $k$
and is no worse than claiming the outside option after exactly $k$ 
stages of inspection.

\begin{lemma} \label{lem:patient}
For all $k \geq 0$ and $\delta>0$, there exists a prompt inspection
policy $\altstg$ such that $\altstg(\omega,\tau) > k$ 
for all $\omega,\tau$ and
\begin{equation} \label{eq:patient}
   \E \left[ \alloc{\altstg} v - c^{\altstg} 
               \given \sigfld^{k,\infty} \right]  \geq
   \E \left[ \alloc{\altstg} (\sigma^k - \delta) - c^k
               \given \sigfld^{k,\infty} \right] 
\end{equation}
almost everywhere.
\end{lemma}
\begin{proof}
Let $\mathcal{G}$ denote the collection of all
sets $U \in \sigfld^{k,0}$ such that there exists a 
prompt inspection policy $\altstg$ with 
$\altstg(\omega,\tau) > k$ that satisfies~\eqref{eq:patient} 
pointwise on $U$. We claim that
\begin{enumerate}
\item $\mathcal{G}$ is closed under countable unions;
\item for any set $V \in \sigfld^{k,0}$ such that $\mu(V)>0$
there exists $W \in \mathcal{G}$ such that $W \subseteq V$ and $\mu(W)>0$.
\end{enumerate}
Treating these two claims as given for the moment, let us see
how they imply the lemma. Let $\gamma = \sup \{ \mu(U) \mid U \in \mathcal{G} \}$.
We may choose sets $U_1,U_2,\ldots \in \mathcal{G}$ such that
$\mu(U_n) > \gamma - \frac{1}{n}$ for all $n \geq 1$. The set
$U = \bigcup_{n \geq 1} U_n$ belongs to $\mathcal{G}$ and
satisfies $\mu(U) \geq \mu(U_n)$ for all $n$, so it must be the
case that $\mu(U) = \gamma$. Let $V$ denote the complement
of $U$. If $\mu(V) > 0$ then there exists $W \in \mathcal{G}$ such
that $W \subseteq V$ and $\mu(W)>0$; then $U \cup W \in \mathcal{G}$
and $\mu(U \cup W) > \gamma$, contradicting the definition of $\gamma$.
Therefore it must be the case that $\gamma=1$. From the definition of
$\mathcal{G}$ we know that there is a prompt inspection policy $\altstg$
satisfying~\eqref{eq:patient} pointwise on
$U$. Since $\mu(U) = \gamma = 1$ it follows that 
$\altstg$ satisfies~\eqref{eq:patient} almost everywhere.

It remains to prove the two claims about $\mathcal{G}$ enumerated 
above. Let $U_1,U_2,\ldots$ be any countable collection of sets in 
$\mathcal{G}$, and for each $n \geq 1$ let $\altstg_n$ be a prompt
inspection policy that satisfies~\eqref{eq:patient} pointwise on $U_n$
and such that $\altstg_n(\omega,\tau) > k$ for all $\omega,\tau$.
Let $V_n = U_n \setminus \left( \bigcup_{m<n} U_m \right)$, and 
define
$$
  \altstg(\omega,\tau) = \begin{cases}
    \altstg_n(\omega,\tau) & \text{if } \omega \in V_n \\
    0 & \text{otherwise.}
  \end{cases}
$$
(Note that the sets $\{V_n\}_{n \geq 1}$ are pairwise disjoint,
so the cases in the definition of $\altstg(\omega,\tau)$ are mutually exclusive.)
It is easy to verify that $\altstg$ satisfies the definition of a prompt
inspection policy, and that it satisfies~\eqref{eq:patient} pointwise
on $U$. Thus, $\mathcal{G}$ is closed under countable unions.

Now suppose $V \in \sigfld^{k,0}$ and $\mu(V) > 0$. The
function $\sigma = \sigma^k - \delta \indic_V$ satisfies 
$\int_V \sigma \, d\mu < \int_V \sigma^k \, d\mu$, so by
Lemma~\ref{lem:sigmak} it must be the case that 
$\sigma \not\in \di^k$: there exists an inspection policy
$\stg$ satisfying $\stg(\infty) > k$ pointwise, such that 
\begin{equation} \label{eq:patient.2}
  \E \left[ \allocstg v - c^\stg \given \sigfld^{k,\infty} \right] 
> 
  \E \left[ \allocstg \sigma - c^k \given \sigfld^{k,\infty} \right]
\end{equation}
holds on a set $U$ of positive measure. By Lemma~\ref{lem:prompt}
we may assume without loss of generality that $\stg$ is a
prompt inspection policy, since $\sigma$ is $\sigfld^{k,0}$-measurable.
Then all of the random variables appearing in~\eqref{eq:patient.2} ---
namely $\allocstg,\, c^\stg,\, c^k,\, v,\, \sigma$ --- are $\sigfld^{\infty,0}$-measurable.
By the conditional independence property~\eqref{eq:cond-indep} we may
conclude that
\begin{align*}
  \E \left[ \allocstg v - c^\stg \given \sigfld^{k,\infty} \right] 
&=
  \E \left[ \allocstg v - c^\stg \given \sigfld^{k,0} \right] 
\\
  \E \left[ \allocstg \sigma - c^k \given \sigfld^{k,\infty} \right]
&=
  \E \left[ \allocstg \sigma - c^k \given \sigfld^{k,0} \right]
\end{align*}
and hence the set $U$ of sample points that satisfy~\eqref{eq:patient.2}
is a $\sigfld^{k,0}$-measurable set. Observe that $\sigma = \sigma^k$ 
on the complement of $V$, and Lemma~\ref{lem:strike-discourages} 
precludes the possibility that~\eqref{eq:patient.2} holds on any 
positive-measure set where $\sigma = \sigma^k$. Thus, $U \subseteq V$.
At this point, we have proven that $U$ is a $\sigfld^{k,0}$-measurable
subset of $V$ with positive measure, and that inspection policy 
$\stg$
satisfies~\eqref{eq:patient} pointwise on $U$
(since $\sigma = \sigma^k - \delta$ pointwise on $U$). 
Thus, $V$ contains a positive-measure subset that belongs to $\mathcal{G}$,
as claimed.
\end{proof}

\begin{lemma} \label{lem:indifferent}
Let $\qstg_k$ denote an inspection policy defined as follows: 
for all $\omega,\tau$, $\qstg_k(\omega,\tau)$ is the least
$\ell > k$ such that $\sigma^\ell \leq \sigma^k$. If no such
$\ell$ exists then $\qstg_k(\omega,\tau) = \infty$. Furthermore
$\alloc{\qstg_k}=1$ if and only if $\qstg_k(\infty)=\infty$ and
$v > \sigma^k$. For the inspection policy thus defined, we have
\begin{equation} \label{eq:indifferent}
  \E \left[ \alloc{\qstg_k} v - c^{\qstg_k} \given \sigfld^{k,\infty} \right]
=
  \E \left[ \alloc{\qstg_k} \sigma^k - c^k \given \sigfld^{k,\infty} \right].
\end{equation}
\end{lemma}
\begin{proof}
First observe that $\qstg_k$ is a well-defined prompt inspection policy:
for any $j \in \mathbb{N}$ the event $\{\qstg_k(\infty) > j\}$ is
equal to the intersection of the events $\{\sigma^k < \sigma^i\}$ for 
$i = k+1,k+2,\ldots,j,$ which belongs to $\sigfld^{j,0}$ because
$\sigma^k$ and $\sigma^i$ are both $\sigfld^{j,0}$-measurable.
Next, observe that the left side of~\eqref{eq:indifferent} is pointwise
less than or equal to the right side, by Lemma~\ref{lem:strike-discourages}.
To prove that the reverse inequality holds pointwise, we fix any
$\delta > 0$ and aim to prove that 
\begin{equation} \label{eq:indiff.0}
  \E \left[ \alloc{\qstg_k} v - c^{\qstg_k} \given \sigfld^{k,\infty} \right]
>
  \E \left[ \alloc{\qstg_k} \sigma^k - c^k \given \sigfld^{k,\infty} \right] - 2 \delta
\end{equation}
holds on a set of measure at least $1-\delta$. If this is true
for all $\delta>0$, then in fact~\eqref{eq:indifferent} must
hold almost everywhere.

For ease of notation, let $\qstg = \qstg_k$.
We shall construct a sequence of prompt
inspection policies $(\altstg_\ell)_{\ell=1}^{\infty}$
inductively, such that for all $\ell$ the inequalities
\begin{align}
\label{eq:indiff.1}
  \altstg_\ell(\infty) & \geq (k + \ell) \wedge \qstg(\infty) \\
\label{eq:indiff.2}
  \E \left[ 
    \alloc{\altstg_\ell} v - c^{\altstg_\ell} \given
    \sigfld^{k,\infty} \right]
   & >
  \E \left[
    \alloc{\altstg_\ell} \sigma^k - c^{k} \given
    \sigfld^{k,\infty} \right] - \delta + \frac{\delta}{2^\ell}
\end{align}
hold pointwise. For $\ell=1$ the existence of such a policy
$\altstg_\ell$ is implied by Lemma~\ref{lem:patient}. Assume
now that $\altstg_{\ell}$ is already defined. Use
Lemma~\ref{lem:patient} to obtain a 
prompt inspection policy
$\altstg$ that satisfies $\altstg(\infty) > k+\ell$ and
\begin{equation} \label{eq:indiff.3}
  \E \left[ \alloc{\altstg} v - c^\altstg \given \sigfld^{k+\ell,\infty} \right]
  \geq
  \E \left[ \alloc{\altstg} \sigma^{k+\ell}
              - c^{k+\ell} \given \sigfld^{k+\ell,\infty} \right]
  - \frac{\delta}{2^{\ell+1}}
\end{equation}
pointwise. Letting $V$ denote the
set of sample points where $\altstg_\ell(\infty) = k+\ell < \qstg(\infty)$,
 we define $\altstg_{\ell+1}$ to be equal to $\altstg$ on $V$
and equal to $\altstg_\ell$ on the complement of $V$. 
First let us verify that
$\altstg_{\ell+1}$ is a prompt inspection policy. 
Promptness of $\altstg_{\ell+1}$ 
follows from the fact that $\qstg, \altstg_\ell, \altstg$
are all prompt. The first and third properties of an 
inspection policy are trivial to verify. For the second
property, fix $\tau,j$ and consider the set of $\omega$ such that
$\altstg_{\ell+1}(\omega,\tau) > j$. If $j \leq k+\ell$ then
this is equal to the set of $\omega$ such that $\altstg_\ell(\omega,\tau) > j$
or $\qstg(\omega,\tau) > j$, hence it is $\sigfld^{j,\tau}$-measurable.
If $j > k+\ell$ there are two ways that $\altstg_{\ell+1}(\omega,\tau)$
could be greater than $j$: either 
$\altstg_\ell(\omega,\tau) > j$ or
$\altstg_\ell(\omega,\tau) = k+\ell < \qstg(\omega,\tau)$ 
and $\altstg(\omega,\tau) > j$. Both of these events are 
$\sigfld^{j,\tau}$-measurable. This completes the verification
that $\altstg_{\ell+1}$ is a valid inspection policy. 

To verify~\eqref{eq:indiff.1}, note that $V$ is precisely the
set of sample points at which $\altstg_\ell(\infty) < (k+\ell+1) \wedge 
\qstg(\infty)$. Thus, on the complement of $V$ we have
$\altstg_{\ell+1}(\infty) = \altstg_\ell(\infty) \geq (k+\ell+1) \wedge
\qstg(\infty)$. Meanwhile, on $V$, we have 
$\altstg_{\ell+1}(\infty) = \altstg(\infty) \geq k+\ell+1 \geq
(k+\ell+1) \wedge \qstg(\infty)$.
To verify~\eqref{eq:indiff.2}, first observe that 
% $\alloc{\altstg_{\ell+1}} = \alloc{\altstg_{\ell}} + \indic_V
% \alloc{\altstg}$ holds pointwise. 
on the complement of $V$ we have 
$\alloc{\altstg_{\ell+1}} = \alloc{\altstg_\ell}$
and $c^{\altstg_{\ell+1}} = c^{\altstg_\ell}$,
whereas on $V$ we have
$\alloc{\altstg_\ell} = 0$, $c^{\altstg_\ell} = c^{k+\ell}$,
$\alloc{\altstg_{\ell+1}} = \alloc{\altstg}$, and
$c^{\altstg_{\ell+1}} = c^{\altstg}$.
Combining these observations with the fact that 
$V \in \sigfld^{k+\ell,\infty}$ we obtain
\begin{align}
\nonumber
  \E \left[ \alloc{\altstg_{\ell+1}} v - c^{\altstg_{\ell+1}}
    \given \sigfld^{k+\ell,\infty} \right] 
&=
  \E \left[ \alloc{\altstg_{\ell}} v - c^{\altstg_\ell}
    \given \sigfld^{k+\ell,\infty} \right] +
  \indic_V
  \E \left[ \alloc{\altstg} v - c^{\altstg} + c^{k+\ell} 
    \given \sigfld^{k+\ell,\infty} \right]
\\
\nonumber
& > 
  \E \left[ \alloc{\altstg_{\ell}} v - c^{\altstg_\ell}
    \given \sigfld^{k+\ell,\infty} \right] +
  \indic_V \E \left[ \alloc{\altstg} \sigma^{k+\ell} 
    \given \sigfld^{k+\ell,\infty} \right] - \frac{\delta}{2^{\ell+1}} \\
\label{eq:indiff.4}
  & >
  \E \left[ \alloc{\altstg_{\ell}} v - c^{\altstg_\ell}
    \given \sigfld^{k+\ell,\infty} \right] +
  \indic_V \E \left[ \alloc{\altstg} \sigma^{k} 
    \given \sigfld^{k+\ell,\infty} \right] - \frac{\delta}{2^{\ell+1}}.
\end{align}
In the last line, we have used the fact that 
$\sigma^{k+\ell} > \sigma^k$ pointwise on $V$,
which is a consequence of the definition of $\qstg = \qstg_k$
and of the fact that $\qstg(\infty) > k+\ell$ pointwise on $V$.
Now, taking the conditional expectation of both sides 
of~\eqref{eq:indiff.4} with respect to $\sigfld^{k,\infty}$
and using the fact that $\alloc{\altstg_{\ell+1}} = 
\alloc{\altstg_\ell} + \indic_V \alloc{\altstg}$ pointwise,
we obtain
\begin{align}
\nonumber
  \E \left[ \alloc{\altstg_{\ell+1}} v - c^{\altstg_{\ell+1}}
    \given \sigfld^{k,\infty} \right] 
& \geq
  \E \left[ \alloc{\altstg_{\ell}} \sigma^k - c^{k}
    \given \sigfld^{k,\infty} \right] +
  \E \left[ \indic_V \alloc{\altstg} \sigma^{k} 
    \given \sigfld^{k,\infty} \right] - \frac{\delta}{2^{\ell+1}} \\
& >
  \E \left[ \left(
    \alloc{\altstg_{\ell}} + \indic_V \alloc{\altstg}
    \right) \sigma^k - c^k \given \sigfld^{k,\infty} \right] 
     - \delta + \frac{\delta}{2^{\ell}} 
    - \frac{\delta}{2^{\ell+1}} \\
& =
  \E \left[ \alloc{\altstg_{\ell+1}} \sigma^k - c^k \given \sigfld^{k,\infty}
       \right] - \delta + \frac{\delta}{2^{\ell+1}},
\label{eq:indiff.5}
\end{align}
which completes the induction.

Now, the monotone convergence theorem
implies that $\E \left[ c^{k+\ell} \right] \to \E \left[ c^\infty \right]$
as $\ell \to \infty$, and hence we may choose $\ell$ large 
enough that $\E \left[ c^\infty - c^{k+\ell} \right] < \delta^2$.
Then by Markov's Inequality, the relation
$\E \left[ c^\infty - c^{k+\ell} \given \sigfld^{k,\infty} \right] < \delta$
holds pointwise on a set of measure at least $1-\delta$.
For the inspection policy $\qstg \vee \altstg_\ell$ we have
\begin{equation} \label{eq:indiff.6}
  \left( \alloc{\qstg \vee \altstg_\ell} - \alloc{\altstg_\ell} \right) v
\geq
  \left( \alloc{\qstg \vee \altstg_\ell} - \alloc{\altstg_\ell} \right) \sigma^k
\end{equation}
because the construction of $\qstg$ ensures that 
$v \geq \sigma^k$ at any point where $\alloc{\qstg} = 1$.
Combining~\eqref{eq:indiff.2} with~\eqref{eq:indiff.6} we obtain
\begin{align}
\nonumber
  \E \left[ \alloc{\qstg \vee \altstg_\ell} v - c^{\qstg \vee \altstg_\ell}
    \given \sigfld^{k,\infty} \right]
  & =
  \E \left[ \alloc{\altstg_\ell} v - c^{\altstg_\ell} \;\; + \;\;
  \left( \alloc{\qstg \vee \altstg_\ell} - \alloc{\altstg_\ell} \right) v \;\; - \;\;
  \left( c^{\qstg \vee \altstg_\ell} - c^{\altstg_\ell} \right) 
  \given \sigfld^{k,\infty} \right] \\
\nonumber
  & >
  \E \left[ \alloc{\altstg_\ell} \sigma^k - c^k  \;\; + \;\;
  \left( \alloc{\qstg \vee \altstg_\ell} - \alloc{\altstg_\ell} \right) \sigma^k
  \;\; - \;\; \left( c^{\infty} - c^{k+\ell} \right) \given \sigfld^{k,\infty} \right]
  - \delta  \\
\nonumber
  & =
  \E \left[ \alloc{\qstg \vee \altstg_\ell} \sigma^k - c^k \given \sigfld^{k,\infty} \right]
  - \E \left[ c^{\infty} - c^{k+\ell} \given \sigfld^{k,\infty} \right] - \delta \\
  & >
  \E \left[ \alloc{\qstg \vee \altstg_\ell} \sigma^k - c^k \given \sigfld^{k,\infty} \right]
  - 2 \delta,
\label{eq:indiff.7}
\end{align}
where all of the inequalities except the last one hold pointwise, and the
last inequality holds pointwise on a set of measure at least $1-\delta$,
by our assumption on $\ell$. Modify $\qstg \vee \altstg_\ell$ into
an inspection policy $\stg$ defined by $\stg(\omega,\tau) = 
(\qstg \vee \altstg_\ell)(\omega,\tau)$ and 
$\allocstg = \indic_{v > \sigma^k} \cdot \alloc{\qstg \vee \altstg_\ell}$.
In other words, $\stg$ uses the same inspection rule as 
$\qstg \vee \altstg_\ell$ but a modified acquisition rule that
only acquires the item if $v > \sigma^k$. Note that
$\left( \alloc{\qstg \vee \altstg_\ell} - \allocstg \right) 
\cdot \left( v - \sigma^k \right) \leq 0$ pointwise. Upon
rearranging terms and taking conditional expectations, this implies
\begin{equation} \label{eq:indiff.8}
  \E \left[ \left(
    \allocstg - \alloc{\qstg \vee \altstg_\ell} 
    \right) v \given \sigfld^{k,\infty} \right]
  \geq
  \E \left[ \left(
    \allocstg - \alloc{\qstg \vee \altstg_\ell}
    \right) \sigma^k \given \sigfld^{k,\infty} \right].
\end{equation}
Summing~\eqref{eq:indiff.7} and~\eqref{eq:indiff.8} we 
find that
\begin{equation} \label{eq:indiff.9}
  \E \left[ \allocstg v - c^{\stg} \given \sigfld^{k,\infty} \right]
  >
  \E \left[ \allocstg \sigma^k - c^k \given \sigfld^{k,\infty} \right]
  - 2 \delta
\end{equation}
on a set of measure at least $1-\delta$.

Now we shall apply
Lemma~\ref{lem:continuation} to the pair of 
inspection policies $\qstg$ and $\stg$. Observe
that $k < \qstg(\infty) \leq \stg(\infty)$ pointwise,
and that $\alloc{\qstg} = \alloc{\stg}$
holds at every sample point where $\qstg(\infty) = \stg(\infty)$.
(If $\qstg(\infty) = \stg(\infty) < \infty$ then $\alloc{\qstg} = \allocstg = 0$,
and if $\qstg(\infty) = \stg(\infty) = \infty$ then
$\alloc{\qstg} = \allocstg = \indic_{v > \sigma^k}$.)
Consequently, the lemma guarantees that
\begin{equation} \label{eq:indiff.10}
  \E \left[ \alloc{\qstg} v - c^{\qstg} \given \sigfld^{k,\infty} \right]
  +
  \E \left[ \left( \allocstg - \alloc{\qstg} \right) \sigma^{\qstg}
    \given \sigfld^{k,\infty} \right] 
 \geq
  \E \left[ \allocstg v - c^\stg \given \sigfld^{k,\infty} \right].
\end{equation}
On a set of measure at least $1-\delta$ we have
\begin{align}
\nonumber
  \E \left[ \allocstg v - c^\stg \given \sigfld^{k,\infty} \right]
  & > 
  \E \left[ \allocstg \sigma^k - c^k \given \sigfld^{k,\infty} \right]
  - 2 \delta \\
\nonumber
  & =
  \E \left[ \alloc{\qstg} \sigma^k - c^k \given \sigfld^{k,\infty} \right]
  + \E \left[ \left( \allocstg - \alloc{\qstg} \right) \sigma^k
        \given \sigfld^{k,\infty} \right] - 2 \delta \\
  & \geq
  \E \left[ \alloc{\qstg} \sigma^k - c^k \given \sigfld^{k,\infty} \right]
  + \E \left[ \left( \allocstg - \alloc{\qstg} \right) \sigma^{\qstg}
        \given \sigfld^{k,\infty} \right] - 2 \delta,
\label{eq:indiff.11}
\end{align}
where the last line was derived using the fact that 
$\sigma^{k} \geq \sigma^{\qstg}$ whenever 
$\allocstg - \alloc{\qstg} > 0$, a consequence 
of the definition of the inspection policy $\qstg$.
Combining~\eqref{eq:indiff.10}
with~\eqref{eq:indiff.11} and canceling the common term on the 
left and right sides, we obtain~\eqref{eq:indiff.1} which
completes the proof.
\end{proof}

\subsection{Generalized covered call value}
\label{genccv-appendix}

In this section we generalize the notion of covered call value
to the setting of multi-stage inspection, and we state and
prove \Cref{lem:general-ubp-restated}, which generalizes the 
amortization lemma (Lemma 1 from main text  Subsection 2.2) that underpins
many of the results in the paper.
As before, the 
lemma asserts that the expected value of items acquired by
the bidder, minus the inspection cost invested in all items,
is bounded above by the expected covered call value of items
acquired by the bidder, and it provides a sufficient condition
for this upper bound to be sharp.

\begin{definition} \label{def:gen-ccv}
For $k \geq 0$ let $\kappa^k = \min\{\sigma^0,\sigma^1,\ldots,\sigma^k\}$.
Let $\kappa^\infty = \lim_{k \to \infty} \{ \kappa^k \}$.
(Note that the limit has a well-defined
value in $\reals \cup \{-\infty\}$ since $\kappa^0,\kappa^1,\ldots$ is a 
non-increasing sequence.) 
The random variable $\kappa = \min \{ \kappa^\infty, v \}$ is
the {\em generalized covered call value}. 
\end{definition}

To illustrate these definitions consider the single-bidder, single-item, single-inspection-stage case. Then $c^0, c^1, \ldots$
degenerates to a single random variable $c = c^1$.
Since there is only one stage of inspection, it is clear
that the optimal inspection policy pays a cost of $c$,
learns the value $v$, and then claims the better of 
$v$ or $\sigma^0$. Thus by its definition as the indifference point for implementing the optimal policy, 
$\E \left[ \max\{v,\sigma^0\} - c \right] = \sigma^0$,
or equivalently,
$\E \left[ (v-\sigma^0)^+ \right] = \E [ c ]$.
Thus, $\sigma^0$ is the same as the strike price 
in Definition 1. On the other hand, $\sigma^1$ is defined by the fact that
a bidder who has already reached inspection
stage 1 (viz. learned her value) is 
indifferent between stopping immediately and claiming a
reward of $\sigma^1$ versus applying
the optimal policy that 
engages in at least one more stage 
of inspection.  Given there are no inspection stages past the first, this policy must simply claim the item of value $v$,
and hence the indifference point $\sigma^1$ is equal to $v$. The same logic
applies to $\sigma^2, \sigma^3, \ldots$; all
of them are equal to $v$.\footnote{ More generally, in environments 
% where the number
% of inspection stages is bounded above by $b$
with at most $b$ stages of inspection, then
$\sigma^b = \sigma^{b+1} = \cdots = v.$}
Thus
the generalized strike price $\sigma^0$
represents the  strike price, $\sigma$, defined in main text Subsection 2.2
while the generalized strike prices 
$\sigma^1,\sigma^2,\ldots$ are all 
equal to the item's value, $v$. The generalized
covered call value $\kappa$ is therefore equal
to $\min \{\sigma, v\}$, matching the definition
from main text Subsection 2.2.  

To generalize the crucial amortization lemma (Lemma 1) to 
environments with multiple stages of
inspection, we will use a property that generalizes the idea of ``exercising in the money''.
\begin{definition} \label{def:non-exposed}
An inspection policy is {\em non-exposed} if it never ceases inspection
at a stage when the strike price is strictly greater than the minimum
strike price encountered along the inspection path,
% never acquires an item of value $v < \kappa$, 
and it never
fails to acquire an item of value $v > \kappa$ if it
completes all stages of inspection. More formally,
$\stg$ is non-exposed if
\begin{enumerate}
\item 
$
    \Pr \left( \stg(\infty) < \infty \mbox{ and } \sigma^{\stg(\infty)} > \kappa^{\stg(\infty)} \right) = 0
$
%% \item
%% $
%%     \Pr \left( \allocstg (v - \kappa) < 0 \right) = 0
%% $
\item
$    \Pr \left( \stg(\infty) = \infty \mbox{ and } \left(1 - \allocstg \right) (v - \kappa) > 0 \right) = 0.
$
\end{enumerate}
\end{definition}

For non-exposed policies we have the following strengthening
of Lemma~\ref{lem:continuation}.

\newcommand{\Em}{{{\E}_m}}

\begin{lemma} \label{lem:non-exp-contin}
Suppose $\altstg, \qstg$ are any two inspection policies that
satisfy $\altstg(\infty) \leq \qstg(\infty)$ 
and $\alloc{\altstg} \leq \alloc{\qstg}$ pointwise.
If $\altstg$ is non-exposed then we have
\begin{equation} \label{eq:non-exp-contin}
  \E \left[ \alloc{\qstg} v - c^\qstg \right] \leq
  \E \left[ \alloc{\altstg} v - c^\altstg \right] + 
  \E \left[ \left( \alloc{\qstg} - \alloc{\altstg} \right) \kappa \right].
\end{equation}
\end{lemma}
\begin{proof}
First note that 
\[
  \E \left[ \alloc{\qstg} v - c^\qstg \right] = 
  \E \left[ \alloc{\altstg} v - c^\altstg \right] +
  \E \left[ (\alloc{\qstg} - \alloc{\altstg}) v - (c^\qstg - c^\altstg) \right],
\]
so~\eqref{eq:non-exp-contin} is equivalent to
\begin{equation} \label{eq:nec.1}
  \E \left[ (\alloc{\qstg} - \alloc{\altstg}) v - (c^\qstg - c^\altstg) \right]
\leq
  \E \left[ \left( \alloc{\qstg} - \alloc{\altstg} \right) \kappa \right].
\end{equation}
A first observation is that the equation
\begin{equation} \label{eq:nec.1.1}
  (\alloc{\qstg} - \alloc{\altstg}) v - (c^\qstg - c^\altstg)
 =
  \left( \alloc{\qstg} - \alloc{\altstg} \right) \kappa
\end{equation}
holds at any sample point where $\qstg(\infty) = \altstg(\infty)$.
When $\qstg(\infty) = \altstg(\infty) < \infty$ this is 
because both sides are equal to 0. When $\qstg(\infty) = \altstg(\infty) = \infty$
we have $c^{\qstg} - c^{\altstg} = 0$, and the only case in which
$\alloc{\qstg} - \alloc{\altstg} \neq 0$ is if
$\alloc{\qstg} = 1, \alloc{\altstg}=0,$
but in that case $v = \kappa$ because 
$\altstg$ is non-exposed. (Recall that $v \geq \kappa$
by the definition of $\kappa$, and that the possibility 
$v > \kappa$ is precluded for a non-exposed policy $\altstg$
when $\altstg(\infty)=\infty$ and $\alloc{\altstg}=0$.)

The set $V$ of sample points where $\qstg(\infty) > \altstg(\infty)$
can be partitioned into sets $V_m \in \sigfld^{m,\infty}$ where
$\qstg(\infty) > \altstg(\infty) = m$. For $m=0,1,2,\ldots$ 
let $\Em$ denote the operator
\[ 
  \Em[f] = \int_{V_m} f \, d\mu,
\]
and note that $\int_V f \, d\mu = \sum_{m=0}^{\infty} \Em[f]$
whenever both sides are well-defined, for instance whenever
$\int_V f^+ \, d\mu < \infty$. We have already argued that
the integrands on
the left and right sides of~\eqref{eq:nec.1} are equal on 
the complement of $V$, so~\eqref{eq:nec.1} is equivalent to
\begin{equation} \label{eq:nec.2}
  \sum_{m=0}^{\infty} 
  \Em \left[ (\alloc{\qstg} - \alloc{\altstg}) v - (c^\qstg - c^\altstg) \right]
\leq
  \sum_{m=0}^{\infty}
  \Em \left[ \left( \alloc{\qstg} - \alloc{\altstg} \right) \kappa \right].
\end{equation}
We will prove~\eqref{eq:nec.2} by comparing the sums term-by-term.
Since $\alloc{\altstg}=0$ and $c^\altstg=c^m$ on $V_m$, the proof 
reduces to showing that
\begin{equation} \label{eq:nec.3}
  \Em \left[ \alloc{\qstg} v - c^\qstg  \right] 
\leq 
  \Em \left[ \alloc{\qstg} \kappa - c^m \right].
\end{equation}

For $n \geq 0$, define an inspection policy $\stg_n$ by
setting $\stg_n(\omega,\tau)$ to be the least $k> n$ such that
$\sigma^k \leq \kappa^n$ if there is any such $k$; otherwise
$\stg_n(\omega,\tau)=\infty$
and $\alloc{\stg_n}(\omega) = \alloc{\qstg}(\omega)$. 
Also define an inspection policy $\qstg'$ by specifying
that $\qstg'(\omega,\tau)=\qstg(\omega,\tau)$ for all
$\omega,\tau$ and that 
$$
  \alloc{\qstg'} = \begin{cases}
    \alloc{\qstg} & \mbox{if } v \geq \kappa^\infty \\
    0             & \mbox{if } v < \kappa^\infty
  \end{cases}.
$$
Consider the inspection
policy $\qstg_n$ defined by the pointwise minimum
$\qstg' \wedge \stg_n$. 
We claim that for $0 \leq m \leq n$, 
\begin{equation} \label{eq:nec.4}
  \Em \left[ \alloc{\qstg_n} v - c^{\qstg_n} \right]
\leq
  \Em \left[ \alloc{\qstg_n} \kappa^n - c^m \right].
\end{equation}
For each fixed $m$ the proof is by induction on $n$.
In the base case $n=m$, note that $\kappa^m = \sigma^m$
pointwise on $V_m$, because $\altstg(\infty)=m$ and
$\altstg$ is non-exposed. Thus, when $n=m$ the
inequality~\eqref{eq:nec.4} is equivalent to 
$\Em \left[ \alloc{\qstg_n} v - c^{\qstg_n} \right] 
\leq \Em \left[ \alloc{\qstg_n} \sigma^m - c^m \right]$,
which is valid because $\sigma^m$ discourages inspection
at stage $m$ and $V_m \in \sigfld^{m,\infty}$. For the 
induction step, let $U$ denote the set of sample
points where $\qstg_n(\infty) \neq \qstg_{n+1}(\infty)$.
Note that $U \in \sigfld^{n+1,\infty}$,
and that on $U$ the relations
$\sigma^{n+1} = \kappa^{n+1}$,
$\qstg_n(\infty)=n+1$, and
$\alloc{\qstg_n} = 0$ hold pointwise,
whereas on the complement of $U$ the relations $\kappa^{n+1} = \kappa^n$
and $\alloc{\qstg_{n+1}} = \alloc{\qstg_n}$
hold pointwise. 
These observations imply the following:
\begin{align}
\label{eq:gubp.2}
  \alloc{\qstg_{n+1}} v &= 
  \alloc{\qstg_n} v + \indic_U \alloc{\qstg_{n+1}} v \\
\label{eq:gubp.3}
  c^{\qstg_{n+1}} &=
  c^{\qstg_n} + \indic_U (c^{\qstg_{n+1}} - c^{n+1})  \\
\label{eq:gubp.4}
  \alloc{\qstg_{n+1}} \kappa^{n+1} &=
  \alloc{\qstg_n} \kappa^{n} +
  \indic_U \alloc{\qstg_{n+1}} \sigma^{n+1}.
\end{align}
We therefore have
\begin{align*}
  \Em \left[ \alloc{\qstg_{n+1}} v - c^{\qstg_{n+1}} \right] 
& =
  \Em \left[ \alloc{\qstg_n} v - c^{\qstg_n} \right] + 
  \Em \left[ \indic_U \left( \alloc{\qstg_{n+1}} v
              - c^{\qstg_{n+1}} + c^{n+1} \right) \right] \\
& \leq
  \Em \left[ \alloc{\qstg_n} \kappa^n - c^m \right] + 
  \Em \left[ \indic_U \alloc{\qstg_{n+1}} \sigma^{n+1} \right] \\
& = \Em \left[ \alloc{\qstg_{n+1}} \kappa^{n+1} - c^m \right],
\end{align*}
where the second line follows from the induction hypothesis 
and the fact that $\sigma^{n+1}$ discourages inspection at stage $n+1$.
This completes the inductive proof of~\eqref{eq:nec.4}.

Now, applying Lemma~\ref{lem:continuation} we find that 
\begin{align}
\nonumber
  \Em \left[ \alloc{\qstg'} v - c^{\qstg'} \right] 
& \leq
  \Em \left[ \alloc{\qstg_n} v - c^{\qstg_n} \right] +
  \Em \left[ \left( \alloc{\qstg'} - \alloc{\qstg_n} \right) 
                \sigma^{\qstg_n} \right] \\
%\nonumber
%& \leq 
%  \Em \left[ \alloc{\qstg_n} v - c^{\qstg_n} \right] +
%  \Em \left[ \left( \alloc{\qstg} - \alloc{\qstg_n} \right) 
%                 \kappa^n \right] \\
\nonumber
& \leq
  \Em \left[ \alloc{\qstg_n} \kappa^n - c^m \right] +
  \Em \left[ \left( \alloc{\qstg'} - \alloc{\qstg_n} \right) 
                 \kappa^n \right] \\
\label{eq:nec.5}
& = 
  \Em \left[ \alloc{\qstg'} \kappa^n - c^m \right].
\end{align}
where the second line follows from~\eqref{eq:nec.4}
and the definition
of $\qstg_n$: the relation
$\sigma^{\qstg_n} = \sigma^{\stg_n} \leq \kappa^n$
holds whenever $\qstg'(\infty) > \qstg_n(\infty)$.

The monotone convergence theorem implies
that 
$
  \lim_{n \to \infty} \Em \left[ \alloc{\qstg'} \kappa^n \right]
  = \Em \left[ \alloc{\qstg'} \kappa^\infty \right],
$
so 
the inequality 
\begin{equation} \label{eq:nec.6}
  \Em \left[ \alloc{\qstg'} v - c^{\qstg'} \right]
  \leq \Em \left[ \alloc{\qstg'} \kappa^\infty - c^m \right]
\end{equation} 
follows 
from~\eqref{eq:nec.5} in the limit as $n \to \infty$.
Now, recalling the definition of $\qstg'$, we see 
that $\alloc{\qstg'}=0$ whenever
$\kappa \neq \kappa^\infty$ and that 
$\alloc{\qstg} - \alloc{\qstg'} =0$ whenever
$\kappa \neq v$. 
In addition, $\qstg(\infty)=\qstg'(\infty)$ and
hence $c^\qstg = c^{\qstg'}$. Therefore,
\begin{align*}
  \Em \left[ \alloc{\qstg} v - c^{\qstg} \right] 
  & =
  \Em \left[ \left( \alloc{\qstg} - \alloc{\qstg'} \right) v \right]
  + \Em \left[ \alloc{\qstg'} v - c^{\qstg'} \right]  \\
  & \leq
  \Em \left[ \left( \alloc{\qstg} - \alloc{\qstg'} \right) v \right]
  + \Em \left[ \alloc{\qstg'} \kappa^\infty - c^{\qstg'} \right] \\
  & =
  \Em \left[ \left( \alloc{\qstg} - \alloc{\qstg'} \right) \kappa \right]
  + \Em \left[ \alloc{\qstg'} \kappa - c^{m} \right] \\
  & =
  \Em \left[ \alloc{\qstg} \kappa - c^{m} \right]
\end{align*}
which establishes~\eqref{eq:nec.3} and 
completes the proof of the lemma.
\end{proof}

The following lemma generalizes
Lemma 1 from main text Subsection 2.2 to the
setting of multi-stage inspection.

\begin{lemma} \label{lem:general-ubp-restated}
For any inspection policy $\stg$, we have
\begin{equation} \label{eq:gubp}
  \E \left[ \allocstg v - c^\stg \right] \leq \E \left[ \allocstg \kappa \right],
\end{equation}
with equality when $\stg$ is non-exposed.
\end{lemma}

The formal proof of the lemma appears below.
The intuition underlying the proof
can be summarized by appealing to the option-theoretic
notions introduced in main text Subsection 2.2. Suppose that a third
party (the ``insurer'') agrees to underwrite all of the bidder's
inspection costs in exchange for being granted a call option
with strike price $\sigma$ along with the right to dictate the 
bidder's inspection policy. The defining property of $\sigma^k$
is that if the bidder is at inspection stage $k$ and the insurer
holds an option with strike price $\sigma^k$, she is indifferent
between telling the bidder to stop immediately or to engage in
a policy that incorporates at least one more stage of 
inspection. Since the value of exercising the option is a
decreasing function of its strike price, this implies the 
following two observations:
\begin{enumerate}
\item If $\sigma^k < \sigma$, the insurer must require 
the bidder to stop inspecting.
\item If $\sigma^k > \sigma$, the insurer must require
the bidder to continue inspecting.
\end{enumerate}
Now consider an option with a renegotiable strike price
that evolves as follows. Initially the strike price is set at
$\sigma = \sigma^0$, but whenever the bidder reaches an 
inspection stage $k$ such that $\sigma^k < \sigma$,
the strike price is reduced to $\sigma^k$. Now there are
only two types of states: those in which 
$\sigma = \sigma^k$ (possibly because $\sigma$ was 
just reduced to $\sigma^k$) and those in which 
$\sigma < \sigma^k$. In the former type of
state, the insurer is indifferent between asking the bidder
to continue inspecting or to stop immediately; in the latter
type of state, she strictly prefers for the bidder to 
continue inspecting. In other words, the insurer is 
indifferent between all non-exposed policies, and strictly
prefers non-exposed policies to exposed ones. Since 
one example of a non-exposed policy is the trivial policy
which stops immediately at stage 0 and yields zero net 
payoff for the insurer, it must be the case that {\em every}
non-exposed policy yields zero expected net payoff for the
insurer, and all other policies yield negative expected net payoff
for the insurer. Given the way we have defined the dynamics
of the option's renegotiable strike price, it is always the case 
that if the bidder  completes all stages of inspection and 
acquires the item, she pays $v - \kappa$ to the insurer.
Thus the insurer's net payoff is $\allocstg (v - \kappa) - c^\stg$.
We have argued that the expected value of this quantity is 
never positive, and that it equals zero when $\stg$ is non-exposed,
exactly as asserted in the lemma.

Having presented this intuition, we now present the formal
proof.
\begin{proof}[Proof of \Cref{lem:general-ubp-restated}]
The inequality~\eqref{eq:gubp} is the special case
of Lemma~\ref{lem:non-exp-contin} in which $\qstg=\stg$
and $\altstg$ is the trivial inspection policy given by
$\altstg(\omega,\tau)=0$ for all $(\omega,\tau)$. So, we
only need to prove that 
$\E \left[ \allocstg v - c^\stg \right] = \E \left[ \allocstg \kappa \right]$
when $\stg$ is non-exposed. To do so, it will be useful
to first construct, for any given $k \in \mathbb{N}$, 
a non-exposed prompt inspection policy $\qstg$ such that $\qstg(\infty) > k$
pointwise, and
\begin{equation} \label{eq:gubp.01}
  \E \left[ \alloc{\qstg} v - c^\qstg \right]
\geq
  \E \left[  \alloc{\qstg} \kappa  \right].
\end{equation}
The construction will proceed by induction on $k$. For $k=0$,
observe that the policy $\qstg_0$ defined in Lemma~\ref{lem:indifferent} 
is non-exposed and prompt. Setting $\qstg = \qstg_0$ and applying 
Lemma~\ref{lem:indifferent} we obtain
\[
  \E \left[ \alloc{\qstg} v - c^{\qstg} \right] \geq
  \E \left[ \alloc{\qstg} \sigma^0 - c^0 \right] \geq
  \E \left[ \alloc{\qstg} \kappa \right]
\]
since $\sigma^0 \geq \kappa$ and $c^0 = 0$.
This completes the base case of the induction.

For $k>0$ assume that we already have a
non-exposed prompt inspection policy $\qstg'$ such
that $\qstg'(\infty) > k-1$ pointwise, and 
\begin{equation} \label{eq:gubp.02}
  \E \left[ \alloc{\qstg'} v - c^{\qstg'} \right] 
  \geq \E \left[ \alloc{\qstg'} \kappa \right].
\end{equation}
Let $V$ be the set of all sample points 
where $\qstg'(\infty) = k$ and note that
$V \in \sigfld^{k,\infty}$. Recalling the policy
$\qstg_k$ defined in Lemma~\ref{lem:indifferent},
we define $\qstg$ to be equal to $\qstg_k$ on $V$
and to $\qstg'$ on $\Omega \setminus V$. By
construction $\qstg$ is prompt and 
$\qstg(\infty) > k$ pointwise. 
For any $\ell > k$ we have
\[
  \qstg(\infty) > \ell \Longleftrightarrow
  \left( \qstg'(\infty) = k \mbox{ and } 
          \qstg_k(\infty) > \ell \right) \mbox{ or }
  \qstg'(\infty) > \ell,
\]
which shows that the event $\qstg(\infty) > \ell$ 
belongs to $\sigfld^{\ell,0}$, confirming that 
$\qstg$ is a valid inspection policy. Our induction
hypothesis that $\qstg'$ is non-exposed implies
that at all sample points in the complement of $V$,
$\qstg$ satisfies the properties that define a 
non-exposed policy. At sample points that belong
to $V$, the facts that $\qstg'$ is non-exposed and
that $\qstg'(\infty) = k$ together imply 
$\sigma^k = \kappa^k$. This means that $\qstg$ 
(which mimics the policy $\qstg_k$ on $V$) stops at the
first $\ell > k$ such that $\sigma^\ell \leq \sigma^k$
if such an $\ell$ exists, and otherwise it completes all
inspection stages and acquires the item if and only
if $v > \sigma^k$. Recalling that $\sigma^k=\kappa^k$,
we see that the former case implies that $\qstg$
stops at a stage $\ell < \infty$ such that 
$\sigma^\ell = \kappa^\ell$, and the latter case
implies that $\qstg(\infty) = \infty$ and that
$\alloc{\qstg}=1$ if and only if 
$v > \sigma^k = \kappa$. This completes the 
verification that $\qstg$ is non-exposed.

Applying Lemma~\ref{lem:indifferent} we find that
\begin{equation} \label{eq:gubp.03}
  \E \left[ \indic_V \left( \alloc{\qstg} v - c^{\qstg} \right) \right]
  =
  \E \left[ \indic_V \left( \alloc{\qstg} \sigma^k - c^k \right) \right]
  = 
  \E \left[ \indic_V \left( \alloc{\qstg} \kappa - c^k \right) \right],
\end{equation}
where the second equation uses the fact that at every sample point in $V$ 
where $\alloc{\qstg} = 1$, the relation $\sigma^k = \kappa$ holds.
Summing~\eqref{eq:gubp.02} and~\eqref{eq:gubp.03}, and
rearranging terms, we obtain
\begin{equation} \label{eq:gubp.04}
  \E \left[ \left( \alloc{\qstg'} + \indic_V \alloc{\qstg} \right) v -
                   c^{\qstg'} - \indic_V c^{\qstg} + \indic_V c^k \right] \geq
  \E \left[ \left( \alloc{\qstg'} + \indic_V \alloc{\qstg} \right) \kappa \right].
\end{equation}
Recall that on the set $V$, we have $\qstg'(\infty) = k$, hence
the relations
$c^{\qstg'} = c^k$ and $\alloc{\qstg'} = 0$
hold pointwise on $V$. 
These observations, together with the fact that $\qstg = \qstg'$ 
on $\Omega \setminus V$, imply the relations
\begin{align*}
  \alloc{\qstg} &=  
  \indic_{\Omega \setminus V} \alloc{\qstg'} + \indic_V \alloc{\qstg} = 
  \alloc{\qstg'} + \indic_V \alloc{\qstg} 
\\
  c^{\qstg} &= 
  \indic_{\Omega \setminus V} c^{\qstg'} + \indic_V c^{\qstg} =
  c^{\qstg'} - \indic_V c^k + \indic_V c^{\qstg} 
\end{align*}
hence~\eqref{eq:gubp.04} simplifies to
% \begin{equation} \label{eq:gubp.05}
$  \E \left[ \alloc{\qstg} v - c^\qstg \right] \geq 
  \E \left[ \alloc{\qstg} \kappa \right], $
% \end{equation}  
as desired. This completes the induction step, and
establishes the existence of the claimed inspection
policy $\qstg$ for every $k \in \mathbb{N}$.

The monotone convergence theorem implies 
that $\E \left[ c^\infty - c^k \right] \to 0$
as $k \to \infty$, so for any given $\eps>0$ we 
may choose $k$ large enough
that $\E \left[ c^\infty - c^k \right] < \eps$.
Then we may choose a non-exposed prompt inspection policy 
$\qstg$ such that $\qstg(\infty) > k$ pointwise, and
$\E \left[ \alloc{\qstg} v - c^{\qstg} \right] \geq
  \E \left[ \alloc{\qstg} \kappa \right].
$
For the inspection
policy defined by the pointwise maximum
$\qstg \vee \stg$, we have
\begin{align}
\nonumber
  \E \left[ \alloc{\qstg \vee \stg} v - c^{\qstg \vee \stg} \right] &=
  \E \left[ \alloc{\qstg} v - c^{\qstg} \right] +
  \E \left[ \left( \alloc{\qstg \vee \stg} - \alloc{\qstg} \right) v \right] -
  \E \left[ c^{\qstg \vee \stg} - c^{\qstg} \right] \\
  & >
  \E \left[ \alloc{\qstg} \kappa \right] +
  \E \left[  \left( \alloc{\qstg \vee \stg} - \alloc{\qstg} \right) \kappa \right] -
  \eps.
\label{eq:gubp.06}
\end{align}
In deriving the last line, we made use of~\eqref{eq:gubp.01}
together with fact that 
%% \begin{enumerate}
%% \item   $\stg$ is non-exposed, so $v \geq \kappa$ at
%% any sample point where $\allocstg=1$. In particular,
%% $v \geq \kappa$ at any sample point where
%% $\alloc{\qstg \vee \stg} > \alloc{\qstg}$.
%% \item  
$k < \qstg(\infty) \leq (\qstg \vee \stg)(\infty) \leq \infty$,
so $\E \left[ c^{\qstg \vee \stg} - c^{\qstg} \right] \leq
      \E \left[ c^\infty - c^k \right] < \eps.$
%% \end{enumerate}

Now, applying Lemma~\ref{lem:non-exp-contin}
to the inspection policies $\stg$ and $\qstg \vee \stg$,
we obtain 
\begin{align*}
  \E \left[ \allocstg v - c^\stg \right] 
& \geq
  \E \left[ \alloc{\qstg \vee \stg} v - c^{\qstg \vee \stg} \right] -
  \E \left[ \left( \alloc{\qstg \vee \stg} - \allocstg \right)
              \kappa \right] \\
& >
  \E \left[ \alloc{\qstg \vee \stg} \kappa \right] - \eps -
  \E \left[ \left( \alloc{\qstg \vee \stg} - \allocstg \right)
              \kappa \right] \\
& = 
  \E \left[ \allocstg \kappa \right] - \eps.
\end{align*}
As $\eps>0$ was arbitrarily small, we conclude
that $\E \left[ \allocstg v - c^\stg \right] \geq
\E \left[ \allocstg \kappa \right]$. The reverse
inequality was already established in the first
paragraph of this proof, so the proof is complete.
\end{proof}

We conclude this section with a technical lemma
about the generalized covered call value which 
will be useful in the sequel, when we analyze the
first-best procedure and the equilibria of Dutch
auctions.

\begin{lemma} \label{lem:v-k}
At every sample point where $v > \kappa \geq 0$,
except possibly a measure-zero set of exceptions, 
there exists some $k < \infty$ such that 
$\kappa = \sigma^k < \inf_{\ell > k} \{\sigma^\ell\}$.
\end{lemma}
\begin{proof}
For any $k < \infty$ and $\eps > 0$ let $U_{k,\eps}$
denote the set of sample points where 
$\E [ v^+ \given \sigfld^{k,\infty} ] > \sigma^k + \eps \geq \eps$
and $\E [ c^\infty \given \sigfld^{k,\infty} ] < c^k + \eps $.
We claim that $U_{k,\eps}$ has measure zero. Indeed, let $\stg$ denote
the inspection policy that always completes all stages of 
inspection, and then acquires the item if $v \geq 0$. 
The strict inequality 
\begin{equation} \label{eq:v-k.1}
  \E [ v^+ - c^\infty \given \sigfld^{k,\infty} ] > \sigma^k - c^k
\end{equation}
holds pointwise on $U_{k,\eps}$. Since $\allocstg=1$ if and only if 
$v \geq 0$, the equation $\allocstg v = v^+$ holds pointwise.
Since $\sigma^k \geq 0$ at every point of $U_{k,\eps}$,
the inequality $\allocstg \sigma^k \leq \sigma^k$ holds 
pointwise on $U_{k,\eps}$. Combining these relations 
with~\eqref{eq:v-k.1} we find that 
\begin{equation} \label{eq:v-k.2}
  \E [ \allocstg v - c^\stg \given \sigfld^{k,\infty} ] >
  \allocstg \sigma^k - c^k
\end{equation}
holds at every point of $U_{k,\eps}$. On the other hand, 
since $\sigma^k$ discourages inspection, the
opposite inequality 
$
  \E [ \allocstg v - c^\stg \given \sigfld^{k,\infty} ] \leq
  \allocstg \sigma^k - c^k
$ 
holds pointwise almost everywhere. The only way 
to reconcile these two statements is to conclude that
$U_{k,\eps}$ has measure zero.

Let $V$ denote the set of all sample points
at which 
%%  $v^+ = \lim_{k \to \infty} \E [ v^+ \given \sigfld^{k,\infty} ]$
%% and $c^\infty = \lim_{k \to \infty} \E [ c^\infty \given \sigfld^{k,\infty}]$.
%% By the martingale convergence theorem, $V$ has measure 1. 
%% Since $c^\infty = \lim_{k \to \infty} c^k$, we find that 
%% for every $\eps>0$, 
the relations 
\begin{equation} \label{eq:v-k.3}
  \E [ v^+ \given \sigfld^{k,\infty} ] > v^+ - \eps
  \qquad \mbox{and} \qquad
  \E [ c^\infty \given \sigfld^{k,\infty} ] < c^k + \eps
\end{equation}
are violated for only finitely many values of $k$. 
By the martingale convergence theorem, we have
$\lim_{k \to \infty} \E [ v^+ \given \sigfld^{k,\infty} ] = v^+$
and 
$\lim_{k \to \infty} \E [ c^\infty \given \sigfld^{k,\infty} ] = c^\infty
= \lim_{k \to \infty} c^k$
almost everywhere, and this implies that 
$V_{\eps}$ has measure 1. 

Now consider the set 
$$
  W_{\eps} = V_{\eps} \setminus 
            \left( \bigcup_{k=0}^{\infty} U_{k,\eps} \right),
$$
which also has measure 1. At any point in 
$W_{\eps}$ there is some $k_0 < \infty$ such
that for all $k \geq k_0$, both of the relations
in line~\eqref{eq:v-k.3} hold. However, since the
point does not belong to $U_{k,\eps}$, it must be
the case that 
$\E [ v^+ \given \sigfld^{k,\infty} ] > \sigma^k + \eps \geq \eps$
does not hold. In other words, either
$\sigma^k < 0$ or 
$\E [ v^+ \given \sigfld^{k,\infty} ] \leq \sigma^k + \eps$.
If there is any $k$ such that $\sigma^k < 0$, then $\kappa <0$.
Otherwise, for all $k \geq k_0$, 
$\E [ v^+ \given \sigfld^{k,\infty} ] \leq \sigma^k + \eps$.
Combining this with~\eqref{eq:v-k.3}, we have that at any point
of $W_\eps$ where $\kappa \geq 0$,
\begin{equation} \label{eq:v-k.4}
  \forall k \geq k_0 \quad 
  v^+ < \E [ v^+ \given \sigfld^{k,\infty} ] + \eps \leq \sigma^k + 2 \eps
\end{equation}
and therefore
\begin{equation} \label{eq:v-k.5}
  v^+ < \left( \liminf_{k \to \infty} \sigma^k \right) + 2 \eps.
\end{equation}
Since each of the sets
$W_{\eps}$ has measure 1 for $\eps=1,\frac12,\frac13,\ldots$, their 
intersection $W = \bigcap_{n=1}^{\infty} W_{1/n}$ has
measure 1, and the relation 
$v^+ \leq \liminf_{k \to \infty} \sigma^k$ 
holds at every point of $W$ where $\kappa \geq 0$. 
Recalling that $\kappa = \min\{v, \kappa^\infty\}$
and that $\kappa^\infty = \lim_{k \to \infty} \kappa^k = 
\inf_{k \in \mathbb{N}}  \sigma^k $,
we may conclude that at any point of $W$ where $v > \kappa \geq 0$,
we must have 
\begin{equation} \label{eq:v-k.6}
  \liminf_{k \to \infty} \sigma^k \geq v^+ = v > \kappa = \inf_{k \in
  \mathbb{N}} \sigma^k,
\end{equation}
from which we may conclude that the infimum on the right side
of~\eqref{eq:v-k.6} is achieved at a finite $k$, and that if $k$
is the greatest integer satisfying $\kappa = \sigma^k$ then
$\kappa < \inf_{\ell > k} \{\sigma^\ell\}$.
\end{proof}

\subsection{First-best Procedure}
\label{sec:multistage-firstbest}

\newcommand{\tk}{{\tilde{\kappa}}}

Weitzman's optimal search procedure generalizes,
in the setting of multi-stage inspection, 
to the following {\em descending-priority procedure}.
In describing the procedure, we assume that steps of
the procedure are numbered by (possibly transfinite,
but at most countable)
ordinal numbers, so that it is meaningful to refer 
to a step in which one or more bidders have already 
completed all of the (countably many) inspection 
stages. When we specialize to a setting 
in which bidders are guaranteed to learn all 
information about their value after performing a
finite number of stages of inspection, then we
may assume without loss of generality 
that their inspection stage
advances to $k=\infty$ as soon as they have learned
all information about their value, and then we can
assume that steps of the procedure are indexed by 
natural numbers rather than transfinite ordinals.

The descending-priority procedure assigns to each bidder
$i$ a {\em priority} defined according to her
current inspection stage $k(i)$ as follows: 
if $k(i) < \infty$ then the priority is set 
to the generalized strike price $\sigma_i^{k(i)}$; 
if $k(i)=\infty$ then the priority is set to $v_i$. 
In every step of the procedure, the bidder $i$ with
highest priority is selected, breaking ties arbitrarily.
If this priority is negative, we terminate the procedure
without allocating the item. 
If $k(i) < \infty$ then we perform the next inspection
stage for bidder $i$, increment $k(i)$, and recompute
bidder $i$'s priority using this updated value of $k(i)$. 
If $k(i)=\infty$ then we allocate the item to bidder $i$
and terminate.

To analyze the descending-priority procedure, it will
be useful to define a state variable $\tk_i$ for each
bidder $i$. In any step of the procedure $\tk_i$ is 
set equal to $\kappa_i^{k(i)}$ if $k(i)<\infty$,
and otherwise $\tk_i = \kappa_i$. 
In other words, $\tk_i$ is equal to the infimum
of the priority values assigned to bidder $i$
at past and present steps of the procedure;
accordingly, we will refer to it as her 
{\em min-priority}. 
%% Let us also 
%% define the {\em active bidder} in any step of the
%% procedure to be the one who is selected either to 
%% perform an additional stage of inspection or to 
%% receive the item; all
%% other bidders are called {\em dormant}.

The following property of the descending-priority
procedure will be instrumental in its analysis. 

\begin{lemma} \label{lem:descending-priority}
In any step of the descending-priority procedure,
the priority and min-priority of any bidder are
equal unless she is the unique bidder with
maximum priority; in the latter case, her min-priority
is greater than or equal to that of all other bidders.
\end{lemma}
\begin{proof}
The proof is by (transfinite) induction on the 
steps of the procedure. Initially, when $k(i)=0$
for each bidder, each bidder's priority equals
her min-priority. 
the priority of each bidder $i$ is equal to 
$\sigma_i^0$, which in turn equals $\tk_i$.

If the current step of the procedure is a 
successor ordinal, that let $i$ denote the 
bidder who performed a stage of inspection 
at the end of the preceding step. Any bidder
$i' \neq i$ could not have been the unique 
maximum-priority bidder (else the procedure
would have selected $i'$ rather than $i$ to
perform a stage of inspection) so by the 
induction hypothesis, the priority and 
min-priority of $i'$ were equal in the preceding
step. They remain equal in the current step
because $k(i')$ remains unchanged. As for 
bidder $i$, in the preceding step she had
the maximum priority; by the induction
hypothesis this means that she also had the
maximum min-priority. In the current step
there are two cases: if her new priority is
equal to her min-priority, then the lemma's
conclusion is clearly satisfied. If her new
priority is strictly greater than her min-priority,
then her min-priority must have retained its
value from the preceding step; this is greater
than or equal to the min-priority of all other bidders
because they, too, have retained their min-priority
from the preceding step.

If the current step is a limit ordinal, then
there are two types of bidders to consider:
{\em dormant bidders}, for whom the 
inspection counter $k(i)$ reached
its current value at a strictly earlier step
of the procedure, and 
{\em active bidders}, for whom $k(i)=\infty$ at the current step
but $k(i)<\infty$ at all strictly earlier steps.

For any dormant bidder $i$, the induction hypothesis
guarantees that the priority and min-priority are
equal: consider the earliest step at which $k(i)$
reached its current value. We know that $i$ could
not have been the unique bidder with maximum
priority in that step, as otherwise she would have
been selected to perform an additional inspection
stage. Thus, at that earlier step the priority and
min-priority of $i$ were equal, and neither of them
has subsequently been updated. 

On the other hand,
for an active bidder $i$, since $k(i)=\infty$,
her priority is equal to $v_i$ whereas her min-priority
is equal to $\kappa_i$. If $v_i = \kappa_i$ then
the lemma's conclusion holds for bidder $i$, so 
assume henceforth that $v_i > \kappa_i$. 
The descending-priority 
procedure never selects a bidder with negative
priority to perform an inspection stage, so 
$\kappa_i \geq 0$. Now applying \autoref{lem:v-k}
we find that there is some $k < \infty$ such 
that $\kappa_i = \sigma_i^k <
\inf_{\ell > k} \{ \sigma_i^\ell \}$. 
Let $t$ denote the step of the 
descending-priority procedure in which 
bidder $i$ was selected to advance from inspection
stage $k$ to $k+1$. During step $t$ her priority
$\sigma_i^k$ was greater than or equal to the
priority of all other bidders, and 
since $\sigma_i^k <
\inf_{\ell > k} \{ \sigma_i^\ell \}$
it follows that her priority was {\em strictly}
greater than the priority of all other bidders
in every step from $t+1$ until she reached inspection
stage $\infty$. Thus, in the current step
bidder $i$ is the only active bidder. 
Furthermore, the min-priorities
of all bidders including $i$ are
the same as they were in step $t$,
which implies that the min-priority
of bidder $i$ is greater than or equal
to that of all other bidders, as claimed.
\end{proof}

\begin{theorem} \label{thm:multistage-firstbest}
The expected welfare achieved by the
descending-priority procedure is equal
to $\E [ \max_i \kappa_i^+ ]$.
No other procedure can attain a higher
expected welfare.
\end{theorem}
\begin{proof}
By \autoref{lem:general-ubp-restated}, for any
procedure the expected net utility of bidder $i$
is at most $\E [ \mathbb{A}_i \kappa_i ]$, where
$\mathbb{A}_i$ is the indicator of the event that
the item is awarded to $i$. Summing over bidders,
the expected welfare of any procedure is at most
$\E \left[ \sum_i \mathbb{A}_i \kappa_i \right]$,
which in turn is bounded above by $\E \left[ \max_i (\kappa_i)^+ \right]$.

To see that the descending-priority procedure 
attains this upper bound, we will show that it 
induces a non-exposed inspection policy for each
bidder, and that it always awards the item to a
bidder $i$ in $\arg\max_i \{ \kappa_i \}$ 
unless $\kappa_i < 0$, in which case the item is not allocated.
From these two facts, \autoref{lem:general-ubp-restated}
implies that the expected welfare of the
descending-priority procedure is equal to
$\E \left[ \max_i (\kappa_i)^+ \right]$, exactly 
as in the proof of Theorem 1. 

Consider any step
of the descending-priority procedure and suppose
that either
\begin{enumerate}
\item $k(i) < \infty$ and 
$\sigma_i^{k(i)} > \kappa_i^{k(i)}$, or
\item $k(i) = \infty$ and
$v_i > \kappa_i$.
\end{enumerate}
In both case the priority of bidder $i$ is strictly 
greater than her min-priority. By 
\autoref{lem:v-k} this means that $i$
is the unique bidder with maximum priority.
In the first case this means she will be selected to perform
another stage of inspection. In the second case it means
the item will be awarded to her. Thus, the
descending-priority procedure never stops
inspecting bidder $i$ when she is at a stage
with $\sigma_i^{k(i)} > \kappa_i^{k(i)}$, and it
never fails to award the item to a bidder 
who has completed all inspection stages and
found that $v_i > \kappa_i$. This confirms
that the policy is non-exposed. 

To conclude we must show that the procedure 
always awards the item to a
bidder $i$ in $\arg\max_i \{ \kappa_i \}$ 
unless $\kappa_i < 0$, in which case the item is not allocated.
If the procedure terminates by awarding the item to $i$, then
$k(i) = \infty$, so the min-priority of bidder
$i$ is $\kappa_i$. By \autoref{lem:v-k} this is 
greater than or equal to the min-priority of
every other bidder $i'$, which in turn is an upper
bound on $\kappa_{i'}$, so $i$ is among the bidders
with maximum covered call value as claimed. Furthermore,
it cannot be the case that $\sigma_i^k < 0$ for any
$k \in \mathbb{N}$ --- as otherwise the procedure
would never have advanced bidder $i$ beyond inspection
stage $k$ --- nor can it be the case that $v_i < 0$,
since the procedure never awards the item to a bidder
with negative value. Accordingly, $\kappa_i \geq 0$ as
claimed. The remaining case to consider is that the
procedure terminates without awarding the item to any 
bidder. This can only happen if the priority of every
bidder is strictly negative. Since each bidder's priority
is an upper bound on her min-priority, which in turn is an
upper bound on her covered call value, this means that all
covered call values are strictly negative as claimed.
\end{proof}

\subsection{Application to Dutch Auction Analysis}
\label{sec:multistage-dutch}

We now sketch
how our main results concerning the analysis of the
Dutch auction (main text Lemma 2, Theorem 2, and Corollary 2)
extend to the setting of multi-stage inspection, 
as claimed in main text Subsection 6.1.
Since we are 
considering a Dutch auction for a single item, we 
may drop the subscript $j$ from our $\sigma$-fields
and denote them by $\sigfld_{i}^{k,\tau}$. 

In the Dutch auction, we will assume that 
the descending price is given by a continuous
non-increasing function $t$, with $t(\tau)$ denoting
the price when the clock time is $\tau$. The
$\sigma$-field $\sigfld_i^{k,\tau}$ is generated
by bidder $i$'s private type $\theta_i$, the outcomes
of her first $k$ stages of inspection, and a random
variable $s(\tau)$ describing the state of the auction 
at clock time $\tau$; this state variable $s(\tau)$
takes the value $(i',\tau')$ if some
bidder $i'$ was allocated the item at time $\tau' < \tau$
and otherwise it takes a null value $\perp$. 

A strategy
for bidder $i$ is given by an inspection policy
$\stg_i$, subject to the feasibility constraint that
the bidder cannot attempt to acquire the item if it
has already been allocated to another bidder:
$\alloc{\stg_i}(\omega,\tau)=0$ unless $s(\tau)=\perp$.
We will use $\awardi(\omega,\tau)$
to denote the indicator random variable of the event that
the item is allocated to bidder $i$ at time $\tau$ or 
earlier. The only case in which $\awardi(\omega,\tau)$
may differ from $\alloc{\stg_i}(\omega,\tau)$ is when two
or more bidders attempt to acquire the item at time $\tau$,
in which case the winner is decided by random tie-breaking. 
The price charged to bidder $i$ is 
\[
  t_i(\omega) = \begin{cases}
    0 & \mbox{if } \awardi(\omega,\tau)=0 \;\; \forall \tau \\
    \sup \{ t(\tau) \mid \awardi(\omega,\tau)=1 \}
      & \mbox{otherwise.}
  \end{cases}
\]

The covered counterpart of bidder $i$ is a bidder (numbered $i+n$,
as in main text Section 4) whose inspection costs are identically zero and
whose value is equal to the random variable $\kappa_i$, the
generalized covered call value of bidder $i$. A Dutch auction
strategy for bidder $i+n$ is an inspection policy $\stg_{i+n}$
adapted to the filtration $\{ \sigfld_i^{k,\tau} \}$, subject to the
same feasibility constraint articulated above: bidder $i+n$ cannot attempt to 
acquire the item if it has already been awarded to another bidder.
By allowing the strategy $\stg_{i+n}$ to be adapted to the 
filtration $\{ \sigfld_i^{k,\tau} \}$ we are implicitly assuming
that bidder $i+n$ is granted knowledge of all of bidder $i$'s
inspection costs and outcomes, as well as her value for acquiring 
the item. The justification for this assumption is the same as 
the justification given in the first paragraph of the proof
of Theorem 2 and its accompanying footnote.

As in main text Section 4, we will use $\mathcal P$ to denote a
(possibly correlated) distribution of types for
$n$ bidders, and $\ccpt{\mathcal P}$
denotes the distribution of their covered counterparts.

\begin{lemma}[Lemma 2 restated]
 \label{lem:general-equivwelfare}
The highest expected welfare achievable by any procedure when types are distributed
according to $\mathcal P$ is equal to the highest expected welfare achievable when bidders are replaced with their
covered counterparts, who face no inspection costs and have types distributed
according to $\ccpt{\mathcal P}$.
\end{lemma}

\begin{proof}
The highest expected welfare achievable
when agents' types are jointly distributed according to
$\ccpt{\mathcal P}$ is 
$\E \left[ \max_i (\kappa_i)^+ \right]$,
achieved for example by running a second-price auction. 
We have seen in the proof of~\autoref{thm:multistage-firstbest}
that this is also the expected welfare of the first-best
procedure when agents' types are jointly distributed
according to ${\mathcal P}$. A mechanism which implements
this first-best welfare is the dynamic VCG mechanism in which
every bidder reveals all of their private information at the
outset of the auction and subsequently updates their private
information each time they perform a new stage of inspection;
the mechanism directs each bidder to carry out their share of
the sequence of inspections dictated by
the descending-priority procedure, allocates the item as
dictated by that procedure, and charges the winner a price
whose expected value is equal to the negative externality 
that the winner imposes on the other bidders by participating
in the auction.
\end{proof}

Rather than repeating the full proof of 
Theorem 2 --- the correspondence between Dutch auction
equilibria in our model and equilibria when bidders are
replaced with their covered counterparts --- in the 
interest of space we merely sketch how to modify the 
proof to the setting of multi-stage inspection.
First, in order to state the theorem, 
the definition of functional equivalence
must be generalized to the setting of multi-stage
inspection.

\begin{definition} 
Two auction outcomes are said to be {\em functionally equivalent} if they award
the item to the same bidder at the same price, and each bidder bears the same inspection cost in both auction outcomes.
Two strategies for a bidder are functionally equivalent (against a given profile
of opponents' strategies) if they always result in
functionally equivalent outcomes.
Two equilibria are functionally equivalent if they result in functionally
equivalent outcomes for every profile of types.
\end{definition}

\begin{thm}
There is a mapping from equilibria of the Dutch auction in our model with types jointly distributed
according to $\mathcal P$ to equilibria of the Dutch auction with the covered counterpart
distribution $\ccpt{\mathcal P}$ where bidders know their values without inspection.
This mapping preserves bidder expected utility and auctioneer expected revenue, and it
induces a bijection on functional equivalence classes of equilibria.
\end{thm}

Recall that the crux of the proof of Theorem 2 was a pair of
mappings, $\mu$ and $\lambda$, from strategies
of bidder $i$ to normalized strategies of her covered counterpart,
bidder $i+n$, and vice-versa. As before, 
we can define the strategy $\mu(\stg)$ for bidder $i+n$ to simply
be a simulation of inspection policy $\stg$ as executed by bidder 
$i$. In fact, due to our 
convention that a strategy for bidder $i+n$ is an 
inspection policy adapted to bidder $i$'s filtration
$\{ \sigfld_i^{k,\tau} \}$, we can literally define
$\mu$ to be the identity function $\mu(\stg) = \stg$:
for bidder $i+n$ there is no distinction between simulating $\stg$ and
executing $\stg$. 

The definition of normalized strategies for a covered counterpart
bidder generalizes in the obvious way: 
a strategy $\stg$ of bidder $i+n$ is normalized
if
\begin{enumerate}
\item it always performs all stages of inspection at the earliest 
possible moment --- i.e., $\stg(\omega,0) = \infty$ for
all $\omega$;
\item the price at which it attempts to acquire
the item (if it is not yet claimed by another bidder)
is given by a function $b(v_i)$ whose value 
depends only on $v_i$;
not on the outcomes of intermediate stages of inspection;
\item the function $b(v_i)$ is non-decreasing.
\end{enumerate}
If $b$ is a normalized
strategy of bidder $i+n$ then we can define $\lambda(b)$ to be 
the strategy for bidder $i$ which operates as follows: if the
current inspection stage is $k<\infty$, then $\lambda(b)$ 
advances to the next inspection stage when the price is
$b(\kappa_i^k)$; if the current inspection stage
is $k=\infty$ then $\lambda(b)$ claims the item when the price
is $b(\kappa_i)$. This policy 
$\lambda(b)$ is non-exposed: if $\sigma_i^k > \kappa_i^k$
then it immediately advances to the next inspection stage
without waiting for the price to descend further, and if
$k=\infty$ and $v_i > \kappa_i$ then it acquires the item 
immediately. 

The proof of Theorem 2 depended on four claims. The proofs
of Claims 1, 3, and 4 are unchanged except for changes in 
notation and terminology; for instance, ``almost surely
exercises in the money'' changes to ``is non-exposed'' 
when we generalize from one to many stages of inspection.
In Claim 2, the proof that $b_{i+n}$ and 
$\mu(\lambda(b_{i+n}))$ are functionally equivalent
is unchanged but the proof that $b_i$ and 
$\lambda(\mu(b_i))$ requires a non-trivial 
generalization. In place of the single threshold 
$b_i^{\mathrm{insp}}(\theta_i)$ appearing in the 
original proof of Claim 2, a strategy $\stg_i$ of 
bidder $i$ defines a sequence of random variables 
$b_i^k$, each denoting the price at which $\stg_i$ 
will perform the $k^\mathrm{th}$ stage of inspection
if no other bidder has yet claimed the item. 
More formally, $b_i^k = t(\tau_i^k)$, where 
$$
  \tau_i^k = \inf \{ \tau \mid \stg_i(\omega,\tau) \geq k \}
$$
Restricted to the set of sample points where $s(\tau_i^k) = \perp$,
the value of $\tau_i^k$ (and hence also $b_i^k$) 
must be $\sigfld^{k-1,0}$-measurable, 
due to our assumption
that the event $\stg_i(\omega,\tau) \geq k$ is measurable with respect to the
$\sigma$-field generated by $i$'s private type, her first
$k-1$ inspection outcomes, and $s(\tau)$. Now, as in the
proof of Claim 2, we reason that bidder $(i+n)$'s normalized strategy
$\mu(\stg_i)$ is specified by a bid function $b$ that 
satisfies $b(\kappa_i^k) = b_i^k$, and hence that
$\lambda(\mu(\stg_i))$ advances to the $k^{\mathrm{th}}$ 
inspection stage precisely at price $b_i^k$, unless the
item has already been allocated, just like strategy
$\stg_i$. Similarly, both strategies acquire the item
at price $b(\kappa_i)$, unless it has already been
allocated. This completes the proof that they are
functionally equivalent. As in the original proof of
Claim 2, the functional equivalence of a normalized
strategy $\stg_{i+n}$ and $\mu(\lambda(\stg_{i+n}))$ 
is an easy consequence of the fact that bidder $i+n$
has no inspection costs.

\end{document}